%% file: vldb-mrjoin.tex
\documentclass{vldb}
\pdfoutput=1

\vldbTitle{Improving Distributed Similarity Join in Metric Space with Error-bounded Sampling}
\vldbAuthors{Jiacheng Wu, Yong Zhang, Jin Wang, Chunbin Lin, Yingjia Fu, Chunxiao Xing}
\vldbDOI{https://doi.org/10.14778/xxxxxxx.xxxxxxx}
\vldbVolume{12}
\vldbNumber{xxx}
\vldbYear{2019}

\usepackage{balance} 
\usepackage{times,amsmath,epsfig}
\usepackage{graphicx}
\usepackage{balance}
\usepackage{booktabs}
\usepackage{mathtools}
\usepackage{multirow}
\usepackage{paralist}
\usepackage{amssymb}
\usepackage{pifont}
\usepackage{epstopdf}
\usepackage{xspace,colortbl,subfigure,multirow}
\usepackage{xcolor}
\usepackage{bm}
\usepackage{bbold}
\usepackage{amsmath}
\usepackage{accents}
\usepackage{statex}
\usepackage{comment}

\makeatletter
\newif\if@restonecol
\def\widebar{\accentset{{\cc@style\underline{\mskip14mu}}}}
\def\Widebar{\accentset{{\cc@style\underline{\mskip8mu}}}}
\makeatother

\usepackage[lined,boxed,vlined,ruled]{algorithm2e}

\begin{document}

\title{Improving Distributed Similarity Join in \\Metric Space with Error-bounded Sampling}\vspace{-2em}

\author{Jiacheng Wu$^\dag$~~~~~~Yong Zhang$^\dag$~~~~~~Jin Wang$^\sharp$~~~~~~Chunbin Lin$^\ddag$~~~~~~Yingjia Fu$^\diamond$~~~~~~Chunxiao Xing$^\dag$\\
$^\dag$ RIIT, TNList, Dept. of Computer Science and Technology, Tsinghua University, Beijing, China. \\
	$^\sharp$ Computer Science Department, University of California, Los Angeles.\\
	$^\diamond$ Department of Mathematic, University of California, San Diego.\\
	$^\ddag$ Amazon AWS\\
	\fontsize{9}{9}\selectfont\ttfamily\upshape
	wu-jc18@mails.tsinghua.edu.cn; \{zhangyong05,xingcx\}@tsinghua.edu.cn; \\
	\fontsize{9}{9}\selectfont\ttfamily\upshape
	jinwang@cs.ucla.edu; lichunbi@amazon.com; yif051@ucsd.edu
}

\maketitle

\input{src/notations}
\input{src/sec0-abstract}

\input{src/sec1-introduction}

\input{src/sec2-preliminary}
\input{src/sec3-sampling-theory}
\input{src/sec4-sampling-alg}

\input{src/sec5-partition}
\input{src/sec6-discussion}
\input{src/sec7-evaluation}
\input{src/sec8-related}
\input{src/sec9-conclusion}

\bibliographystyle{abbrv}

\input{vldb-mrjoin.bbl}
\end{document}

%% file: src/notations.tex

\newtheorem{definition}{Definition}
\newtheorem{example}{Example}
\newtheorem{lemma}{Lemma}
\newtheorem{theorem}{Theorem}
\newcommand{\name}{\textsf{SP-Join}\xspace}
\newcommand{\cmark}{\ding{51}}%
\newcommand{\xmark}{\ding{55}}%

\newcolumntype{B}{>{\centering\arraybackslash}m{1.5cm}}
\newcolumntype{D}{>{\centering\arraybackslash}m{2.5cm}}

\newcommand{\reminder}[1]{{\mbox{$<==$}} [[[ { \bf #1 } ]]] {\mbox{$==>$}}}
\newcommand{\redfont}[1]{{\color{red}{#1}}}

\newcommand{\ttfont}[1]{{\tt #1}}
\newcommand{\Bi}[1]{\ensuremath{\mathcal{B}_{#1}}\xspace} 

\newcommand{\LB}{\ensuremath{\textsf{LB}}\xspace}
\newcommand{\UB}{\ensuremath{\textsf{UB}}\xspace}

\newcommand{\lowb}[1]{\ensuremath{\textsf{LB}_{\mathcal{B}_{#1}}}\xspace}
\newcommand{\return}{\textsf{return}\xspace}
\newcommand{\topk}{top-\ensuremath{k}\xspace}

\newcommand{\sij}[2]{\ensuremath{\textbf{s}_{#1,#2}^n}\xspace}
\newcommand{\objij}[2]{\ensuremath{o_{#1}^{#2}}\xspace}
\newcommand{\nij}[2]{\ensuremath{n_{l}^{#1, #2}}\xspace}
\newcommand{\inneri}[1]{\ensuremath{\mathcal{V}_{#1}}\xspace} 
\newcommand{\outeri}[1]{\ensuremath{\mathcal{W}_{#1}}\xspace} 
\newcommand{\innerni}[1]{\ensuremath{\mathcal{V}^n_{#1}}\xspace} 
\newcommand{\outerni}[1]{\ensuremath{\mathcal{W}^n_{#1}}\xspace} 
\newcommand{\vecv}[1]{\ensuremath{\textbf{v}_{#1}}\xspace} 
\newcommand{\lenij}[2]{\ensuremath{\texttt{Len}_{l}^{#1, #2}}\xspace}
\newcommand{\lij}[3]{\ensuremath{\mathcal{L}_{#1}^{#2, #3}}\xspace}

\newcommand{\NPPara}[2]{\bm{\mathrm{A}}(#1,#2)} 
\newcommand{\NPParas}{\bm{\mathrm{A}}} 
\newcommand{\NPBound}[2]{\bm{\mathrm{B}}(#1,#2)} 
\newcommand{\NPBounds}{\bm{\mathrm{B}}} 
\newcommand{\NPObjF}{G} 
\newcommand{\NPParaAdv}[2]{\bm{\mathrm{C}}(#1,#2)} 
\newcommand{\NPParasAdv}{\bm{\mathrm{C}}} 

\newcommand{\simmetrics}{\ensuremath{\sim}\xspace}
\newcommand{\simfunc}{\ensuremath{SIM}\xspace}
\newcommand{\all}{\ensuremath{\textsc{All}}\xspace}
\newcommand{\fjac}{\ensuremath{\textsc{Jac}}\xspace}
\newcommand{\fed}{\ensuremath{\textsf{ED}}\xspace}

\newcommand{\bigo}{\ensuremath{\mathcal{O}}\xspace} 
\newcommand{\bigC}{\ensuremath{\mathcal{C}}\xspace} 
\newcommand{\bigS}{\ensuremath{\mathcal{S}}\xspace} 
\newcommand{\bigD}{\ensuremath{\mathcal{D}}\xspace} 
\newcommand{\bigJ}{\ensuremath{\mathcal{J}}\xspace} 
\newcommand{\bigW}{\ensuremath{\mathcal{W}}\xspace} 
\newcommand{\bigV}{\ensuremath{\mathcal{V}}\xspace}
\newcommand{\bigU}{\ensuremath{\mathcal{U}}\xspace} 
\newcommand{\bigP}{\ensuremath{\mathcal{P}}\xspace} 
\newcommand{\bigX}{\ensuremath{\mathcal{X}}\xspace} 
\newcommand{\bigY}{\ensuremath{\mathcal{Y}}\xspace} 
\newcommand{\bigT}{\ensuremath{\mathcal{T}}\xspace} 
\newcommand{\bigF}{\ensuremath{\mathcal{F}}\xspace} 
\newcommand{\bigA}{\ensuremath{\mathcal{A}}\xspace} 
\newcommand{\bigZ}{\ensuremath{\bm{\mathcal{Z}}}\xspace} 
\newcommand{\bigB}{\ensuremath{\bm{\mathcal{B}}}\xspace} 
\newcommand{\bigR}{\ensuremath{\mathcal{R}}\xspace} 
\newcommand{\bigK}{\ensuremath{\mathcal{K}}\xspace} 

\newcommand{\kdtree}{\textsf{KPM}\xspace}
\newcommand{\massjoin}{\textsf{MassJoin}\xspace}
\newcommand{\fsjoin}{\textsf{FSJoin}\xspace}
\newcommand{\mpass}{\textsf{MPASS}\xspace} 
\newcommand{\clusterjoin}{\textsf{ClusterJoin}\xspace}
\newcommand{\mrsim}{\textsf{MRSimJoin}\xspace}
\newcommand{\thetajoin}{\textsf{ThetaJoin}\xspace}
\newcommand{\prefix}{\textsf{PrefixFilter}\xspace}
\newcommand{\vsmjoin}{\textsf{VSmartJoin}\xspace}
\newcommand{\randm}{\textsf{Random}\xspace}
\newcommand{\strat}{\textsf{Stratified}\xspace}
\newcommand{\twolv}{\textsf{TwoStage}\xspace}
\newcommand{\gener}{\textsf{Gen}\xspace}
\newcommand{\dstaw}{\textsf{Dist}\xspace}
\newcommand{\basictr}{\textsf{Iter}\xspace}
\newcommand{\entropy}{\textsf{Learn}\xspace}

\newcommand{\eddist}{\textsc{Edit Distance}\xspace}
\newcommand{\ed}{\textsc{Edit}\xspace}
\newcommand{\overlap}{\textsc{Overlap}\xspace}
\newcommand{\jac}{\textsc{Jaccard}\xspace}
\newcommand{\eu}{\textsc{Euclidean}\xspace}
\newcommand{\cosine}{\textsc{Cosine}\xspace}
\newcommand{\dice}{\textsc{Dice}\xspace}
\newcommand{\eudist}{\textsc{Euclidean Distance}\xspace}
\newcommand{\jacdist}{\textsc{Jaccard Distance}\xspace}
\newcommand{\lpnorm}{\textsc{$L_p$-Norm Distance}\xspace}
\newcommand{\lonenorm}{\textsc{$L_1$-Norm}\xspace}
\newcommand{\lonenormdist}{\textsc{$L_1$-Norm Distance}\xspace}
\newcommand{\vdatao}{\textsc{Netflix}\xspace}
\newcommand{\vdatat}{\textsc{SIFT}\xspace}
\newcommand{\sdatao}{\textsc{AOL}\xspace}
\newcommand{\sdatat}{\textsc{PubMed}\xspace}
\newcommand{\qqgram}{{$q$-gram}\xspace}
\newcommand{\innerpart}{\textsc{Kernel Partition}\xspace}
\newcommand{\outerpart}{\textsc{Whole Partition}\xspace}

\newcommand{\Matrix}[1]{\ensuremath{\bm{\mathrm{#1}}}}

\newcommand{\samplesize}{k\xspace}
\newcommand{\tabincell}[2]{\begin{tabular}{@{}#1@{}}#2\end{tabular}}

%% file: src/sec0-abstract.tex

\begin{abstract}
Given two sets of objects, metric similarity join finds all similar pairs of objects according to a particular distance function in metric space. 
There is an increasing demand to provide a scalable similarity join framework which can support efficient query and analytical services in the era of Big Data. 
The existing distributed metric similarity join algorithms adopt random sampling techniques to produce pivots and utilize holistic partitioning methods based on the generated pivots to partition data, which results in data skew problem since both the generated pivots and the partition strategies have no quality guarantees.

To address the limitation, we propose \name, an end-to-end framework to support distributed similarity join in metric space based on the MapReduce paradigm, which (i) employs an estimation-based stratified sampling method to produce pivots with quality guarantees for any sample size, and (ii) devises an effective cost model as the guideline to split the whole datasets into partitions in map and reduce phases according to the sampled pivots.
Although obtaining an optimal set of partitions is proven to be NP-Hard, \name adopts efficient partitioning strategies based on such a cost model to achieve an even partition with explainable quality.
We implement our framework upon Apache Spark platform and conduct extensive experiments on four real-world datasets. 
The results show that our method significantly outperforms state-of-the-art methods. 

\end{abstract}

%% file: src/sec1-introduction.tex

\section{Introduction}\label{sec-intro}

Nowadays the emerging data lake problem~\cite{DBLP:conf/cidr/TerrizzanoSRC15} has called for more efficient and effective data integration and analytics techniques over massive datasets.
As similarity join is a fundamental operation of data integration and analytics, scaling up its performance would be an essential step towards this goal.
Given two sets of objects, similarity join aims at finding all pairs of objects whose similarities are greater than a predefined threshold. 
It is an essential operation in a variety of real world applications, such as click fraud detection~\cite{DBLP:conf/kdd/WangMP13}, bioinformatics~\cite{DBLP:journals/pvldb/WandeltSBL13}, web page deduplication~\cite{DBLP:conf/www/XiaoWLY08} and entity resolution~\cite{DBLP:conf/icde/KolbTR12}. 
There are many distance functions to evaluate the similarity between objects, such as \eddist for string data, \eudist for spatial data and \lpnorm for images. 
It is necessary to design a generalized framework to accommodate a broad spectrum of data types as well as distance functions.
To this end, we aim at supporting similarity join \emph{in metric space}, which is corresponding to a wide range of similarity functions.
Without ambiguity, we will use ``metric similarity join'' for short in the rest of this paper.

Acting as an essential building block of data integration, it is expensive to perform the  metric similarity join in the era of Big Data where the number of computations grows quadratically as the data size increases. There is an increasing demand for more efficient approaches which can scale up to increasingly vast volume of data as well as make good use of rich computing resources.
To pluck the valuable information from the large scale of data, many big data platforms, e.g. Hadoop~\footnote{https://hadoop.apache.org/}, Apache Spark~\footnote{https://spark.apache.org/} and Apache Flink~\footnote{https://flink.apache.org/}, adopt MapReduce~\cite{DBLP:journals/cacm/DeanG08} as the underlying programming model to process large datasets on clusters.
In order to take full advantage of the MapReduce framework, it requires to overcome bottleneck regarding communication costs in the distributed environment. 
Moreover, due to the well-known problem of ``the curse of last reducer'', it is also necessary to balance the workload between workers in a distributed environment.

\begin{figure}[!t]
	\hspace{-1.1em}
	\includegraphics [scale=0.35]{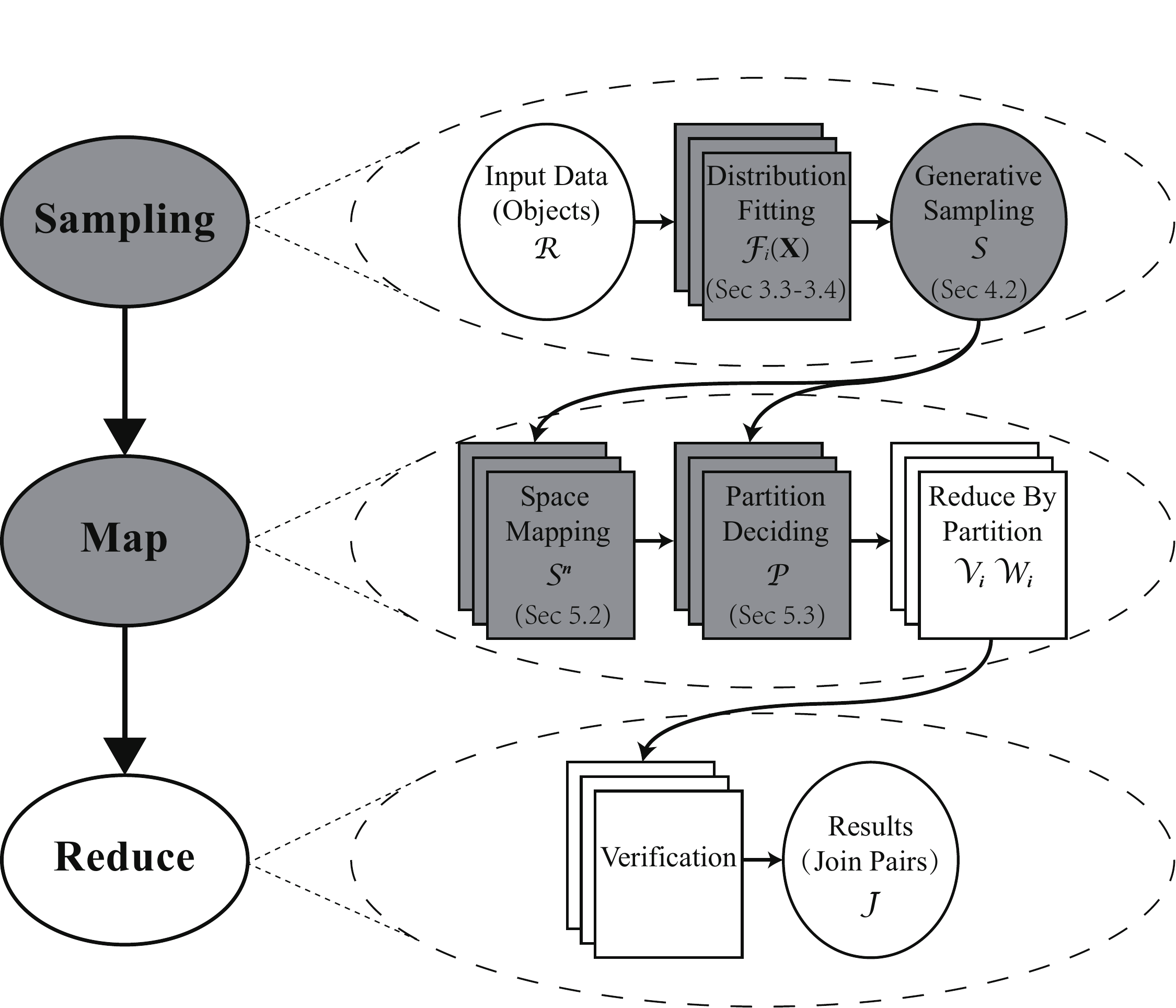}
	\caption{\name: Overall Framework}
	\label{fig-arch}\vspace{-.5em}
\end{figure}

To provide efficient join operations,  previous studies~\cite{DBLP:journals/tkde/ChenYCGZC17,DBLP:journals/pvldb/SarmaHC14,DBLP:conf/kdd/WangMP13,DBLP:conf/icde/FriesBSS14} employ a three-phase (i.e., sampling, map, and reduce) framework for metric similarity join using MapReduce as is shown in the left part of Figure~\ref{fig-arch}.
In the sampling phase, some objects are sampled as pivots to represent the whole dataset. 
In the map phase, the dataset is divided into partitions according to those pivots using holistic partitioning methods.
In the reduce phase, all partitions are shuffled into reducers and the verification is performed on each reducer. 
The union of results from all reducers is the output for similarity join.

Unfortunately, the existing work suffers from the data skew problem, which can significantly damage the overall performance. The data skew problem is caused by the following two facts: (i) the pivots are produced by uniform random sampling methods and (ii) the partitions are generated by the holistic partitioning methods without quality guarantee. 
The existing studies focus on designing different partition strategies while ignoring the selection of high quality pivots -- they simply employ random sampling to generate the pivots, which can result in very low quality pivots.  
Relying on the low quality pivots to perform partitioning will provide heavily skewed partitions, which damages the overall performance. 
One possible way for the existing work to get high quality pivots is to increase the sample size, which will bring extra overhead and also damage the performance.  
In addition, if the sample size is too large, the benefits will be counteracted by the overhead of map phase~\footnote{Experimental results regarding this are shown in Section~\ref{sec-exp} later.}.

We argue that it is essential to adopt statistical tools such as sampling techniques to avoid the bottleneck brought by improper pivots selection.
Sampling techniques are widely adopted by the problem of Approximate Query Processing (AQP), which have been proved to be effective in helping data scientists identify the subset of data that needs further drill-down, discovering trends, and enabling fast visualization~\cite{DBLP:conf/sigmod/ChaudhuriDK17}.
Therefore, there has been a long stream of research work about AQP in the database community, which are applied to problems of data visualization~\cite{DBLP:journals/pvldb/MackeZHP18}, query optimization~\cite{DBLP:conf/sigmod/ParkMSW18} and business intelligence (BI)~\cite{DBLP:conf/sigmod/DingHCC016}.
Motivated by these works, in this paper we focused on devising effective sampling approaches to boost the overall performance of distributed metric similarity join. 

We propose \textbf{S}ampling \textbf{P}owered Join (\name), a scalable framework to mitigate the overhead of metric similarity join in big data analysis. 
The workflow of \name is shown in the right part of Figure~\ref{fig-arch}, where the highlighted parts are brand new techniques proposed by us compared with previous studies. 
To provide pivots with quality guarantee, we aim at adopting \emph{stratified sampling} instead of randomly selecting the pivots in the sampling phase.
However, it is rather expensive to directly apply the stratified approach in the distributed environment.
The reason is that it requires to first categorize the objects into separate strata and then perform random sampling within each stratum.
While it is straightforward in case of single node, it is non-trivial to build the strata in distributed environment due to the heavy network transmission overhead.
We address this problem by conducting stratified sampling from another aspect: \emph{Instead of constructing strata by shuffling objects, we can conduct stratified sampling on each single node in a cluster separately with necessary statistical information}. 
In this way, we can obtain the global samples by aggregating the results from all nodes with only one map job.
Therefore, we first utilize statistical tools to make a wise decision on the contribution of each node without really construct the global strata.
Then propose two light-weighted sampling algorithms to conduct stratified sampling with trivial overhead under MapReduce framework. 
In addition, we theoretically analyze the statistical quality guarantee of our sampling techniques.

Obtaining high quality pivots is just the first step towards efficient similarity join algorithm. 
To better utilize the pivots, it further requires effective strategies to split the dataset into partitions in map and reduce phases.
To this end, we propose a cost model to formally quantify the cost of the partition problem.
According to this cost model, it is proven that obtaining optimal partitions is NP-hard.
We thus propose efficient and explainable partitioning strategies which can progressively reduce the partition cost and achieve good performance in practice.
Compared with previous approaches, our method can make better use of the sampled pivots and include much fewer dissimilar objects within each partition, which can significantly reduce the network communication cost as well as overall computation time. 

The contributions of this paper are summarized as following:
\begin{compactitem}
	\item We propose \name, a MapReduce-based framework to support similarity join in metric space. Our framework can efficiently support multiple distance functions with different kinds of data.
	\item  We devise novel sampling approaches for metric similarity join in MapReduce framework that employ light-weighted and theoretically sound techniques for selecting representative pivots. With such sampling techniques, we can estimate the global data distribution from that of each node in a cluster. Then it could provide more prior knowledge of the data distribution so as to enhance the overall performance. 
	\item We conduct comprehensively theoretical study on the proposed sampling techniques and provide a progressive quality guarantee on given sample size.
	\item We theoretically study the cost model of partitioning in map phase and propose effective partition strategies to ensure load balancing accordingly.
	\item We implement our framework upon Apache Spark and conduct extensive experiments for a variety of distance functions using real world datasets. Experimental results show that our method significantly outperforms state-of-the-art techniques on a variety of real applications.
\end{compactitem}	

Note that the motivation and detailed application parts of this work has been published as a 4-page poster in~\cite{icde19mrjoin}.
This paper contains significant improvement in technical contribution, theoretical analysis, and as well as experimental results compared with that short version. Thus it does not violate the policy regarding ``Originality and Duplicate Submissions”.

The rest of this paper is organized as follows: Section~\ref{sec-prelim} introduces necessary preliminaries and problem settings. 
Section~\ref{sec-sampling} describes the foundation of statistics in the sampling phase.
Section~\ref{sec-pivot} proposes two sampling techniques to select the pivots for partition as well as provides the theoretical analysis of error bounds.
Section~\ref{sec-partition} presents the partition strategies for map phase. 
Section~\ref{sec-dis} proposes some necessary discussions about our framework.
The experimental results are shown in Section~\ref{sec-exp}. 
Section~\ref{sec-related} reviews the related work. 
Finally, the conclusion is given in Section~\ref{sec-con}.

%% file: src/sec2-preliminary.tex
\section{Problem Statement}\label{sec-prelim}

In this paper, we focus on the problem of similarity join in metric space. 
First, we give the definition of metric space distance and its properties as is shown in Definition~\ref{def-metric}.

\begin{definition}[Metric Space Distance~\cite{bryant1985metric}] \label{def-metric}
	Let \bigU be the domain of data, $o_x, o_y$ and $o_z$ are arbitrary objects in \bigU. A metric space distance on \bigU is any function $\bigD: \bigU \times \bigU \rightarrow \mathbb{R}$  satisfying\vspace{-.5em}\\
\begin{compactitem}
\item Non-negativity: $\forall o_x,o_y,\ \bigD(o_x,o_y) \geq 0$
\item Coincidence Axiom: $\bigD(o_x,o_y) = 0$ iff $o_x=o_y$
\item Symmetry: $\bigD(o_x,o_y) = \bigD(o_y, o_x)$
\item Triangle Inequality:	$\bigD(o_x, o_z) \leq \bigD(o_x,o_y) + \bigD(o_y,o_z)$
\end{compactitem}
\end{definition}	
Based on the above definition, we formally define our problem in Definition~\ref{def-simjoin}. 

\begin{definition}[Metric Similarity Join]\label{def-simjoin}
	Given two sets of objects \bigX and \bigY which consist of $m$-dimensional vectors, metric similarity join aims at finding all pairs of $\langle o_x,o_y\rangle$ from $\bigX \times \bigY \in \bigD(o_x,o_y) \leq \delta$ where \bigD is a metric space distance function specified by the user and $\delta$ is its threshold.
\end{definition}

\begin{example}
	This example shows the similarity join in metric space with \lonenormdist. 
	Given two $m$-dimensional vectors $X$ and $Y$, the \lonenormdist between them can be calculated as:
	$L_1dist(X, Y) = \sum_{i=1}^m (|x_i - y _i|)$, where $x_i$ and $y_i$ are the value in the $i^{th}$ dimension of $X$ and $Y$, respectively.
	
	As the distance function on a dataset containing 4 objects:\\ $o_1 =$ [16,35,5,32,31,14,10,11],
	$o_2 =$  [15,33,2,35,29,13,11,12],\\ $o_3 =$ [10,27,8,26,37,23,15,13] and $o_4 = $ [9,30,4,25,34,25,18,14].
	 Assume the given threshold for similarity join is $\delta=30$, then $\bigD(o_1, o_2)=|16-15|+|35-33|+...+|11-12|=14$, $\bigD(o_1, o_3)=45$,  $\bigD(o_1, o_4)=45$,  $\bigD(o_2, o_3)=49$,  $\bigD(o_2, o_4)=47$,  $\bigD(o_3, o_4)=18$. Therefore, the final results of similarity join are $\langle o_1,o_2\rangle$ and $\langle o_3,o_4\rangle$.
\end{example}	

In this paper, without loss of generality, we focus on the self-join of a dataset $\bigR$. 
Notice that, it is straightforward to extend our framework to support non-self-join case. 
We also use \lonenormdist to illustrate the techniques in this paper.
But our techniques could be generalized to all distance functions in metric space, including those for string data (see Section~\ref{subsec-dis-str} for more discussions).

%% file: src/sec3-sampling-theory.tex
\section{Estimating Distribution}\label{sec-sampling}
 
 In this section, we introduce the foundation of our sampling approaches.
 We first illustrate the motivation of devising effective sampling approach in Section~\ref{subsec-motivation}.
 We then provide an overview of our sampling techniques in Section~\ref{subsec-spover}.
 Next, we introduce two crucial steps, i.e. Distribution Estimation (Section~\ref{subsec-dist}) and Confidence Calculation (Section~\ref{subsec-confidence}) of the sampling phase.
 
\subsection{Motivation}\label{subsec-motivation}

Recall the three-phase framework of state-of-the-art methods in Figure~\ref{fig-arch}, in this section we aim at improving the sampling phase, which dominates the overall performance.
The cornerstone of our approaches comes from the fact that each pivot reveals a piece of information about the underlying but unknown distribution of the entire dataset.
Therefore, a set of perfect pivots, which reflects a concise statistical distribution of the underlying data, can bring significant performance benefits for metric similarity join. 
In the ideal case, if we have access to an incredibly precise distribution of the underlying data, we can divide the workload evenly across all nodes and minimize the network communication cost.
To demonstrate it, we first intuitively use a running example to show the importance of sampling.

\begin{example}
	In Figure~\ref{fig:sample_importance}, we use sampled pivots (large red cross) to split all objects (blue points) in the dataset into four parts. 
	The partitions are split by red lines.
	Among them, Figure~\ref{subfig:badsample} demonstrates that bad samples can lead to skewness in the partition.
	The numbers of data points of four partitions in Figure~\ref{subfig:badsample} are 88, 6, 93, 13, respectively.
	Therefore, the maximum verification number among partitions in this situation is about 4278.
	Meanwhile, once we have good samples as shown in Figure~\ref{subfig:goodsample}, there are balanced partitions.
	Then the cardinality of these four partitions are 47,51,48,54, respectively.
	Correspondingly, the maximum verification number is about 1431, which is approximately 3 times better than the case of bad sampling.
\end{example}

\begin{figure}[h]
	\begin{center}
		\subfigure[\small{Bad Samples}]{
			\label{subfig:badsample}
			\epsfig{figure=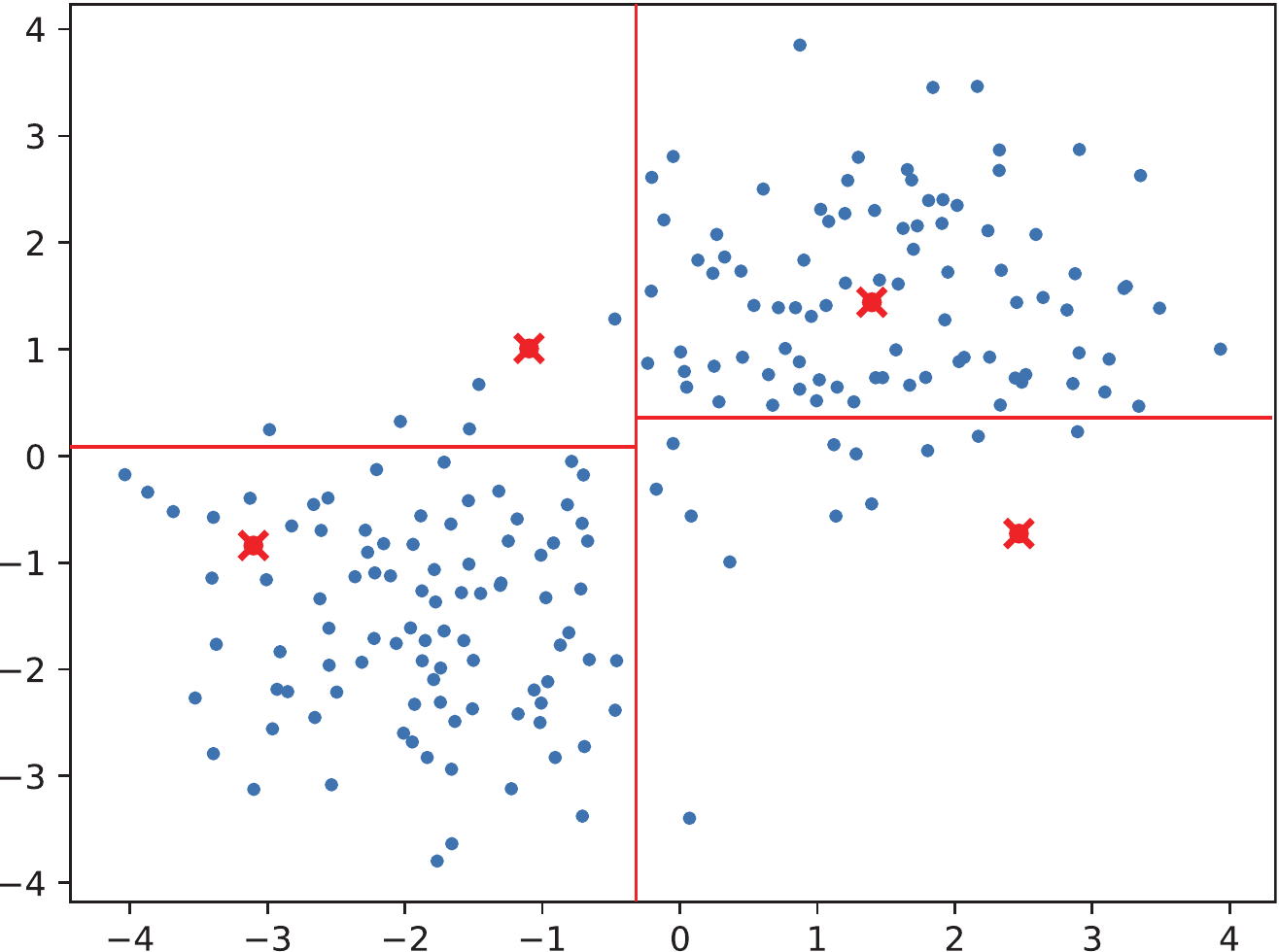,width=0.22\textwidth}}
		\subfigure[\small{Good Samples}]{
			\label{subfig:goodsample}
			\epsfig{figure=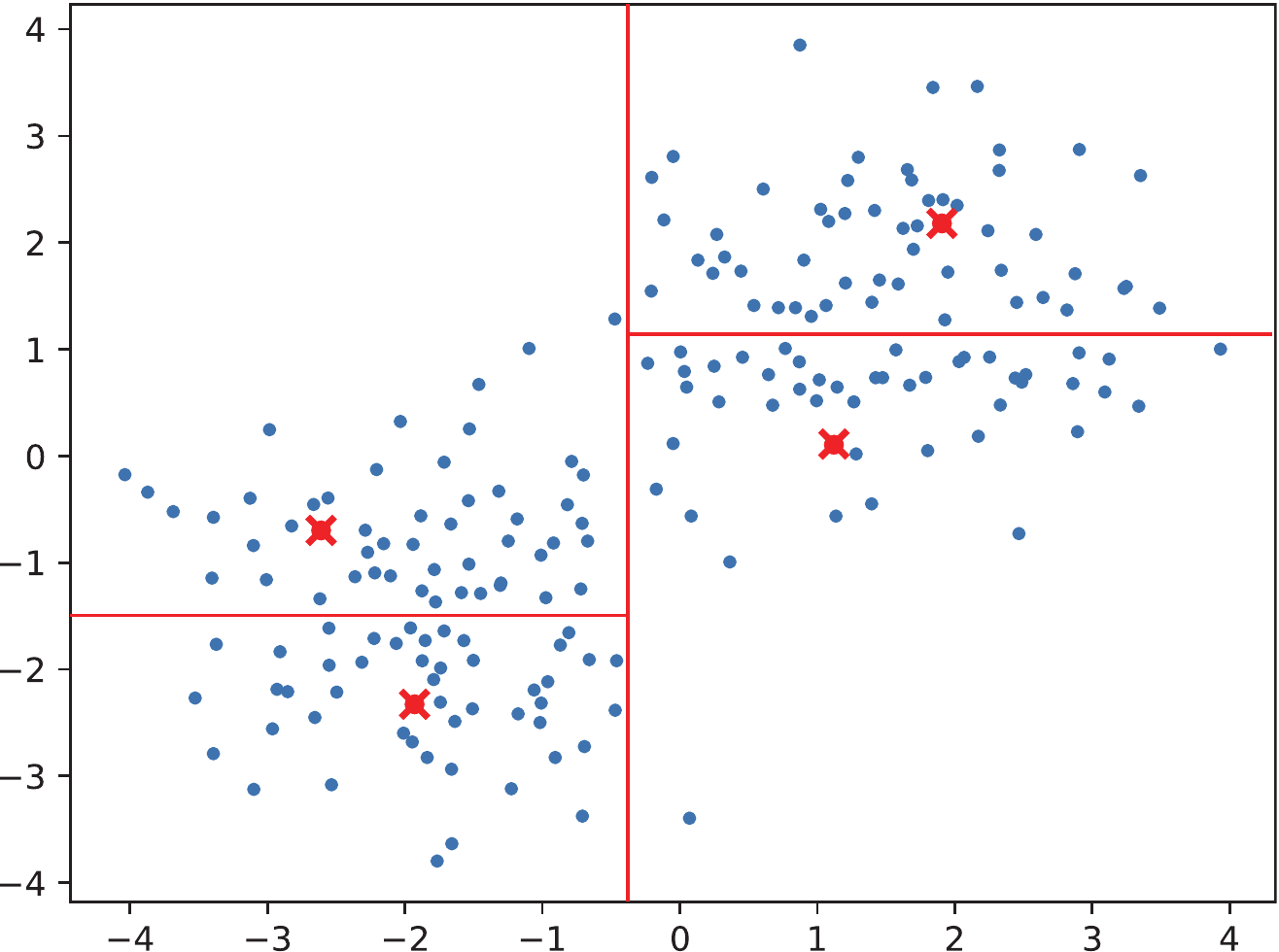,width=0.22\textwidth}}
	\end{center}\vspace{-1em}
	\caption{Importance of Sampling}\label{fig:sample_importance}
\end{figure}


\subsection{Overview of Sampling Process}\label{subsec-spover}

\begin{figure}[h!t]\vspace{-1em}
	\linesnumbered
	\SetVline
	\begin{algorithm}[H]
		\caption{Overall Sampling Algorithm\label{alg:overall}}
		\SetVline
		\For{each local node $i$}{
			Estimate the parametric Distribution $P(\bm{x};)$ \nllabel{alg:overall:st11}\\
			Construct the test statistics for $K_i$ and obtain its value \nllabel{alg:overall:st21} \\
			Obtain the confidence level using Equation~\ref{equ-confn}\nllabel{alg:overall:st22} \\
			Broadcasting the distribution parameters $\langle F_i(\bm{x})$, $c_i^0$, $N_i\rangle$ to all local nodes \nllabel{alg:overall:st31}\\
			Obtain samples. \nllabel{alg:overall:st32}\\
		}
	\end{algorithm}\vspace{-1em}
\end{figure}

In order to tackle the inherent deficiency of previous studies~\cite{DBLP:journals/tkde/ChenYCGZC17,DBLP:journals/pvldb/SarmaHC14,DBLP:conf/kdd/WangMP13} using simple random sampling, we propose a novel framework which aims at supporting stratified sampling to provide better pivots and formal guarantee.
To describe the sampling techniques, we first introduce some terminologies that will be used in Section~\ref{sec-sampling} and Section~\ref{sec-pivot}.
In the distributed environment, we call each single machine a \emph{local node} in the cluster.
The whole dataset $\bigR$ is distributed to all local nodes in our problems setting.
$N_\bigR$ is the cardinality of the dataset; $m$ is the dimension number of each object; $k$ is the number of required samples; $M$ is the total number of nodes. 

From the perspective of statistics, the objects in $\bigR$ can be modeled as $N_\bigR$ independent and identically distributed observations on $m$-dimensional random variables which are denoted as $\bm{X}$. 
Then a dataset can be described with its parametric distribution, i.e. its Cumulative Distribution Function (CDF) $F_i(\bm{x}) = P(\bm{x};\bm{\eta})$, where $\bm{\eta}$ denotes the distribution parameter set. 
Correspondingly, the Probability Density Function (PDF) is denoted as $p(\bm{x};\bm{\eta})$.
If a random variable $\bm{X}$ conforms the distribution, we denote it as: $\bm{X}\sim P(\bm{x};\bm{\eta})$.
Specifically, the distribution of the whole dataset is called \emph{global distribution}; while that of each local node is called \emph{local distribution}.

With above terminology, we then outline our sampling techniques.
To implement stratified sampling in distributed environment, the naive approach needs multiple map and reduce jobs to construct the strata and obtain the samples, which is rather inefficient due to the expensive network transmission cost.
One tempting approach to improve that is to conduct stratified sampling on each local node separately without constructing a global strata.
However, it would definitely compromise the quality of sampling since each local node could only construct strata with local information.
In order to remedy it, we utilize statistical tools to estimate the global distribution from local distributions of each node.
With the help of global information, we can make a wise decision on the sample size and strata construction for each local node.
Following this route, we propose a 3-stage sampling framework shown in Algorithm~\ref{alg:overall}:

\textbf{Stage 1: Distribution Parameter Estimation.} (Section~\ref{subsec-dist}).
During the stage 1, we estimate a parametric distribution $P(\bm{x};\bm{\eta})$ for each local node in a cluster (line~\ref{alg:overall:st11}). 
This can be done by the well-known \emph{Maximum Likelihood Estimation} technique. 

\textbf{Stage 2: Distribution Confidence Calculation.}(Section~\ref{subsec-confidence}).
With only the local distribution of each node obtained in Stage 1, it is not enough to estimate the global distribution as well as conduct stratified sampling just from the parametric distributions.
In order to decide the contribution that each local node makes to the global distribution, we also need to identify \emph{the confidence level that the estimated local distribution holds} in each local node. 
In Stage 2, we finish this step with the tool of \emph{hypothesis testing} (line~\ref{alg:overall:st21}-\ref{alg:overall:st22}).

\textbf{Stage 3: Samples Acquisition.}(Section~\ref{sec-pivot}).
Finally, with the parametric distribution and the confidence level of each local node, we are able to conduct stratified sampling on all local nodes to obtain the final sampled pivots (line~\ref{alg:overall:st31}-\ref{alg:overall:st32}). 
We first propose a distribution-aware sampling approach to better utilize the distribution and its confidence acquired in the first two stages.
Moreover, we further devise a generative sampling approach, which makes the sampling overhead independent of the sample size. 

\subsection{Distribution Estimation}\label{subsec-dist}
\begin{table*}[!t]
	\centering
	\caption{Summary of Common Distributions Represented by \emph{Exponential Family Distribution}}
	\renewcommand{\tabcolsep}{2mm}
	\begin{tabular}{|c|c|c|c|c|c|}
		\hline
		Distribution & Probability Density/Mass Function & $\bm{\eta}$ & $h(\bm{x})$ & $\bm{T}(\bm{x})$ & $\alpha(\bm{\eta})$  \\
		\hline
		Exponential & $f(x|\lambda) = \left\{ \begin{array}{lr}
		\lambda \mathrm{e}^{-\lambda x} & x \geq 0 \\
		0 & x < 0
		\end{array}\right. $ & $-\lambda$& 1 & $x$ & $-\log(-\eta)$ \\ \hline
		Gamma & $f(x|\alpha, \beta) = \dfrac{\beta^\alpha}{\Gamma(\alpha)}x^{\alpha - 1}\mathrm{e}^{-\beta x}$ & $\left[\begin{matrix} -\alpha - 1 \\
		-\beta \\ \end{matrix}\right]$ & 1 & 
		$\left[\begin{matrix} \log x \\
		x \\ \end{matrix}\right]$ & $-\log\Gamma(\eta_1+1) - (\eta_1 + 1)\log(-\eta_2) $ \\
		\hline
		\tabincell{c}{Multivariate\\Normal} & 	$f(\bm{x}|\bm{\mu},\bm{\Sigma}) = \dfrac{\mathrm{e}^{-\frac{1}{2}(\bm{x}-\bm{\mu})^{\mathrm{T}}\bm{\Sigma}^{-1}(\bm{x}-\bm{\mu})}}{\sqrt{(2\pi)^m|\bm{\Sigma}|}}$    & 
		$\left[\begin{matrix} \bm{\Sigma}^{-1}\bm{\mu} \\
		-\frac{1}{2}\bm{\Sigma}^{-1} \\ \end{matrix}\right]$                 
		&  $(2\pi)^{-\frac{k}{2}}$            &  
		$\left[\begin{matrix} \bm{x} \\
		\bm{x}\cdot\bm{x}^{\mathrm{T}} \\ \end{matrix}\right]$           &  $-\dfrac{1}{4}\bm{\eta}_1^{\mathrm{T}}\bm{\eta}_2^{-1}\bm{\eta}_1- \dfrac{1}{2}\ln(|-2\bm{\eta}_2|)$ \\\hline
		Wishart & $f_p(\bm{x}|\bm{V}, n) = \dfrac{|\bm{x}|^{\frac{n-p-1}{2}}\mathrm{e}^{-\frac{\mathrm{tr}(\bm{V}^{-1}\bm{x})}{2}}}{2^{\frac{np}{2}}|\bm{V}|^{\frac{n}{2}}\Gamma_p(\frac{n}{2})}$ & $\left[\begin{matrix} -\frac{1}{2} \bm{V}^{-1} \\ \frac{n-p-1}{2} \\ \end{matrix}\right]$ & 1 & $\left[\begin{matrix} \bm{x} \\ \log|\bm{x}| \\ \end{matrix}\right]$ & $-\dfrac{n}{2}\log|-\bm{\eta}_1| + \log\Gamma_p(\dfrac{n}{2})$ \\ \hline
		Dirichlet & $\begin{aligned} f(\bm{x}|\bm{\alpha}) &= f(x_1, \cdots, x_K | \alpha_1, \cdots, \alpha_K)  \\ &=\dfrac{\Gamma(\sum_{i=1}^K\alpha_i)}{\prod_{i=1}^K\Gamma(\alpha_i)} \prod_{i=1}^K x_i^{\alpha_i - 1}\end{aligned}$ & $\left[\begin{matrix} \alpha_1 \\ \vdots \\
		\alpha_K \\ \end{matrix}\right]$ & $\dfrac{1}{\prod_{i=1}^K x_i}$ & $\left[\begin{matrix} \log x_1 \\ \vdots \\
		\log x_K \\ \end{matrix}\right]$ & $\sum_{i=1}^K \log\Gamma(\eta_i) - \log \Gamma(\sum_{i=1}^K \eta_i)$ \\
		\hline
	\end{tabular}
	\label{tab:exponential_family}\vspace{-1em}
\end{table*}

We first introduce how to adopt the theory of \textbf{M}aximum \textbf{L}ikelihood \textbf{E}stimation (MLE) to estimate the parametric distribution for each local node.
Generally speaking, MLE is a methodology to fit the dataset into a statistical model and then we use Goodness of Fit to describe how well it fits a set of observations in Section~\ref{subsec-confidence}.
However, one problem is that the subset on each local node might conform different types of distributions.
To determine the global distribution, it requires a generalized form to describe the parametric distribution.

Fortunately, the \emph{Exponential Family} (a.k.a Koopman-Darmois Family) Distribution~\cite{efd} provides \textbf{a unified parametric format} to represent the local distributions as well as the global distribution.
It can be used to describe the Probability Density Function (PDF) of most common distributions in the real world, such as Normal, Beta, Gamma, Chi-square distributions by simply varying its parameters, which is formally shown in Definition~\ref{def-efdit}.
\begin{definition}\label{def-efdit}
	Given the parameter set $\bm{\eta}$, the real-valued function of the parameter set $\alpha(\bm{\eta})$ and the statistics functions $\bm{T}(\bm{x})$ and $h(\bm{x})$, the PDF of Exponential Family Distribution can be written as following:
	\begin{equation}
	p(\bm{x};\bm{\eta}) = h(\bm{x})\exp\big({\bm{\eta}^\mathrm{T} \cdot\bm{T}(\bm{x}) - \alpha(\bm{\eta})}\big)
	\end{equation}
	which satisfies
	\begin{equation*}
	\int{h(\bm{x})\exp\big({\bm{\eta}^\mathrm{T}\cdot\bm{T}(\bm{x}) - \alpha(\bm{\eta})}\big)\mathrm{d}\bm{x}} = 1
	\end{equation*}
\end{definition}

By varying $\bm{T}(\bm{x})$ and $h(\bm{x})$, we can use its CDF $P(\bm{x};\bm{\eta})$ to denote different kinds of distribution. 
Table~\ref{tab:exponential_family} shows some different distributions that can be represented with \emph{Exponential Family Distribution} by varying the parameters $\bm{\eta}, h(\bm{x}), \bm{T}(\bm{x}), \alpha(\bm{\eta})$.

Given the data on a local node, we can fit its CDF $P(\bm{x};\bm{\eta})$ by estimating the parameter set $\bm{\eta}$ with Maximum Likelihood Estimation(MLE) approach.
In this way, the formula of parameter set $\bm{\eta}^0$ can be obtained with the help of function $\bm{\mu}(\bm{\eta})$.
Consequently, we could get $\bm{\eta}^0$ since $\bm{T}(\bm{o}_i)$ and $N_\bigR$ can be known for the data on a local node. 
The details are summarized in Lemma~\ref{lem-mle}.

\begin{lemma}\label{lem-mle}
	The parameter set $\bm{\eta}^0$ to describe the distribution of a local node can be estimated as:
	\begin{equation}\label{eq-mle}
		\bm{\eta}^0 = \bm{\mu^{-1}}\Big(\dfrac{\sum_{i=1}^{N_\bigR} \bm{T}(\bm{o}_i) }{N_\bigR}\Big)
	\end{equation}
	where $\bm{\mu}(\bm{\eta}) = \mathbb{E}_{\bm{\eta}}\big(\bm{T}(\bm{x})\big) = \dfrac{\partial \alpha({\bm{\eta}})}{\partial \bm{\eta}}$ is the Mathematical Expectation of the distribution on function $\bm{T}(\bm{X})$ with parameter set $\bm{\eta}$, and $\bm{\mu^{-1}}(\bm{x})$ is the inverse function of $\bm{\mu}(\bm{x})$.
\end{lemma}
\begin{proof}
	Based on the theory of MLE, we first get the likelihood function as following:
	\begin{equation}
	\mathcal{L}(\bm{\eta};\bm{x}) = \prod_{\bm{o}\in \bigR} p(\bm{o}|\bm{\eta}) = \prod_{i=1}^{N_\bigR} h(\bm{o}_i)\exp\big({\bm{\eta}^\mathrm{T} \cdot\bm{T}(\bm{o}_i) - \alpha(\bm{\eta})}\big)
	\end{equation}
	where $\bm{o}$ and $\bm{o}_i$ represent objects in datasets. 
	Then the objective of MLE is to maximize the above likelihood function.
	
	In practice, it is often convenient to work with the natural logarithm of the likelihood function since the original function and the the natural logarithm reach the maximum value at the same parameter set $\bm{\eta}^0$.
	
	Therefore we have:
	\begin{equation}
	\begin{aligned}
	\bm{\eta}^0 = \mathop{\arg\max}_{\bm{\eta}} \big(\mathcal{L}(\bm{\eta};\bm{x})\big) &= \mathop{\arg\max}_{\bm{\eta}} \log\big(\mathcal{L}(\bm{\eta};\bm{x})\big) \\&= \mathop{\arg\max}_{\bm{\eta}} \sum_{i=1}^{N_\bigR}{\big(\bm{\eta}^\mathrm{T} \cdot\bm{T}(\bm{o}_i) - \alpha(\bm{\eta})\big)} 
	\end{aligned}
	\end{equation}
	
	Utilizing that the derivative of function on maximum parameters is 0 based on \textit{Fermat's theorem\footnote{https://en.wikipedia.org/wiki/Fermat\%27s\_theorem\_(stationary\_points)}}, we then obtain the equation contained $\bm{\eta}^0$.
	\begin{equation}\label{mle-proof-1}
	\dfrac{\partial \log\big(\mathcal{L}(\bm{\eta};\bm{x})\big)}{\partial \bm{\eta}}\bigg|_{\bm{\eta}=\bm{\eta}^0} = \sum_{i=1}^{N_\bigR} \bm{T}(\bm{o}_i) - N_\bigR \cdot \dfrac{\partial \alpha(\bm{\eta})}{\partial \bm{\eta}}\bigg|_{\bm{\eta}=\bm{\eta}^0} = 0
	\end{equation}
	
	Moreover, in~\cite{casella2001theory}, we obtain the equation related to $\dfrac{\partial \alpha(\bm{\eta})}{\partial \bm{\eta}}$:
	\begin{equation}\label{mle-proof-2}
	\dfrac{\partial \alpha(\bm{\eta})}{\partial \bm{\eta}} = \bm{\mu}(\bm{\eta}) = \mathbb{E}_{\bm{\eta}}\big(\bm{T}(\bm{x})\big) = \int{\bm{T}(\bm{x})p(\bm{x};\bm{\eta})}\mathrm{d}\bm{x}
	\end{equation}
	
	Finally, based on the equation~\ref{mle-proof-1} and~\ref{mle-proof-2}, we get $\bm{\eta}^0 = \bm{\mu^{-1}}\Big(\dfrac{\sum_{i=1}^{N_\bigR} \bm{T}(\bm{o}_i) }{N_\bigR}\Big)$ where $\bm{\mu^{-1}}$ is the inverse function of $\mathbb{E}_{\bm{\eta}}\big(\bm{T}(\bm{x})\big)$.
\end{proof}

For other distributions whose $\mathbb{E}_{\bm{\eta}}(\bm{T}(\bm{x}))$ cannot be written in an explicit form, we could use Gradient Descent to get the approximate solution of MLE since the derivative is known with the datasets and given forms of function $\alpha(\bm{\eta})$. 
Actually, it is rather easy to obtain the explicit form of $\mathbb{E}_{\bm{\eta}}(\bm{T}(\bm{x}))$ for most common distributions.

\subsection{Confidence Calculation}\label{subsec-confidence}

After estimating $\bm{\eta}^0$ with Lemma~\ref{lem-mle}, we could obtain the data distribution $F_i(\bm{x}) = P(\bm{x};\bm{\eta}^0)$ of each local node.  
The next step is to identify the confidence level that above distribution holds.
The information of confidence level represents to what extent the estimated distribution conforms with the dataset.
And it would influence the sample size and strata construction on each local node.
To describe the techniques, we use the following notations in this section: for local node $i$, its cardinality is $N_i$. 
The test statistics of node $i$ is $K_i$ and the value of it is $K_i^*$.
The confidence level that the distribution holds is $c_i^0$.

The methodology to obtain the confidence level that $P(\bm{x};\bm{\eta}^0)$ holds is based on the \emph{Goodness of Fit}~\cite{huber2012goodness} theory, which is used to describe how a distribution (statistical model) fits the dataset (a set of observations). 
Measures of goodness of fit typically summarize the discrepancy between observed values and the values expected under the distribution in question. 
Such measures can be used in \emph{statistical hypothesis testing}~\cite{cramer2016mathematical}  to test whether outcome frequencies follow a specified distribution. 

Before showing the details in each step, we first provide a general intuition behind these steps. 
Recalled that the dataset can be considered as an observation on random variables which conform the distribution drawn from dataset, i.e. the \emph{true distribution}. 
Intuitively, the true and estimated distribution can be connected with random variables.
In particular, we can propose a random variable called \emph{test statistics} to somehow depict the conformity between the estimated distribution and random variables that conform true distribution. 
Such a conformity can be reflected by the value of test statistics.
Meanwhile, the test statistics also can be validated with the dataset. 
We can then obtain the probability that the value of test statistics is larger than the one validated from true dataset, which also means the probability that the estimated distribution on a local node conforms the true distribution.
\emph{And such a probability can be regarded as the confidence level that the distribution holds}.

To sum up, the procedure of statistical hypothesis testing in Goodness of Fit can be divided into five steps:
\begin{compactitem}
	\item[1.] Propose the hypothesis $H_0$, which is a testable statement on the basis of observing a process modeled via a set of random variables.
	\item[2.] Construct the statistic $K_i$ based on the hypothesis and its related random variable. Therefore $K_i$ is a random variable constructed by the other random variables.
	\item[3.] Prove that the statistics conforms a distribution under the assumption that the hypothesis is true.
	\item[4.] Compute the value of test statistics $K_i$ (denoted as $K_i^*$) from data on the local node. That is, to replace the random variables in $K_i$ with the true gotten value.
	\item[5.] Obtain the confidence, which is the probability (i.e. p-value) that the statistical summary (e.g. mean) of a given distribution would be no worse than actual calculated results. 
\end{compactitem}

Next we will follow these five steps to acquire the confidence level. 
As mentioned in Section~\ref{subsec-spover}, the CDF of random variable $\bm{X}$ can be viewed as the true distribution of data on a local node. 
And objects belonging to the local node can be considered as the independent \emph{observations} on random variable $\bm{X}$, which means \textbf{the observations can also be considered as random variables $\bm{X}_1, \cdots, \bm{X}_{N_i}$ which are independent and conform the same true distribution with $\bm{X}$.} 
With such explanation from the statistical aspect, given the estimated distribution, we can propose the hypothesis(Step 1):
\begin{equation}
H_0: \bm{X} \sim P(\bm{x};\bm{\eta}^0)
\end{equation}
Here $\bm{X}$ is the set of random variables which are observed from true distribution, and $P(\bm{x};\bm{\eta}^0)$ is the estimated distribution.

Correspondingly, if the hypothesis is true, the data on a local node can be regarded as conforming the estimated distribution $P(\bm{x};\bm{\eta}^0)$ obtained in Section~\ref{subsec-dist}. 
Then the test statistics $K_i$ can be constructed as the difference between estimated distribution $P(\bm{x};\bm{\eta}^0)$ and true data distribution on the local node.
If they are close enough, then we could assert the hypothesis is true.

In order to evaluate the closeness, we utilize the method proposed in~\cite{cramer2016mathematical} to discretize the continuous space of random variables $\bm{X}$ into a finite number $t$ disjoint cells $\bigZ_1, \bigZ_2, ..., \bigZ_{t}$, where $t$ is a hyper-parameter and can be set empirically. Then we can derive the test statistics using Lemma~\ref{lem-disc} (Step 2).

\begin{lemma}\label{lem-disc}
	The test statistics $K_i$ can be written as 
	\begin{equation}\label{equ-statk}
	K_i = \sum_{j=1}^{t} \dfrac{\big(\nu_j - N_i \cdot q_j(\bm{\eta}^0)\big)^2}{N_i \cdot q_j(\bm{\eta}^0)}
	\end{equation}
	where $\nu_j = \sum_{i=1}^{N_i}\mathbf{1}_{\{\bm{X}_k \in \bigZ_j\}}$ is the \emph{Frequency} of a cell based on the observations; $q_j(\bm{\eta}^0)= \int_{\bigZ_j}{p(\bm{x};\bm{\eta}^0)\mathrm{d}\bm{x}}$ is the \emph{probability} that the estimated distribution falls into the corresponding cell.
\end{lemma}
\begin{proof}
	For the Exponential Family Distribution, the probability of event collections in cell $\bigZ_j$ can be written as:
	\begin{equation*}\label{equ-cell}
	q_j(\bm{\eta}^0) = \int_{\bigZ_j}{h(\bm{x})\exp\big({(\bm{\eta}^0)^{\mathrm{T}} \cdot\bm{T}(\bm{x}) - \alpha(\bm{\eta}^0)}\big)\mathrm{d}\bm{x}}\nonumber
	\end{equation*} 
	Since $\bm{X} \sim P(\bm{x};\bm{\eta})$, $q_j(\bm{\eta}^0)$ is considered as the \emph{probability} of a cell.
	
	Here the CDF of $\bm{X}$ is considered as true data distribution.
	If it conforms the estimated distribution, we could assert the observed \emph{frequency} that objects (the observation of $\bm{X}$) locate in a cell would be more close to the theoretical probability of that cell based on estimated distribution.
	
	Then our hypothesis can be rewritten as:
	\begin{equation*}\label{equ-hypnew}
	H_0: \mathbb{P}\{\bm{X}\in \bigZ_j\} = q_j(\bm{\eta}^0), \forall j.
	\end{equation*}
	where $\mathbb{P}\{\bm{X}\in \bigZ_j\}$ is the \emph{frequency} that data (the observation of $\bm{X}$) locates in cell $\bigZ_j$.
	
	Moreover, given the observations on node $i$, the frequency of $X \in \bm{X}$, $X \in \bigZ_j$ can be calculated as $\dfrac{\sum_{i=1}^{N_i}\mathbf{1}_{\{\bm{X}_k \in \bigZ_j\}}}{N_i} = \dfrac{\nu_j}{N_i}$, where $\nu_j$ is the frequency of $X \in \bigZ_j$.
	
	Finally, regarding the goodness of fit theory, in order to depict the relative distance between frequency and probability, the test statistics $K_i$ for each local node uses a measure which is the sum of differences between observed and expected outcome frequency, i.e. counts of observations.
	Each of them is squared and divided by the expectation~\cite{cramer2016mathematical}, which thus follows Equation~\ref{equ-statk}.
\end{proof}	

We want to specifically clarify that $K_i$ is a random variable since $\nu_j$ and $\bm{\eta}^0$ are random variables depending on $\bm{X}_1, \cdots \bm{X}_{N_i}$. 
In addition, when using objects on node $i$ to replace the random variable $\bm{X}$, we can then calculate the value of $K_i$. 
In order to distinguish the test statistics from its value, we denote the value as $K_i^*$.

We can prove that statistics $K_i$ conforms a chi-squared distribution by leveraging the statistical tool of Pearson's Chi-Square Test (Step 3), which is formally stated in Theorem~\ref{the-chisqt}.
Note here the value of $w$ and $t$ could be different for different node $i$.
\begin{theorem}\label{the-chisqt}
	If $\bm{X} \sim P(\bm{x};\bm{\eta}^0)$, then the test statistic $K_i$ conforms the Chi-Squared distribution with $t-w-1$ degrees of freedom:
	\begin{equation*}
	K_i \sim \chi^2_{t-w-1}
	\end{equation*}
	where $w$ is the number of parameters in $\bm{\eta}^0$ and $t$ is the number of cells explained above.
\end{theorem}
\begin{proof}
	See~\cite{cramer2016mathematical}.
\end{proof}

Then, we can calculate the value $K_i^*$ of test statistics based on objects of node $i$.
It can be achieved by replacing the random variables (observations) $\bm{X}_k$ with the value of objects on node $i$~\footnote{which can be obtained locally without network transmission} to calculate the value. 
Formally, we can compute the $K_i^*$ in Equation~\ref{equ-kstar} (Step 4).
\begin{equation}\label{equ-kstar}
K_i^* = \sum_{j=1}^{t} \dfrac{\big(\sum_{k=1}^{N_i}\mathbf{1}_{\{\bm{o}_k \in Z_j\}} - N_i \cdot q_j(\bm{\eta}^0)\big)^2}{N_i \cdot q_j(\bm{\eta}^0)}
\end{equation}
where $\bm{o}$ and $\bm{o}_k$ represents the value of objects in local node $i$ rather than random variables compared with $\bm{X}_k$.
Concretely, $\bm{o}_k$ can be considered as the corresponding value of $\bm{X}_k$ shown in data of local node $i$.

After we have the true value of random variable $K_i^*$ and its theoretical distribution under the establishment of hypothesis, it is nature to compare how the true value deviates from the theoretical distribution. 
In other words, it is the probability that the random variable $K_i$ is larger than the true value. 
If the probability is low, it means there is little possibility that the hypothesis is true.
As a result, the confidence to the hypothesis is also low. 
That is the reason why we just use the probability to calculate the confidence in this scenario.

According to above results, given the value $K^*_i$ of test statistics, we can get the confidence level with Equation~\ref{equ-confn} (Step 5). 
\begin{equation}\label{equ-confn}
c^0_i = \sup \{c | K^*_i > \chi^2_{t-w-1}(c)\}
\end{equation}

Actually, this confidence level is the \emph{maximum probability} which makes the hypothesis true.
As is further explained in~\cite{cramer2016mathematical}, if $K^*_i > \chi^2_{t-w-1}(c)$ holds with a given probability threshold $c$, we can assert that the hypothesis is true, i.e. the data conforms the distribution with the parameter set $\bm{\eta}^0$. 
In the practice of the sampling techniques, the $c_i^0 \geq 0.95$ empirically.
In this case, the data of all local nodes can fit at least one distribution in the Exponential Family Distribution.
If there are multiple possible distributions, we select the distribution with the maximum confidence as the result. 

%% file: src/sec4-sampling-alg.tex
\section{Sampling Algorithms}\label{sec-pivot}

With the distribution parameters and confidence obtained in Section~\ref{sec-sampling}, we can then acquire the sampled pivots using only one map job (Stage 3 in Algorithm~\ref{alg:overall}).
We first propose a distribution-aware sampling approach by leveraging these statistics (Section~\ref{subsec-distsample}).
Then we improve it with a generative approach, which makes the cost of sampling independent from the sample size (Section~\ref{subsec-gensample}).
Finally, we make theoretical analysis and provide a formal quality guarantee of our techniques (Section~\ref{subsec-errorbd}). 

\subsection{Distribution-aware Sampling}\label{subsec-distsample}

After the first two stages in Algorithm~\ref{alg:overall}, we can utilize the confidence to decide the number of samples that should be obtained from each node $i$ as Equation~\ref{equ-sps}.
\begin{equation}\label{equ-sps}
   \frac{{N_i}/{c_i^0}}{\sum_j({N_j/c_j^0})} \cdot k
\end{equation}
The intuition is that the higher confidence there is, the more we know about the distribution of a node.
Therefore, we should fetch more samples from the nodes whose distribution is associated with lower confidence so as to acquire more knowledge about them and make the global distribution more reliable.
After that, we can conduct stratified sampling on each local node $i$ with the help of the estimated distribution $F_i(\bm{x})$.

Next we analyze the quality of our estimation of the global distribution using the global test statistic $\widebar{K}$ and the global confidence $\widebar{c^0}$. 
A lower bound of $\widebar{c^0}$ is deduced in Theorem~\ref{tho-glbconf}.
\begin{theorem}\label{tho-glbconf}
	The lower bound of global confidence $\widebar{c^0}$ is the minimum value of all confidences ${c_i^0}$ from each node.
	\begin{equation}
	\widebar{c^0} \geq \min_i {c_i^0}
	\end{equation}
\end{theorem}
\begin{proof}
	Since the global distribution is obtained from the combination of all local distributions, we consider the test statistic $\widebar{K}$ of global distribution as the sum of the test statistics of these individual local distributions, i.e. $\widebar{K} = \sum_i K_i$. 
	According to Theorem~\ref{the-chisqt}, we have $\widebar{K} = \sum_i K_i \sim \sum_i\chi^2_{(t-w-1)}$. 
	Thus, the global confidence $\widebar{c^0}$ can be decided based on the value of global test statistic $\widebar{K^*}$ as shown in Equation~\ref{equ-glconf}.
	\begin{equation}\label{equ-glconf}
	\widebar{c^0} = \sup \{c | \widebar{K^*} > \sum_i\chi^2_{(t-w-1)}(c)\}
	\end{equation}
	
	We denote $\widebar{c} = \min_i c_i^0$ and will prove $\widebar{c} \geq \widebar{c^0}$.
	
	Based on the definition of $c_i=\sup\{c|K_i^* > \chi^2_{(t-w-1)}(c)\}$ in Equation~\ref{equ-confn}, we have the inequality for any $c\leq c_i^0$:
	\begin{equation}\label{ineq-proof}
	K_i^* > \chi^2_{t - w - 1}(c), (\forall c \leq c_i^0)
	\end{equation}
	
	Noticed that $\widebar{c}$ is the minimum value among $c_i^0$, which means $\widebar{c} \leq c_i^0$ for any $i$. Therefore, derived from the Inequality~\ref{ineq-proof}, the following formula is established for any $i$.
	\begin{equation}\label{ineq-proof-2}
	K_i^* > \chi^2_{t - w - 1}(\widebar{c})
	\end{equation}
	
	Next, we sum up the both sides of the Inequalities~\ref{ineq-proof-2} on different $i$ to get the following one:  
	\begin{equation}
	\widebar{K^*} = \sum_i K_i^* > \sum_i\chi^2_{(t-w-1)}(\widebar{c})
	\end{equation}
	
	Finally, based on the definition of $\widebar{c^0}$ in Equation~\ref{equ-glconf}, we can see that $\widebar{c} \geq \widebar{c^0}$, which completes the proof.
\end{proof}

\begin{figure}[h!t]\vspace{-1em}
	\linesnumbered
	\SetVline
	\begin{algorithm}[H]
		\caption{distribution-aware Sampling}\label{alg:distsample}
		\SetVline
		\Begin{
			 Calculate local sample size $lc$ for each node using Equation~\ref{equ-sps} \nllabel{alg:distsample:jud}\\
			For each node $i$, split its space into $\lfloor\sqrt{lc}\rfloor$ boxes as $\bigB$ \nllabel{alg:distsample:split} \\
			\For{each box $\bigB_j$}{
				Calculate probability $\mathbb{P}\{\bm{X} \in \bigB_j\}$ under $F_i(\bm{x})$ \nllabel{alg:distsample:poss}\\
				Get $lc \cdot {\mathbb{P}\{\bm{X} \in \bigB_j)\}} $ samples based on $c_i^0$ \nllabel{alg:distsample:samp}}
			Combine all samples from above boxes into $\bigS_i$ \nllabel{alg:distsample:comb} \\
			Collect all $\bigS_i$ from each node and combine them into $\bigS$ \nllabel{alg:distsample:red}\\
			\return sampled pivots $\bigS$
		}
	\end{algorithm}\vspace{-1em}
\end{figure}

Based on above analysis, we propose a \emph{distribution-aware sampling approach} as is demonstrated in Algorithm~\ref{alg:distsample}. 
For each local node $i$, we already obtain the distribution $F_i(\bm{x})$ and the confidence $c_i^0$ using the methods described in Section~\ref{sec-sampling}. 
Then we determine the number of samples on each node (line~\ref{alg:distsample:jud}).
Next we utilize the distribution information to do sampling on one single node: we try to split the space of random variable values into $\lfloor\sqrt{lc}\rfloor$ boxes (line~\ref{alg:distsample:split}).
Each box has its correspondingly probability $\mathbb{P}\{\bm{X} \in \bigB_j\}$ w.r.t the estimated distribution (line~\ref{alg:distsample:poss}). 
Then we can use the corresponding probability to determine the portion we should sample from each range, and conduct stratified sampling (line~\ref{alg:distsample:samp}). 
In this process, we will randomly reject some samples with possibility $1-c_i^0$ and perform resampling until we get $lc \cdot {\mathbb{P}\{\bm{X} \in \bigB_j\}}$ samples. The reason is that $c_i^0$ is the confidence of fitting and we could consider it as the acceptance degree of each sample.
Then on each node, we collect samples from all ranges and construct the local sample collection (line~\ref{alg:distsample:comb}). 
Similarly, we combine the samples from all nodes and get the final result (line~\ref{alg:distsample:red}).

\subsection{Generative Sampling}\label{subsec-gensample}

One potential bottleneck of above distribution-aware sampling is that it needs to fetch samples from each local node using a map job. 
Then the network communication cost will increase linearly with the total sample size $k$. 
In this section, we propose a \emph{generative sampling approach} to further reduce the overhead. 
The cornerstone is that the higher level goal of sampling is to obtain representative pivots so as to help split the whole dataset into partitions in the following phases. 
To reach this goal, the samples should reveal enough knowledge about the global distribution of dataset. 
Nevertheless, they do not have to come from the original dataset. 
Therefore, instead of directly utilizing the distribution of each local node to obtain samples, we combine them to simulate a global distribution of the whole dataset. 
Then \emph{unlike the previous approaches which fetch samples from each local node, we generate the sampled pivots according to this global distribution}. 
In this way, we only need to transmit some parameters instead of real sampled objects from each local node. 
And \emph{the network communication cost is independent from the total sample size.}

For the generative sampling approach, the first two stages are the same with that of the distribution-aware sampling: utilizing fit of goodness to get the distribution parameters and confidence level of each local node.
The next question is how to combine them into a global one. 
To reach this goal, we define three random variables for the given $M$ local nodes. $\bm{X}$: the random variables for local distribution; $E$: the discrete random variable from $i = \{1,...,M \}$, which denote the nodes to perform sampling; $C$: selector, the value $1$($0$) represents accept(reject). Then we can deduce the conditional distributions between them as follows:
\begin{equation}\label{eq-gec}
p(E=i|C=c) = 
\dfrac{N_i\cdot(c_i^0)^{-c}}{\sum_j{N_j\cdot(c_j^0)^{-c}}}
\end{equation}
\begin{equation}\label{eq-gxe}
p(\bm{X}|E=i) = \dfrac{\mathrm{d}F_i(\bm{X})}{\mathrm{d}\bm{X}} = f_i(\bm{X})
\end{equation}
\begin{equation}\label{eq-gce}
p(C=c|E=i) = (-1)^{c+1}\cdot c_i^0 + \mathrm{\mathbf{1}}_{\{c=0\}}
\end{equation}
Here, $\mathrm{\mathbf{1}}_{\{cond\}}$ is the 0/1 indicator function.
With above PDF of conditional distributions, we could determine the PDF of global distribution as $p(\bm{X}, E, C)$. 

Although the idea is simple, it is non-trivial to explicitly represent this global distribution.
As a result, it is difficult to obtain the joint distribution (global distribution) of those conditional distributions and perform sampling directly. 
To address this issue, we adopt the Gibbs Sampling approach~\cite{DBLP:journals/technometrics/Kim00} which could generate samples of the joint distribution from conditional ones.

\begin{figure}[!t]\vspace{-.5em}
	\linesnumbered \SetVline
	\begin{algorithm}[H]
		\caption{Generative Sampling}\label{alg:gensample}
		\SetVline
		\Begin{
			Collect all sampling distribution types and parameters; \nllabel{alg:gensample:col} \\
			Construct these conditional distributions in Euqtion~\ref{eq-gec} to~\ref{eq-gce} for each node; \nllabel{alg:gensample:cons} \\
			$\bigS$ = $Gibbs Sampling(k)$ \nllabel{alg:gensample:samp} \\
			\return sampled pivots $\bigS$
		}
	\end{algorithm}\vspace{-1em}
\end{figure}

The process of generative sampling method is shown in Algorithm~\ref{alg:gensample}. 
Similar to Algorithm~\ref{alg:distsample}, we first fit the distribution $F_i(\bm{x})$ and get the confidence $c_i^0$. 
Next we collect the distribution information and confidence from all local nodes (line~\ref{alg:gensample:col}) and then construct conditional distributions (line~\ref{alg:gensample:cons}). 
Finally, we use the Gibbs Sampling method in Algorithm~\ref{alg:gibbs_sample} approach to generate samples.

\begin{figure}[h!t]\vspace{-.5em}
	\linesnumbered \SetVline
	\begin{algorithm}[H]
		\caption{Our Gibbs Sampling ($k$) \label{alg:gibbs_sample}}
		\KwIn{$k$: The sampling Size}
		\KwOut{$\bigS$: The set of sampled pivots}
		\SetVline
		\Begin{
			Initialize $s_0 = \{ \bm{x_0}, e_0, c_0 \}$, $i = 1$ \nllabel{alg:gibbs_sample:init}\\
			Append $s_0$ to $\bigS$\\
			\While{i < k}{
				$s_i.e \sim p(E|C=s_{i-1}.c)$ \\
				$s_i.\bm{x} \sim p(X|E=s_i.e)$ \\
				$s_i.c \sim p(C|E=s_i.e)$ \\
				\If{$s_i.c == 1$}{
					Append $s_i.\bm{x}$ to $\bigS$ \nllabel{alg:gibbs_sample:acc}\\
					$i = i + 1$
				}\Else{
					$s_i = s_{i-1}$ \nllabel{alg:gibbs_sample:rej}
				}
			}	
			\return $\bigS$ \nllabel{alg:gibbs_sample:fin};
		}
	\end{algorithm}\vspace{-.5em}
\end{figure}

We show the Gibbs Sampling in our forms in Algorithm~\ref{alg:gibbs_sample}. 
We begin with some random initial value to get the first sample $s_0$ (line~\ref{alg:gibbs_sample:init}). 
Then we sample each component of next sample, e.g., $s_i.c$ from the distribution of that component conditioned on all other components sampled so far. Actually, those conditional distributions have been given before. 
After obtaining each component, we will get the next sample $s_i$. 
If $s_i.c$ equals $0$, we need to reject this sample (line~\ref{alg:gibbs_sample:rej}). 
Otherwise, we accept the sample (line~\ref{alg:gibbs_sample:acc}). 
This generated samples will be used as condition values at the next iteration. 
We repeat above procedures until we get enough samples (line~\ref{alg:gibbs_sample:fin}).

The above procedure can be finished on each local node in parallel after the distribution parameters and the confidence in each local node were broadcast to others. 
Meanwhile, the network transmission required in this process is far less than that of transmitting real samples. 
Thus, this approach can achieve better performance and scalability.

\noindent\textbf{Generative vs. Distribution-aware Sampling}
We further make a comparison between the two sampling approaches. 
Compared with distribution-aware sampling, the main advantage of the generative sampling method is that it does not need to transmit the concrete samples via the network. 
In this process, the only step requiring network transmission is that every local node broadcasts its distribution parameters and confidence.
The network communication cost of the two sampling approaches are analyzed as follows:
\begin{compactitem}
	\item The communication cost of the distribution-aware sampling is $O(k \cdot (M-1))$.  For each local node, we would sample around $\frac{k}{M}$ objects on average and send $\frac{k}{M}$ local samples to other $M-1$ local nodes. 
	Thus, the communication of each node is $O(\frac{k}{M}\cdot (M-1))$ and the total communication would be $O(\frac{k}{M}\cdot (M-1)\cdot M) = O(k \cdot (M-1))$.
	\item The communication cost of the generative sampling is $O(M\cdot(M-1))$.
	The reason is that for each local node it only needs to send the distribution parameters and types to other $M-1$ local nodes. 
	All samples are just generated on each node without any network communication.
\end{compactitem}
Since $M \ll k$, the cost of the generative sampling is far less than that of the distribution-aware sampling.

\subsection{Error Bound Analysis}\label{subsec-errorbd}

Finally, we make a theoretical analysis on the quality of generative sampling approach.
The goal is to show that unlike the simple random sampling which requires unbounded size of samples to improve the sampling quality,  \emph{our approach has a formal guarantee of the sampling quality w.r.t a given sample size}.

We first give the definition of error between the true data distribution and samples.
Generally speaking, it is difficult to describe the empirical distribution of $m$-dimensional dataset to quantify the error.
Fortunately, during the partition process in map and reduce phases, we only use one dimension in each step of partition.  
Therefore, we can use the CDF of marginal distribution $P(\bm{X})$ to describe the quality. 
It can be defined by the marginal distribution on $\bm{x}$ of $P(\bm{X}, E, C)$:

\begin{equation}
P(\bm{x}) = \int_{\bm{X} \leq \bm{x}} (\iint_{E, C} p(\bm{X}, E, C)\mathrm{d}E\mathrm{d}C )\mathrm{d}\bm{X}
\end{equation}

Similarly, the empirical distribution $\~{P}(\bm{x})$ is defined as:
 \begin{equation}
\~{P}(\bm{x}) = \dfrac{|\{\bm{X} \leq \bm{x}| \bm{X} \in S_{\bm{X}} \}|}{|S_{\bm{X}}|}
\end{equation}
where $S_{\bm{X}}$ is the set of samples.
 	
Specifically, we select the maximum bias between true marginal distribution and empirical marginal distribution on one dimension. 
Following this route, we define the \textbf{sampling error} in Definition~\ref{def-sperr}.

\begin{definition}\label{def-sperr}
	Given the CDF $P(\bm{X})$ of a distribution and the set of samples $S_{\bm{X}}$, the error of sampling can be regarded as:
\begin{equation}
D_k( S_{\bm{X}}, P(\bm{x})) = D_k( \~{P}(\bm{x}), P(\bm{x})) = \max_{d = 1}^m \sup_{\bm{x} \in \mathbb{R}^m}{|\~{P}_{\bm{x}_d}(\bm{x}) - {P}_{\bm{x}_d}(\bm{x})|}
\end{equation}
where $P_{\bm{x}_d}(\bm{x})$($\~{P}_{\bm{x}_d}(\bm{x})$) is the \textbf{marginal distribution} of $P(\bm{x})$($\~{P}(\bm{x})$) on the dimension $d$.
\end{definition} 
With such a definition of sampling error, we can formally obtain a theoretical error bound of our approach w.r.t sampling size $k$ for the dataset with $m$-dimensional objects according to the \emph{Asymptotic Theory for Brownian Motion}~\cite{pRES92a}.
The details are shown in Theorem~\ref{tho-errbd}.
\begin{theorem}\label{tho-errbd}
	Given the sample size $k$, the probability that sampling error exceeds a very small constant $\epsilon$ is less than $2m\cdot e^{-2k\epsilon^2}$, which is formally described as: 
	\begin{equation}\label{eq-errbound}
	\mathbb{P}\{D_k( S_{\bm{X}}, P(\bm{x})) \geq \epsilon\} < 1 - (1 - 2\cdot e^{-2k\epsilon^2})^m \approx 2m\cdot e^{-2k\epsilon^2}
	\end{equation}
\end{theorem}
\begin{proof}
	First of all, the samples $\bm{X} \in S_{\bm{X}}$ from the generative sampling approach can be considered as independent identically distributed due to the property of Gibbs Sampling.
	
	As ${P}_{\bm{x}_d}(\bm{x})$ and $\~{P}_{\bm{x}_d}(\bm{x})$ is as the true distribution and empirical distribution on one dimension respectively, we obtain the following inequalities for any dimension $d$ based on Asymptotic Theory for Brownian Motion~\cite{pRES92a}:  
	\begin{equation}
	\mathbb{P}\{\sup_{\bm{x} \in \mathbb{R}^m}{|\~{P}_{\bm{x}_d}(\bm{x}) - {P}_{\bm{x}_d}(\bm{x})|} < \epsilon\} > 1 - 2\cdot e^{-2k\epsilon^2}
	\end{equation}
	
	In addition, since $D_k( S_{\bm{x}}, P(\bm{x}))$ is the maximum values of \\$\sup_{\bm{x} \in \mathbb{R}^m}{|\~{P}_{\bm{x}_d}(\bm{x}) - {P}_{\bm{x}_d}(\bm{x})|}$ among all dimensions, we can assert that the  event $\{D_k( S_{\bm{x}}, P(\bm{x})) < \epsilon\}$ is the same as  the event $\{(\forall d) \sup_{\bm{x} \in \mathbb{R}^m}{|\~{P}_{\bm{x}_d}(\bm{x}) - {P}_{\bm{x}_d}(\bm{x})|} < \epsilon \}$. 
	Since the errors among different dimension are independent, we get the following expression based on the Chain Rule of probability:
	
	\begin{equation}
	\begin{split}
	\mathbb{P}\{D_k( S_{\bm{x}}, P(\bm{x})) < \epsilon\} &= \prod_{d=1}^m \mathbb{P}\{\sup_{\bm{x} \in \mathbb{R}^m}{|\~{P}_{\bm{x}_d}(\bm{x}) - {P}_{\bm{x}_d}(\bm{x})|} < \epsilon\} \\
	&> (1 - 2\cdot e^{-2k\epsilon^2})^m
	\end{split}
	\end{equation}
	
	Also, due to the fact that sample size $k$ is always larger than the dimension $m$, the value of $2\cdot e^{-2k\epsilon^2}$ is rather small compared with $m$ in the binomial expression, we can derive the following formula: 
	
	\begin{equation}
	(1 - 2\cdot e^{-2k\epsilon^2})^m  > 1 - 2m\cdot e^{-2k\epsilon^2}
	\end{equation}
	
	Finally, since we have:
	\begin{equation*}
	\mathbb{P}\{D_k( S_{\bm{X}}, P(\bm{x})) \geq \epsilon\} = 1 - \mathbb{P}\{D_k( S_{\bm{X}}, P(\bm{x})) < \epsilon\}
	\end{equation*}
	we get the theorem, which completes the proof.
\end{proof}

With the guarantee provided by Theorem~\ref{tho-errbd}, the hyper-parameter of sample size $k$ can be determined according to the level of errors $\epsilon$ that can be tolerated in the sampling phase.
Meanwhile, previous studies just decide it empirically without any guideline.
As a result, they have no choice but to increase the sample size in order to improve the quality of sampled pivots.
Such a formal guarantee servers as a concrete evidence why our method would be definitely more superior than the random sampling approach adopted by previous studies.

%% file: src/sec5-partition.tex
\section{Data Partition Scheme}\label{sec-partition}

In this section, we propose an effective partition scheme based on the sampled pivots in map and reduce phases.
We first theoretically analyze the methodology of the partition to provide a guideline for proposed schemes (Section~\ref{subsec-simplepart}). 
Then we propose two partition strategies: an iterative one (Section~\ref{subsec-basicmap}) and a learning-based one (Section~\ref{subsec-advmap}) to reduce the overhead and make even partitions. 

\subsection{Preliminaries of Partition}\label{subsec-simplepart}

In the map phase, we need to split the dataset into partitions according to the sampled pivots and shuffle the intermediate results to the reducers. 
To reach this goal, we first split the overall space into $p$ areas $\bigP = \{P_1, P_2,\cdots, P_p \}$ according to the $k$ pivots ($p \ll k$).
Based on each area $P_i \in \bigP$,  we construct the corresponding \emph{partition} $\outeri{i}$, which then will be shuffled to a reducer. 
As each $\outeri{i}$ contains all objects on a particular reducer, we call it a \outerpart of the overall dataset $\bigR$.
The task of a partition strategy is to generate a series of \outerpart, i.e. $\bigW = \{\bigW_1,\bigW_2,\cdots, \bigW_{p}\}$. 
While many ways can be explored to evaluate the quality of partition, a particular reasonable one is to use the \emph{total number of verifications among all reducers}.
Thus, we quantify the cost of partition strategies in this way.
We then make a theoretic analysis on how to reduce such cost.

For a particular object $o_i$ and partition $\outeri{h}$, we use a matrix $\NPParas \in \mathbb{R}^{N_\bigR \times p} $ to denote whether $o_i$ belongs to $\outeri{h}$.
\begin{equation}
	\NPPara{i}{h} = \mathrm{\mathbf{1}}_{\{o_i \in \bigW_h\}}
\end{equation}
Then $\outeri{h}$ can be defined as $\outeri{h} = \{ o_i | \NPPara{i}{h}  = 1 \}$.\\

\noindent\textbf{Cost model}. 
With the help of $\NPParas$,  we then give the cost of partition strategies with the following function: 
\begin{eqnarray}\label{equ-cost}
\NPObjF(\NPParas) & = & \mathbb{1}^{\mathrm{T}} \cdot \NPParas \cdot \NPParas^{\mathrm{T}} \cdot \mathbb{1}\nonumber \\ 
& = &\sum_{i,j}\sum_{h}{\NPPara{i}{h} \cdot \NPPara{j}{h}}\nonumber \\
& = &\sum_h\sum_{i,j}{\NPPara{i}{h} \cdot \NPPara{j}{h}} 
\end{eqnarray}

Here $\NPPara{i}{h} \cdot \NPPara{j}{h}$ would be 1 iff. both $o_i$ and $o_j$ locate in the same partition $\bigW_h$.
According to the definition of \outerpart, one object could appear in different partitions, and there would be a heavier cost (more verifications) if the same pair of objects appear on several different partitions.
Thus \emph{the objective for devising partition strategy is to minimize the value of $\NPObjF(\NPParas)$}.

\noindent\textbf{Correctness of partitioning}. Before talking about how to minimize the cost, we first need to guarantee the correctness of the partitioning.
A partition strategy is correct iff the answers aggregated from all the reducers are identical to the correct join results.
To guarantee this, we need to add some constraints on the basis of Equation~\ref{equ-cost}.
The essence of correctness is that any pair of objects $o_i, o_j \in \bigR$ s.t. $\bigD(o_i, o_j) \leq \delta$ must be verified by at least one reducer.

To formally express this relationship, we make the following notations.
Firstly, we perform matrix multiplication and obtain $\NPParas' = \NPParas \cdot \NPParas^{\mathrm{T}}$, where $\NPParas'(i,j)$ means the number of verifications for pair $\langle o_i, o_j\rangle$ among all reducers.
Secondly, we define the matrix $\NPBounds \in \mathbb{R}^{N_\bigR \times N_\bigR}$ to denote whether the distance between two objects is larger than $\delta$, i.e. whether a pair belongs to the correct join result.
\begin{equation}
\NPBound{i}{j} = \mathrm{\mathbf{1}}_{\{\bigD(o_i, o_j)\leq\delta\}}
\end{equation}

Based on the above analysis, we can formally define the constraint for correctness as: \emph{Given the threshold $\delta$, the set of partitions $\bigW$ should satisfy $\NPParas \cdot \NPParas^{\mathrm{T}} \geq \NPBounds$}. \\

\noindent\textbf{Hardness of minimizing the partition cost}.
Then the problem becomes minimizing the value of $\NPObjF(\NPParas)$ under the constraint of correctness. 
Unfortunately, we find this problem is NP-hard, which is formally shown in Theorem~\ref{the-costnp}.
\begin{theorem}\label{the-costnp}
	The problem of minimizing $\NPObjF(\NPParas)$ with the constraint of $\NPParas \cdot \NPParas^{\mathrm{T}} \geq \NPBounds$ is NP-hard.
\end{theorem}
\begin{proof}
	We will try to rewrite the formula into the form of\\ \textbf{Q}uadratically \textbf{C}onstrained \textbf{Q}uadratic \textbf{P}rogram (QCQP).
	To show this process, we first provide some helper symbols:
	\begin{definition}[KRONECKER PRODUCT]
		Let matrix $\Matrix{X} \in \mathbb{R}^{s \times t}$, matrix $\Matrix{Y} \in \mathbb{R}^{p \times q}$. Then the Kronecker product of $\Matrix{X}$ and $\Matrix{Y}$ is defined as the matrix:
		\begin{equation*}
		\Matrix{X} \otimes \Matrix{Y} = \left[
		\begin{matrix}
		\Matrix{X}(1, 1) \Matrix{Y} & \cdots & \Matrix{X}(1, t) \Matrix{Y} \\
		\vdots  & \ddots & \vdots \\
		\Matrix{X}(s, 1) \Matrix{Y} & \cdots & \Matrix{X}(s, t) \Matrix{Y} \\
		\end{matrix}
		\right]
		\end{equation*}
		the Kronecker product of two matrices $\Matrix{X}$ and $\Matrix{Y}$ is a $sp \times tq$ matrix.
	\end{definition}	
	
	\begin{definition}
		Let matrix $\Matrix{Z}(\cdot, i) \in \mathbb{R}^s $ denotes the columns of $\Matrix{Z} \in \mathbb{R}^{s\times t}$ so that $\Matrix{Z} = [\Matrix{Z}(\cdot, 1), ..., \Matrix{Z}(\cdot, t)]$. Then $vec(\Matrix{Z})$ is defined to be the $st$-vector formed by stacking the columns of $\Matrix{Z}$ on top of one another, i.e., 
		\begin{equation*}
		vec(\Matrix{Z}) = \left[
		\begin{matrix}
		\Matrix{Z}(\cdot, 1) \\
		\vdots \\
		\Matrix{Z}(\cdot, t)
		\end{matrix}
		\right] \in \mathbb{R}^{st}
		\end{equation*}
	\end{definition}
	Then we could rewrite our objective function into the quadratic forms using above operators.
	\begin{eqnarray}
	\NPObjF(\NPParas)&=&\mathbb{1}_{N_\bigR}^{\mathrm{T}} \cdot \NPParas_{N_\bigR \times p} \cdot \NPParas_{N_\bigR \times p}^{\mathrm{T}} \cdot \mathbb{1}_{p} \\ 
	&=&vec(\NPParas^{\mathrm{T}})^{\mathrm{T}} \cdot (\mathbb{1}_{N_\bigR \times N_\bigR} \otimes \mathbb{E}_{p \times p}) \cdot vec(\NPParas^{\mathrm{T}})
	\end{eqnarray}
	which can be considered as the quadratic objective functions.
	
	Next we define a series of Matrix $\prescript{i,j}{}{\mathbb{K}}$ where $\prescript{i,j}{}{\mathbb{K}}(i, j) = 1$ and other elements in $\prescript{i,j}{}{\mathbb{K}}$ are $0$. 
	Then define Matrix $\prescript{i,j}{}{\mathbb{P}} = \prescript{i,j}{}{\mathbb{K}}_{N_\bigR \times N_\bigR} \otimes \mathbb{E}_{p \times p}$.
	Therefore, we can also rewrite the constraint into quadratic forms as following:
	\begin{equation*}
	vec(\NPParas^{\mathrm{T}})^{\mathrm{T}} \cdot \prescript{i,j}{}{\mathbb{P}} \cdot vec(\NPParas^{\mathrm{T}}) - \NPBound{i}{j} \geq 0 , \forall~i,~j \in \mathbb{N}^{N_\bigR}
	\end{equation*}
	Besides, since each element in $vec(\NPParas)$ should be 1 or 0, the constraint can be written as:
	\begin{equation*}
	vec(\NPParas^{\mathrm{T}})^{\mathrm{T}} \cdot \prescript{i,j}{}{\mathbb{K}} \cdot vec(\NPParas^{\mathrm{T}}) - vec{\prescript{i,j}{}{\mathbb{K}}}^{\mathrm{T}} \cdot vec(\NPParas^{\mathrm{T}}) = 0
	\end{equation*}
	Therefore, we reduce our problem to \textbf{QCQP}, which has been proven to be NP-hard, which can be verified by easily reducing to the well-known 0-1 integer programming problem.
\end{proof}	

Therefore, instead of finding an optimal partition scheme, we aim to find some heuristic approaches to get a good partition.
The key insight comes from the observation that one \outerpart consists of two parts: the first part is the set of objects that is uniquely owned by one \outerpart while the second one has overlap with other partitions.
Here we denote the first part as \innerpart, which could be utilized to find a way to reduce the cost.

The \innerpart can be identified with a matrix $\NPParasAdv \in \mathbb{R}^{N_\bigR \times p}, \NPParaAdv{i}{h} \in {0, 1}$. 
Similar with $\NPPara{i}{h}$, the value of $\NPParaAdv{i}{h}$ denotes whether object $o_i$ belongs to $\outeri{h}$. 
Meanwhile, we put more constraints on $\NPParasAdv$: each row of it can only have one element with value 1, and only if the value of $\NPPara{i}{h}$ is 1, we can have $\NPParaAdv{i}{h} = 1$. 
Then we can formally define \innerpart in Definition~\ref{def-inner}.
 
\begin{definition}\label{def-inner}
	Given a matrix $\NPParasAdv$, we define the set of
	\begin{equation}
	\inneri{h} = \{ o_i | \NPParaAdv{i}{h} = 1 \}
	\end{equation} 
	as \innerpart.
\end{definition}

\begin{lemma}\label{def-inout}
	Given a dataset $\bigR$, the \innerpart $\mathcal{V}$ and \outerpart $\mathcal{W}$ satisfying: \\
	
	(1) ${\cup}_{\inneri{i} \in\mathcal{V} } \inneri{i}= \bigR$; $i\neq j \rightarrow \inneri{i} \cap \inneri{j} = \emptyset$ ; 
	
	(2) $\forall \outeri{i} \in \mathcal{W}, \outeri{i} \subseteq \bigR$;
	
	(3) ${\biguplus}_{\inneri{i} \in\mathcal{V}, \outeri{i} \in\mathcal{W}} Reduce(\inneri{i} \cup \outeri{i}) = \bigJ$, where $\bigJ$ is the set of all correct join results.
\end{lemma}

\begin{figure}[h]\vspace{-.5em}
	\centering
	\includegraphics [scale=0.20]{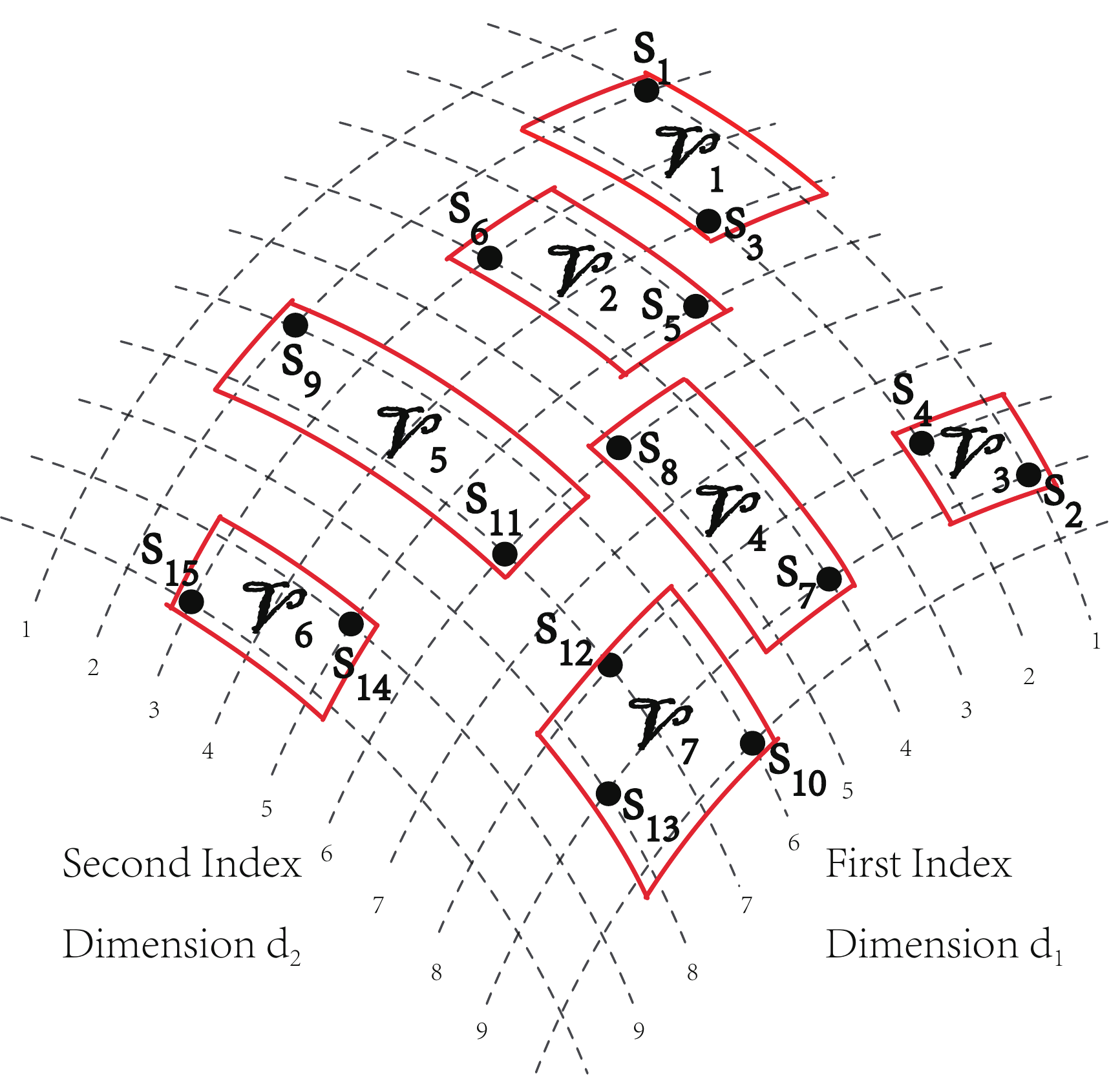}\vspace{-.25em} \hspace{.75em}
	\includegraphics [scale=0.20]{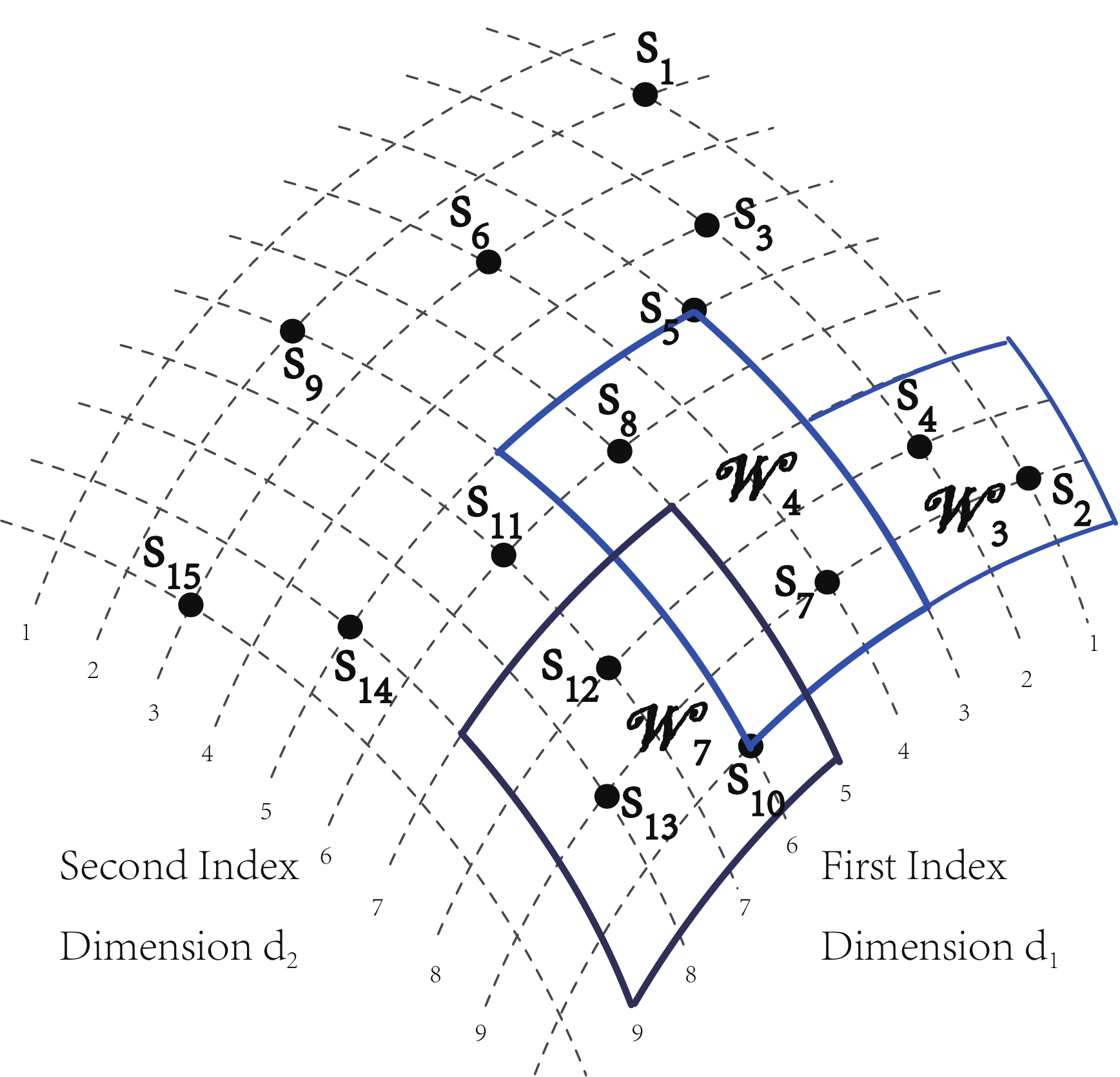}\vspace{-.25em} \hspace{-.75em}
	\caption{The \innerpart and \outerpart}
	\label{fig-inner-outer}\vspace{-1em}
\end{figure}

 We show a running example in Figure~\ref{fig-inner-outer}. All red hyper-rectangles $\inneri{i}$ are in Kernel Partition, and no two red hyper-rectangles have overlaps. Besides, the union of them contains all points shown in the figure. In order to get \outerpart, we expand $\inneri{i}$ with one threshold to get a large hyper-rectangle $\outeri{i}$ in right subfigure denoted by those blue rectangles (We ignored some blue rectangles $\outeri{i}$ in the figure~\ref{fig-inner-outer} for the ease to display).\\

\noindent\textbf{Rewritten cost model}. With the definitions of \innerpart and \outerpart, we can rewrite the cost model in Equation~\ref{equ-cost} as Equation~\ref{equ-newcost}.
\begin{equation}\label{equ-newcost}
\begin{array}{lll}
\NPObjF(\NPParas) & = & \underbrace{\sum_h |\inneri{h}| \cdot |\inneri{h}|}_{inner\ verification\ cost} + \underbrace{\sum_h |\inneri{h}| \cdot (|\outeri{h}| - |\inneri{h}|)}_{outer\ verification\ cost}
\end{array}
\end{equation}
We denote the two parts of Equation~\ref{equ-newcost} as \emph{inner verification cost} and \emph{outer verification cost}, respectively. In the following, we propose two partition strategies: an iterative method (Section~\ref{subsec-basicmap}) to reduce the \emph{inner verification cost} and a learning-based method (Section~\ref{subsec-advmap}) to reduce the \emph{outer verification cost}.

\subsection{Iterative Partition}\label{subsec-basicmap}

In this section, we aim at reducing the inner verification cost, i.e. $\sum_h |\inneri{h}| \cdot |\inneri{h}|$. 
According to the fundamental inequality in Equation~\ref{equ-fdieq}, the cost would be minimized when each \innerpart has the same number of objects since $\sum_h |\inneri{h}| = N_\bigR$ .
\begin{equation}\label{equ-fdieq}
{\sum_h |\inneri{h}|^2}\geq {p} \cdot (\dfrac{\sum_h |\inneri{h}|}{p})^2 = \dfrac{N_\bigR^2}{p}
\end{equation}
 
Therefore, in order to reduce the inner verification cost, we should make all partitions with similar cardinality. 
To this end, we devise the following iterative method to split the dataset \textbf{evenly} in the map phase.
Firstly, in order to utilize the geometric and coordinate information that is unavailable in the metric space, we first perform \emph{Space Mapping} to project all sampled pivots in $\bigS = \{s_1, s_2, \cdots, s_k\}$ into an Euclidean Space with $n$-dimensions, where $n$ is a tunable parameter. 
Here we call the space after mapping \emph{target space} and the space for original similarity metric \emph{origin space}. 
In this way, it will be easy to determine \outerpart and \innerpart with the guarantee of correctness. 
Next we iteratively divide the target space into $p$ areas using the sampled pivots. 
Then we map the remaining objects other than pivots into the target space and allocate them into corresponding areas to decide the \innerpart and \outerpart. 
Finally, all generated partitions  are shuffled to reducers.

We then introduce every step in details. 
In the first step, we need to obtain a set of \emph{dimensional pivot}s, denoted as $\bigA = \{a_1, a_2,\cdots, a_n\}$ to help map a pivot $s_i$ into its \emph{mapped representation} $s_i^n$ with metric distance $\bigD: \bigU \times \bigU \rightarrow \Re^n$:
$s_i^n = (\bigD(a_1, s_i), \bigD(a_2, s_i), ..., \bigD(a_n, s_i))$\\
such a set \bigA can be randomly sampled from the $k$ pivots.
Note that the space mapping will not change the partition an object belongs to. 
Thus once the pivots are evenly distributed in the target space, we can obtain even partitions in origin space.

To split the target space into $p$ areas, we iteratively select a dimension from $[1,n]$ and split the space into two disjoint areas at the median. 
In this process,  as $p$ is not always equal to the power of 2, we can decide whether to split an intermediate area according to the value of $p$. 
For each area $P_i$, we use the Minimum Bounding Box \Bi{i} of $P_i$ which can include all objects in the area to represent the space it occupies. 
We can decide the boundaries of \Bi{i} using the pivots located in $P_i$. 
The minimum and maximum values of the $j^{th}$ dimension of $\Bi{i}$ are denoted as $\Bi{i}^{\bot}[j]$ and $\Bi{i}^{\top}[j]$ respectively. 
And we can denote the boundary of \Bi{i} using its hyper-perimeter denoted as $\prod_j [\Bi{i}^{\bot}[j], \Bi{i}^{\top}[j]]$.
\begin{figure}[h]
	\begin{center}
		\subfigure[\small{Point Distribution}]{
			\label{subfig-pointdist}
			\hspace*{-1.0em}\epsfig{figure=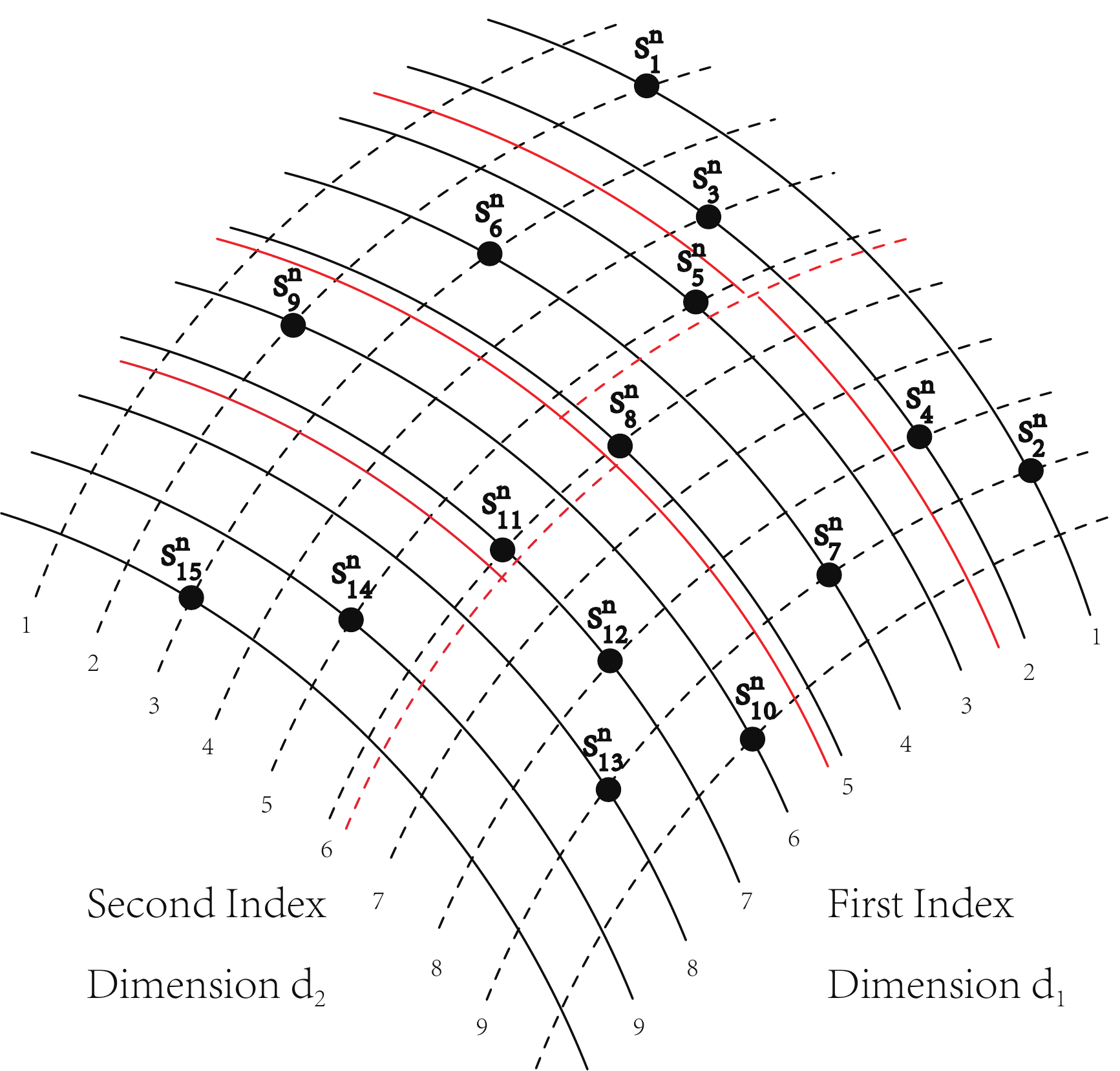,width=0.5\linewidth}}
		\subfigure[\small{Partition Results}]{
			\label{subfig-iterpart}
			\hspace*{-1.0em}\epsfig{figure=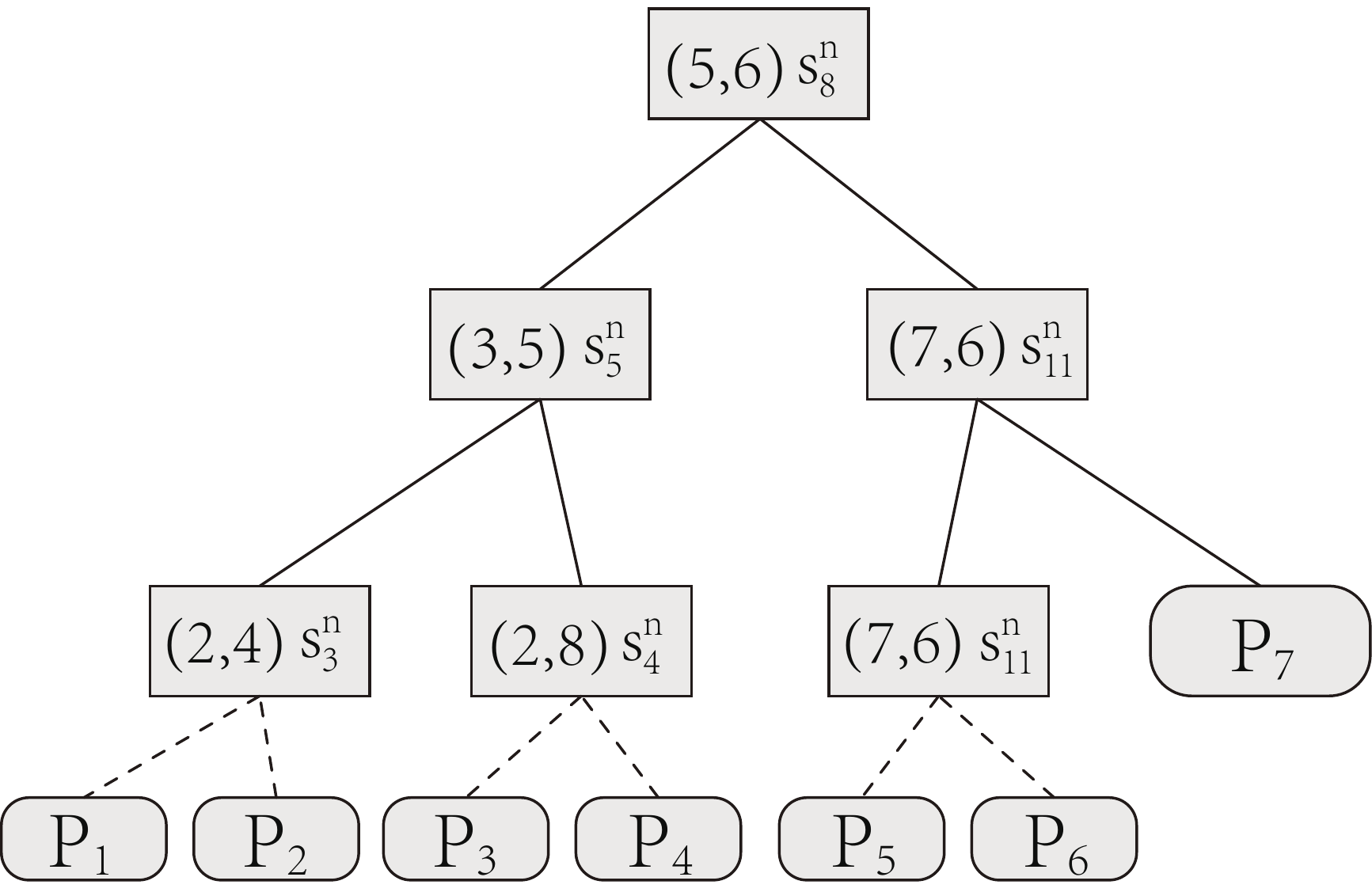,width=0.5\linewidth}}
	\end{center}\vspace{-1em}\vspace{-.5em}
	\caption{The Iterative Partition Method}\label{fig:iter}\vspace{-1em}
\end{figure}

\begin{example}	
Suppose after mapping to target space, we obtain the pivots shown in Figure~\ref{subfig-pointdist}. For ease of display, here we just reduce to 2-dimensional Euclidean Space. We want to split the area into 7 partitions. First we randomly select a dimension, e.g. the first one for splitting. Then we find $s_5^n$ as the fractile, then we split the whole space into two parts represented by the red line nearby $s_5^n$. We repeat this process iteratively with $s_8^n$ and $s_{11}^n$, respectively. Finally, we could get the areas shown in Figure~\ref{subfig-iterpart}.
\end{example}

With the help of target space, we propose the iterative partition method in Algorithm~\ref{alg:basictree}.
We first sort the pivots according to a randomly selected dimension $d$ (line~\ref{alg:basictree:filter}). 
Then we split the mapped pivots into two parts according to the fractile of selected dimension (line~\ref{alg:basictree:split}). 
When we obtain a new area $P_i$, we should first decide whether other pivots belong to this area, then collect them to calculate their Minimum Bounding Box \Bi{i} to represent that area. 
The above process is repeated iteratively until we get $p$ areas.

\begin{figure}[!t]\vspace{-.5em}
	\linesnumbered \SetVline
	\begin{algorithm}[H]
		\caption{Iterative Partition ($S^n$, $p$) \label{alg:basictree}}
		\KwIn{$S^n$: The set of pivots after mapping, $p$: The number of areas}
		\KwOut{$\bigP$: The set of areas}
		\SetVline
		\Begin{
			\If{$p$ is 1}{
				\return leaf node with new index \\
			}
			\Else{
				Sort $S^n$ by a randomly selected dimension $d$\nllabel{alg:basictree:filter} \\
				Let $d_m$ = $\dfrac{\lceil p/2\rceil}{\lfloor p/2\rfloor}$ fractile of $S^n$ in dimension $d$ \nllabel{alg:basictree:split} \\
				Initialize $S_l^n , S_r^n$ as empty sets \\
				\For {each $s^n_i$ in $S^n$} {
					\If { $s^n_{i}[d] < d_m$ }{
						Put $s^n_i$ into $S_l^n$
					}
					\lElse {
						Put $s^n_i$ into $S_r^n$
					}
				}
				$\bigP_l = Iterative Partition(S_l^n, \lceil p/2\rceil)$ \\
				$\bigP_r = Iterative Partition(S_r^n, \lfloor p/2\rfloor)$ \\
				Build $\bigP$ with $\bigP_l$ and $\bigP_r$ \\
				\return \bigP
			}
		}
	\end{algorithm}\vspace{-.5em}
\end{figure}

\begin{figure}[!t]
	\hspace{-1.1em}
	\centering
	\includegraphics [scale=0.3]{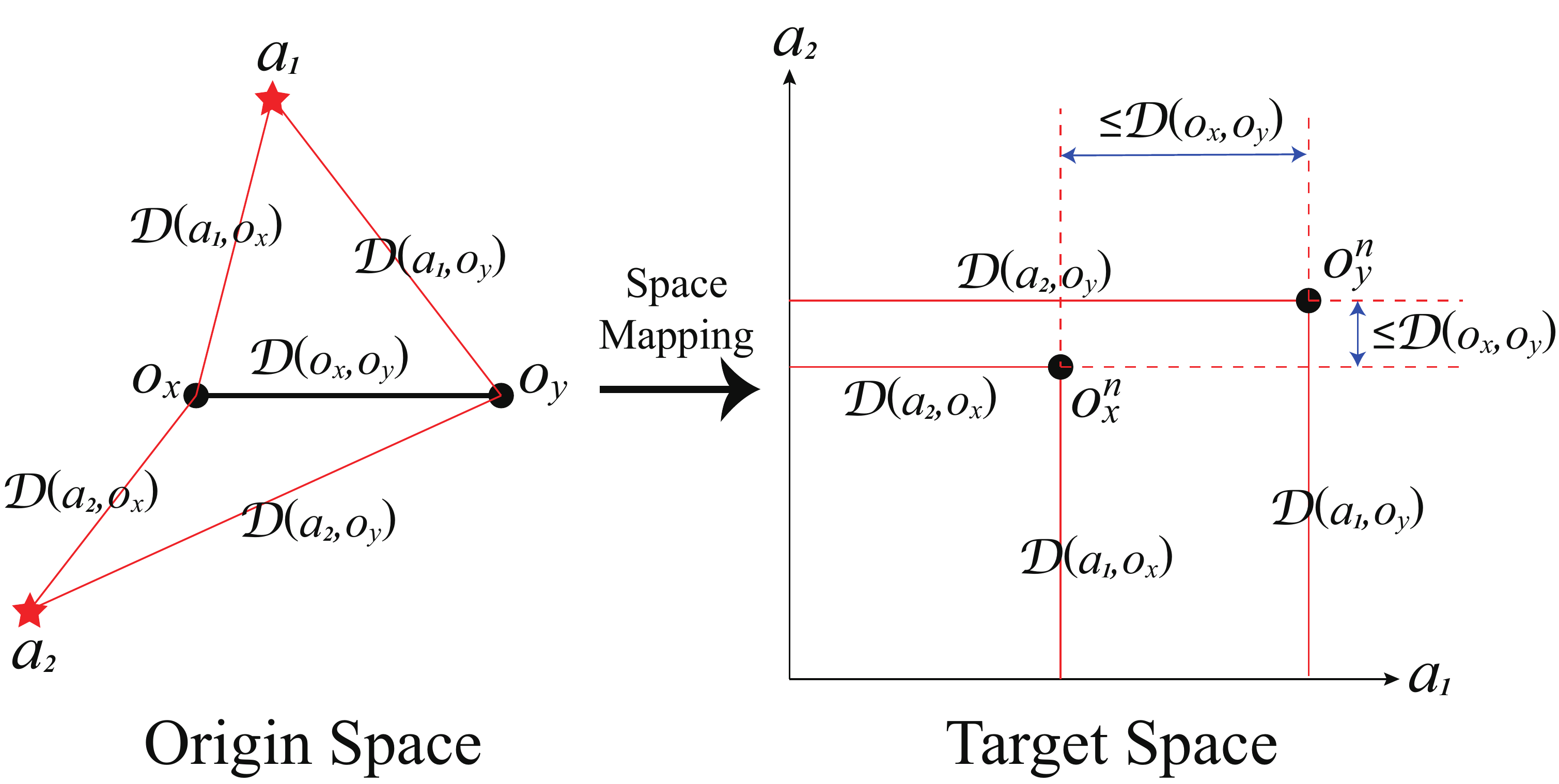}
	\caption{Correctness of Iterative Partition}
	\label{fig-spacemapping}\vspace{-.5em}
\end{figure}
After obtaining the set of areas $\bigP$, we then need to identify the \innerpart and \outerpart for each $P_i \in \bigP$.
As the pivots have been allocated into different areas in the process of iterative partition, we next map the remaining objects into target space and allocate them to corresponding partitions.
For an area $P_i$, its kernel partition $\inneri{i}$ is the collection of objects whose mapped representations are within the hyper-cube subspace $\Bi{i}$, i.e. $\inneri{i} = \{o|o^n\in \Bi{i}\}$.
Similarly, given the threshold $\delta$, its whole partition $\outeri{i}$ is the collection of objects whose mapped representations are contained in a hyper-cube with range $ \prod_d[\Bi{i}^{\bot}[d] - \delta, \Bi{i}^{\top}[d]+\delta]$, which is formally defined as $\outeri{i} = \{o|o^n\in \prod_d[\Bi{i}^{\bot}[d]-\delta, \Bi{i}^{\top}[d]+\delta]\}$, where $o^n$ is the representation forms of object $o$ in target space. 
With such processing, we can also guarantee that all the join results can be found using $\outeri{i}$ and $\inneri{i}$, which is formally stated in Lemma~\ref{lem-inout}.

\begin{lemma}\label{lem-inout}
For any two objects $ o_x, o_y \in \bigR$ s.t. $\bigD(o_x, o_y) < \delta$, there is a partition $k$ s.t. $o_x \in \inneri{k}$ and $o_y \in \outeri{k}$.
\end{lemma}

\begin{proof}
First we show that using $\delta$ to expand the subspace $\Bi{i}$ in target space is reasonable. 
For any pair of objects $o_x, o_y \in \bigR$ s.t. $\bigD(o_x, o_y) < \delta$ and $o_x^n, o_y^n$, we have $|\bigD(a_i, o_x) - \bigD(a_i, o_y)| \leq \bigD(o_x, o_y) < \delta$ for $i \in [1, n]$ because of the \emph{Triangle Inequality} property that holds for all distance functions in \emph{origin space} as the left of figure~\ref{fig-spacemapping}). 
In other word, after Space Mapping in \emph{target space}, we can deduce that $|o_x^n[i] - o_y^n[i]| \leq \bigD(o_x, o_y) < \delta$ for any dimension $i$ as the right of figure~\ref{fig-spacemapping}. 
Then, according to the first property in Lemma~\ref{lem-inout}, $o_x$ must exist in one and only one \innerpart $\inneri{k} = \{o|o^n\in \Bi{k}\}$. 
Meanwhile, from the definition of its corresponding \outerpart and the fact that $|o_x^n[d] - o_y^n[d]| < \delta$ is established for any dimension $i$, we can have: $o_y \in \outeri{k} = \{o|o^n\in \prod_d[\Bi{k}^{\bot}[d]-\delta, \Bi{k}^{\top}[d]+\delta]\}$. 
The proof completes.
\end{proof} 

\subsection{Learning-based Partition}\label{subsec-advmap}

Although the iterative partition method can reduce the inner verification cost, it cannot effectively bound the size of \outerpart.
This is because it only relies on the information in target space and loses some information in the origin space, which may lead to skewness in partitions. 
To alleviate this problem, we devise a learning-based method to reduce the outer verification cost i.e. $|\inneri{h}| \cdot (|\outeri{h}| - |\inneri{h}|)$. 
The key point is that for a given number $p$ of areas, in each iteration we should find a proper dimension with which the pivots can distribute evenly in the origin space. 
To this end, we define a cost function to measure the \emph{compactness} of areas.
The less compactness an area has, the larger radius a node will have, which could lead to an irregular subspace of \innerpart. 
Thus the size of corresponding \outerpart will be larger.
Therefore, we can reduce the size of \outerpart by trying to get partitions with higher degree of compactness. 

To identify the compactness, we first need to recognize similar pivots in the origin space. 
To this end, we assign a label to each $s_i \in \bigS$ so that pivots with the same label are close to each other in origin space. 
We call pivots with the same label ``similar'' pivots.
The labeling function is formally defined in Definition~\ref{def-lbfnc}.
\begin{definition}\label{def-lbfnc}
		A \textbf{Labeling Function} on the set of pivots $\bigS$ is a function $L$: $\bigS$ $\rightarrow$ $N$ s.t.
		\begin{compactitem}
			\item non negative: $\forall s_i \in \bigS, L(\bigS) \geq 0$
			\item finite co-domain: $|L(\bigS)|< +\infty$
		\end{compactitem}
\end{definition}
where the labeling function $L(\cdot)$ can be easily found with clustering techniques. 
That is, we allocate similar pivots into clusters and assign a distinct label to each cluster. 
As \emph{hierarchical clustering} is a universal way to adapt our techniques to various types of distance functions, here we adopt it to fulfill the task of labeling. 
In this way, the similar pivots will be assigned the same label.
	
With the help of assigned labels, we can use the labels of pivots to evaluate the compactness of areas and then utilize a learning-based technique to select the dimension for splitting in each iteration. 
Specifically, we use $Cost(\bigS,L)$ to quantify the compactness of an area. 
The lower value it is, the more pivots with the same labels are in the area, and the better compactness there will be.
Given $\bigS$ and $L$, $\forall y \in L(\bigS)$, the frequency of label $y$ is denoted as $freq(y) = |\{s_i|L(s_i) = y\}|$. 
The cost value is defined as the proportion between the frequency of label $y$ and the total number of labels with a logarithm regularization.
\begin{equation}\label{equ-gain}
Cost(\bigS,L) = \int\limits_{y \in L(\bigS)}{\dfrac{freq(y)}{|\bigS|}\\
			\cdot (- \log{ \dfrac{freq(y)}{|\bigS|}})} \mathrm{d} y
\end{equation}
		
With the help of Equation~\ref{equ-gain}, we then select the dimension for splitting by calculating the \emph{cost variation} in each iteration. 
Specifically, in current iteration if we split the space by dimension $d$ ($d \in [1,n]$), we can get a set of subspaces, denoted as $K_d$. 
The set of $K_d$ is denoted as $\bigK$.
And the cost variation on dimension $d$, denoted as $\bigC_d$, can be calculated as:
\begin{equation}
		\bigC_d(\bigS, \bigK, L) = Cost(\bigS, L) - \int\limits_{K\in\bigK}\dfrac{|K|}{|\bigS|} \cdot Cost(K, L)\mathrm{d}K
\end{equation}
This cost variation illustrates whether there will be better compactness after splitting. 
The larger value it is, the more percentage of pivots with the same label will be in the areas after splitting.

Therefore, we can calculate $\bigC_d(\bigS, \bigK, L)$ for each dimension $d$ and select the one with the maximum value.
To avoid overfitting, we also apply logarithm regularization on the cost variation to compensate for deviations instead of directly using cost variation when proposing a measure function $\bigF_d$, which is detailed in Equation~\ref{equ-mfunc}.
\begin{equation}\label{equ-mfunc}
	\bigF_d(\bigS, \bigK, L) = \dfrac{	\bigC_d(\bigS, \bigK, L)}{\int\limits_{K \in \bigK}{\dfrac{|K|}{|\bigS|} \cdot (-\log{\dfrac{|K|}{|\bigS|}})}\mathrm{d}K}
\end{equation}
	
We can use this measure to rank dimensions and select one with the largest value of $\bigF_d(\cdot, \cdot, \cdot)$ among all dimensions for splitting. 
In this process, the compactness for all areas will be increased after each iteration. 
Correspondingly, the total size of all whole partitions will tend to be smaller. 
As a result, this approach will improve the quality of partition and reduce the outer verification cost.
We can implement the learning-based method by simply replacing line 5 in Algorithm~\ref{alg:basictree}. 
Details are shown in Algorithm~\ref{alg:mltree}.

\begin{figure}[h]\vspace{-.5em}
	\linesnumbered \SetVline
	\begin{algorithm}[H]
		\caption{Learning-based Partition ($\bigS^n$, $p$) \label{alg:mltree}}
		\KwIn{$\bigS^n$: The dataset of sampling data at target space with type, $p$: The number of areas}
		\SetVline  //replace line 5 in Algorithm~\ref{alg:basictree} \\
		\Begin{
			Let $max\_d, max\_gain = 0$ \\
			\For {each dimension $d$}{
				Sort $\bigS^n$ by $d$; \\
				Split $\bigS^n$ into two parts with fractile $m$ by $d$ into $\bigS_l = \{s|s\in \bigS \wedge s^n[d] < m\}$ and $\bigS_r = \{s|s\in \bigS \wedge s^n[d] \geq m\}$ \\
				Let $gain = \bigF_d(\bigS^n,\{\bigS_l, \bigS_r\}, L_c)$\\
				\If{$gain > max\_gain$ }{$max\_gain = gain, max\_d = d$}
			}
			Use dimension $max\_d$ to sort the pivots;\\
		}
	\end{algorithm}\vspace{-0.5em}
\end{figure}

\subsection{Complexity of Partition Algorithms}\label{subsec-complexity}

We make an analysis on the time and space complexity of the two partition strategies.

For the iterative method (Algorithm~\ref{alg:basictree}), the complexity of selecting the dimension for splitting is $\bigo(nk)$ since we need to calculate the variance and correlation by traversing $S^n$ and sorting $S^n$ in $\bigo(k\log k)$ time in each layer of recursion. Meantime, we have $O(\log p)$ layers of recursion considering the termination condition. Thus the time complexity of Iterative Partition would be $\bigo(\log p\cdot(k\log k+nk))$.

For the Learning-based method (Algorithm~\ref{alg:mltree}), first we need to sort $S^n$ once in each dimension per layer of recursion, thus the time complexity is $O(nk\log k)$. 
Meanwhile, calculating the measure function $\bigF$ needs $\bigo(k)$ time if $\bigS^n$ is sorted as we only need to traverse the $\bigS^n$ and the count number of difference labels. 
Same to Algorithm~\ref{alg:basictree}, the recursion has $\bigo(\log p)$ layers. 
Thus, the total time complexity of Learning-based method is $\bigo(\log p\cdot nk\log(k))$.

%% file: src/sec6-discussion.tex
\section{Discussion}\label{sec-dis}

\subsection{More about Sampling Error Bound}\label{subsec-dis-sample}

First we argue that similar with the generative sampling approach, there is also quality guarantee of the distribution-aware one.
As is introduced before, the distribution-aware sampling approach fetches samples from the local nodes with the help of local distribution. 
In this process, we will get the similar bound mentioned in Theorem~\ref{tho-errbd} with the local sample size $lc$.
The reason is that the distribution-aware approach can be considered as the stratified sampling on the global distribution which is constructed with a mixture of local distribution. 
Therefore the global error bound also can be obtained as the combination of the error bound of each local node.
Specifically, it can be calculated by selecting the maximum error among all local nodes.
Actually, the worst case of using the maximum error as the error bound only influences the performance seriously when the data distribution is extremely skewed.
And it seldom happens in practice.

\subsection{Support Metrics for String and Set}\label{subsec-dis-str}

Though the sampling and partition techniques are designed for vector data, it is natural to apply them to similarity metrics for string and set data.
The reason is that we can transform a string or set to a dense vector with existing techniques such as the ordering methods proposed in~\cite{DBLP:conf/sigmod/ZhangHOS10,DBLP:conf/icde/ZhangLWZXY17,zhang2018transformation}.
Previous studies have proposed many filter techniques for particular distance functions. 
Compared with them, our \name is a general framework which aims at making optimization for various kinds of distance functions. 
Moreover, existing filtering techniques for string and set data can be seamlessly integrated into the reduce phase of our framework to further improve the performance. 
Such integration is straightforward: for the objects on each reducer, we can just regard it as an independent dataset and apply the existing techniques designed for a single machine, such as length filter, prefix filter~\cite{DBLP:conf/icde/ChaudhuriGK06}, position filter~\cite{DBLP:conf/www/XiaoWLY08} and segment filter~\cite{DBLP:journals/vldb/YuWLZDF17}. 
This filtering process can be finished within one reducer, which does not need any network communication.

%% file: src/sec7-evaluation.tex
\section{Evaluation}\label{sec-exp}

\begin{table*}  
	\caption{Statistics of Datasets}\vspace{-1em}
	\centering  
	\subtable[Vector Data]{  
		\centering
		\label{tbl:vec-datastat}
		\begin{tabular}{|c|c|c|c|}\hline
			\multirow{2}*{Dataset} & \multirow{2}*{\#  ($10^5$)}  & \multirow{2}*{Length} & \multirow{2}*{Metric}\\ & & &  \\ \hline
			\vdatao & 4.2  & 20 & \eu \\ \hline
			\vdatat & 10.0  & 128 & \lonenorm  \\ \hline
		\end{tabular}
	}  
	\hfill
	\subtable[String Data]{     
		\centering     
		\label{tbl:str-datastat}
		\begin{tabular}{|c|c|c|c|c|c|c|}\hline
			\multirow{2}*{Dataset} & \multirow{2}*{\tabincell{c}{ \# \\ ($10^6$)}} & \multicolumn{2}{c|}{Rec Len.} & \multicolumn{2}{c|}{Token ($10^6$)} & \multirow{2}*{Metric}\\ \cline{3-6}
			& & Max & Avg & Size & $\mathrm{Freq_{max}}$ & \\ \hline
			\sdatao& 10.2 & 245 & 3 & 3.9  & 0.42  & \ed \\ \hline
			\sdatat & 7.4 & 3383 & 110 & 31.08 & 150.8 &  \jac\\ \hline
		\end{tabular} 
	}  
	\label{tbl:datastat}
\end{table*}

\begin{figure*}[h!t]\vspace{-1.5em}
	\begin{center}
		\subfigure[\small{\vdatao}]{
			\label{subfig:samplemr1}
			\hspace*{-1.0em}\epsfig{figure=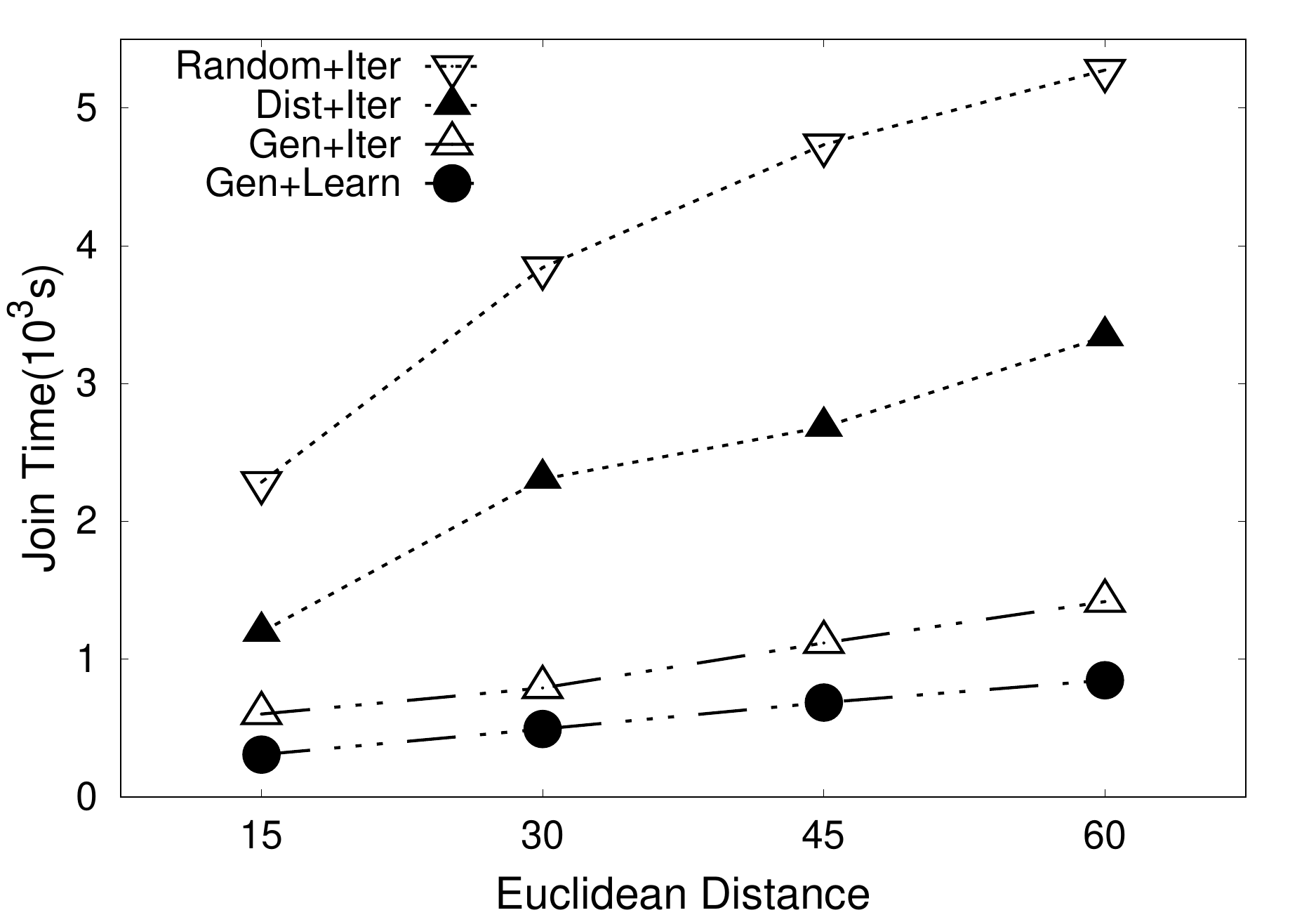,width=0.25\textwidth}}
		\subfigure[\small{\vdatat}]{
			\label{subfig:samplemr2}
			\hspace*{-1.0em}\epsfig{figure=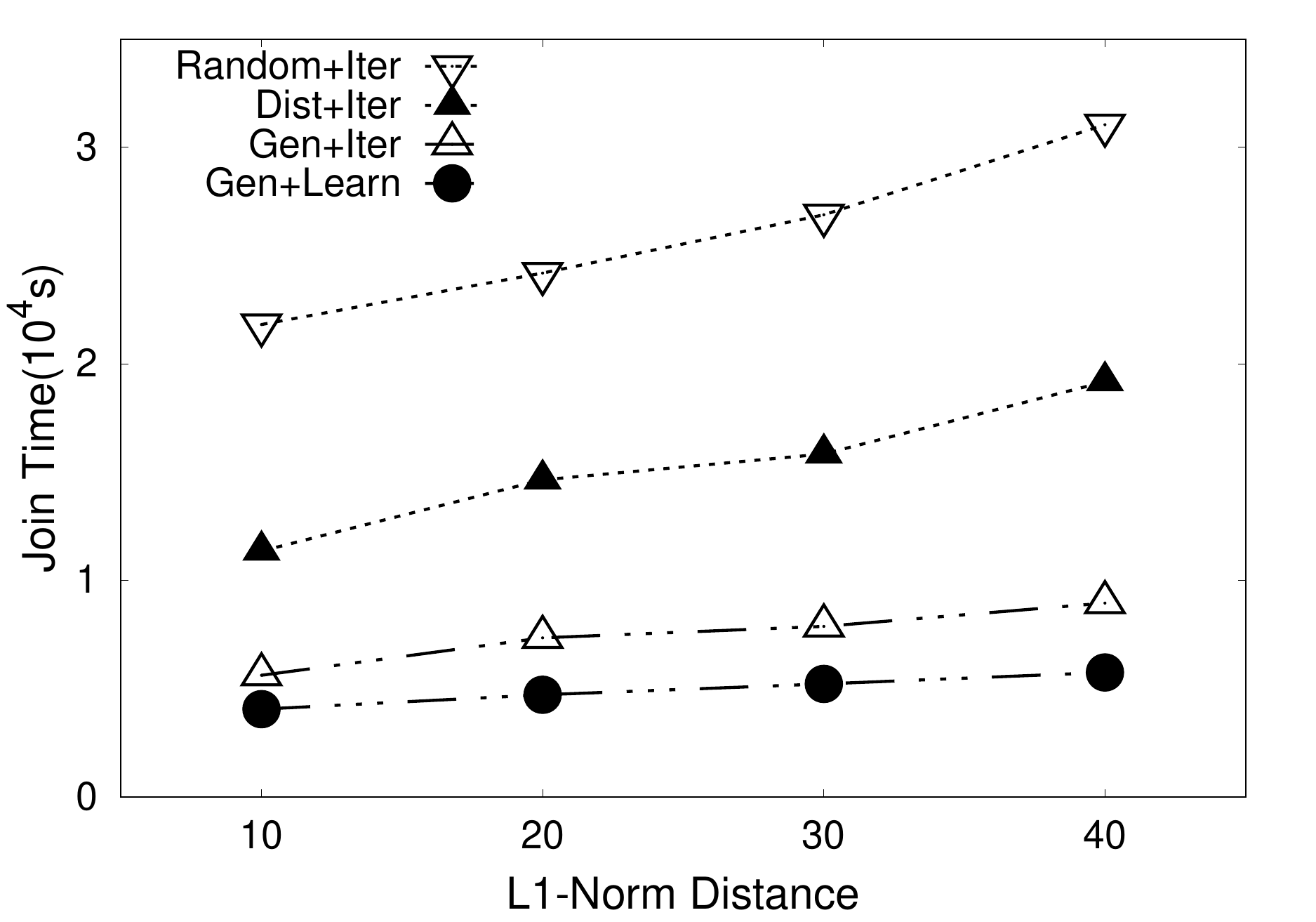,width=0.25\textwidth}}
		\subfigure[\small{\sdatao}]{
			\label{subfig:samplemr3}
			\hspace*{-1.0em}\epsfig{figure=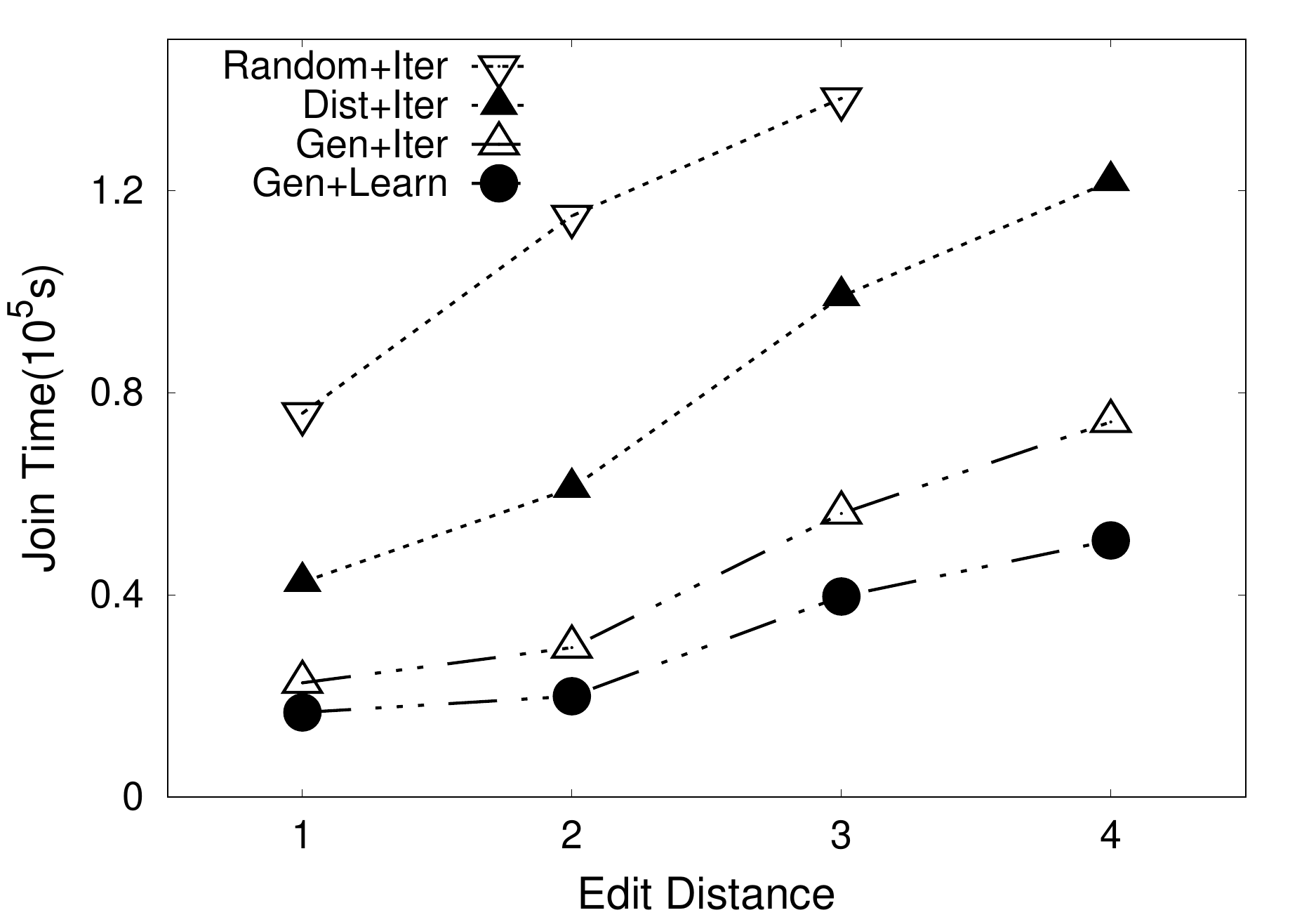,width=0.25\textwidth}}
		\subfigure[\small{\sdatat}]{
			\label{subfig:samplemr4}
			\hspace*{-1.0em}\epsfig{figure=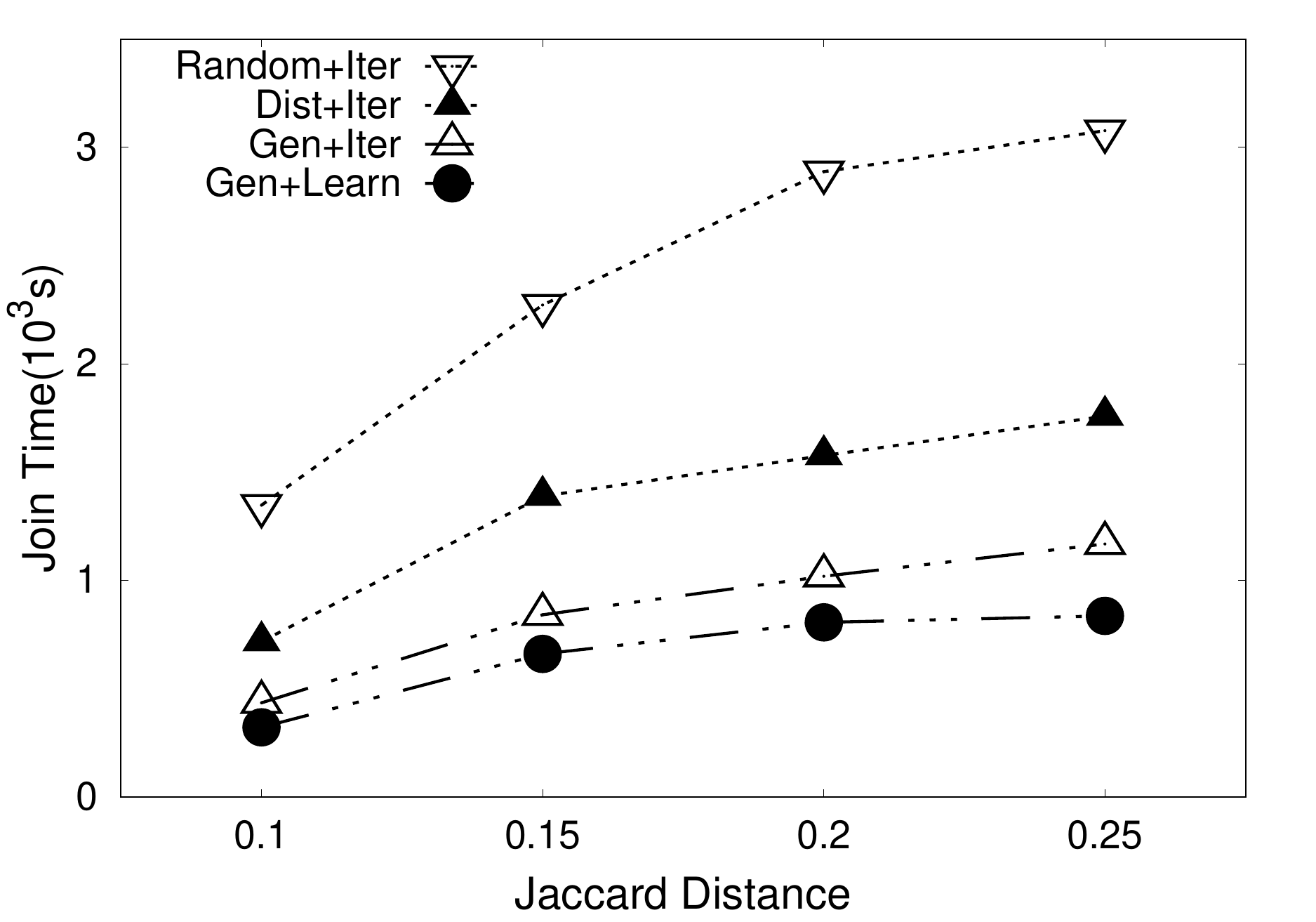,width=0.25\textwidth}}
	\end{center}\vspace{-1em}\vspace{-.5em}
	\caption{Effect of Proposed Techniques}\label{fig:samplemr}\vspace{-.95em}
	\begin{center}
		\subfigure[\small{\vdatao}]{
			\label{subfig:random1}
			\hspace*{-1.0em}\epsfig{figure=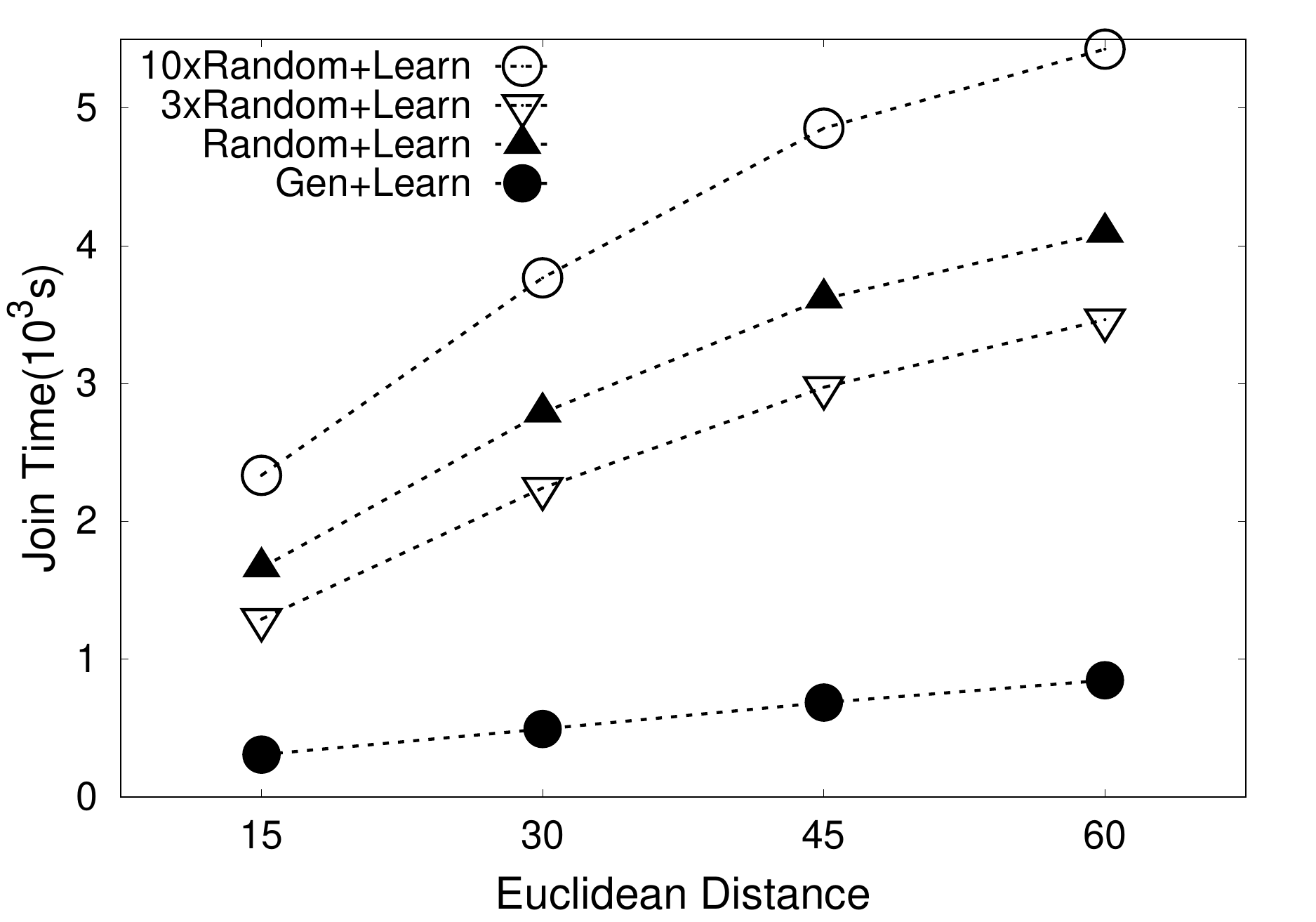,width=0.25\textwidth}}
		\subfigure[\small{\vdatat}]{
			\label{subfig:random2}
			\hspace*{-1.0em}\epsfig{figure=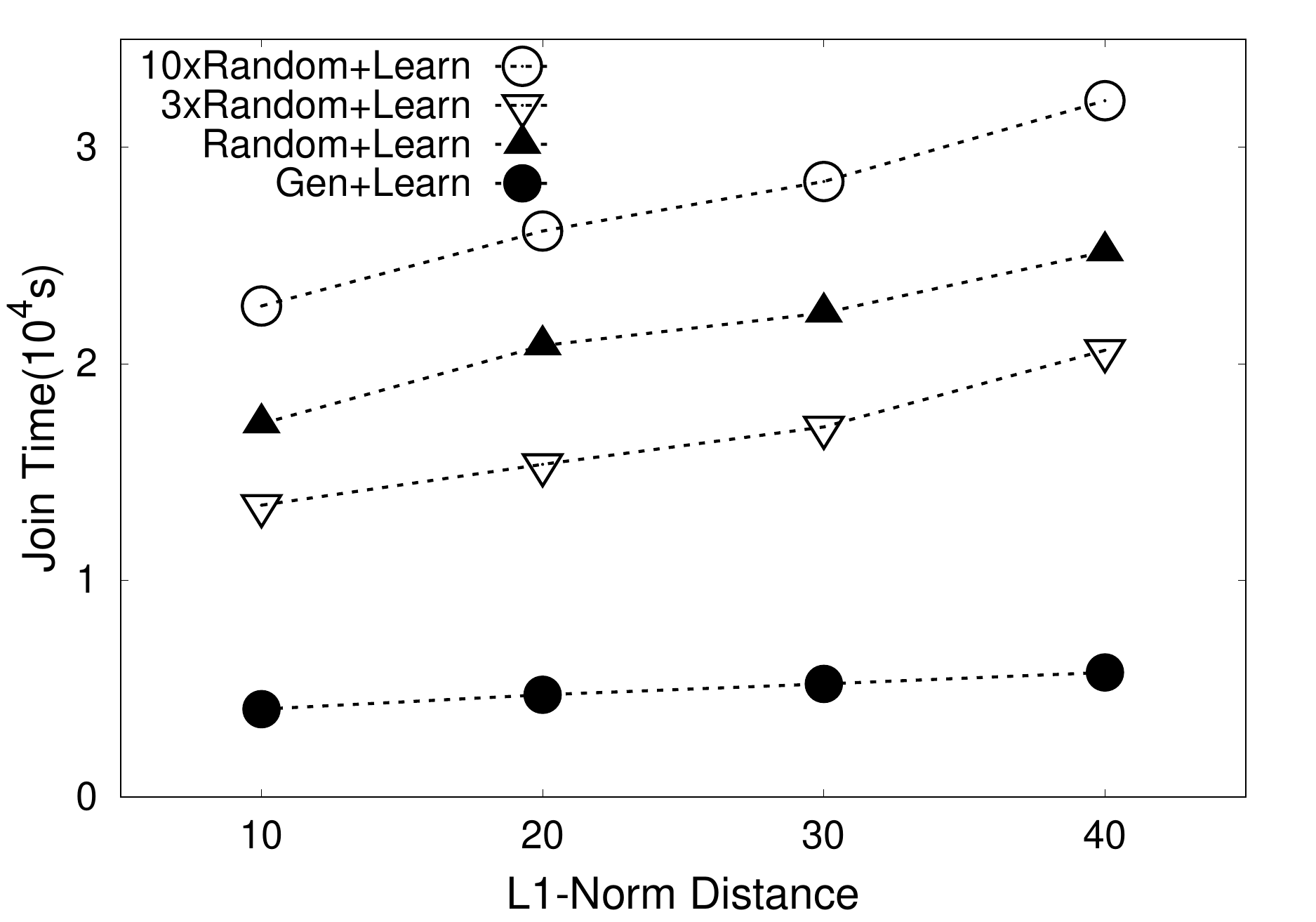,width=0.25\textwidth}}
		\subfigure[\small{\sdatao}]{
			\label{subfig:random3}
			\hspace*{-1.0em}\epsfig{figure=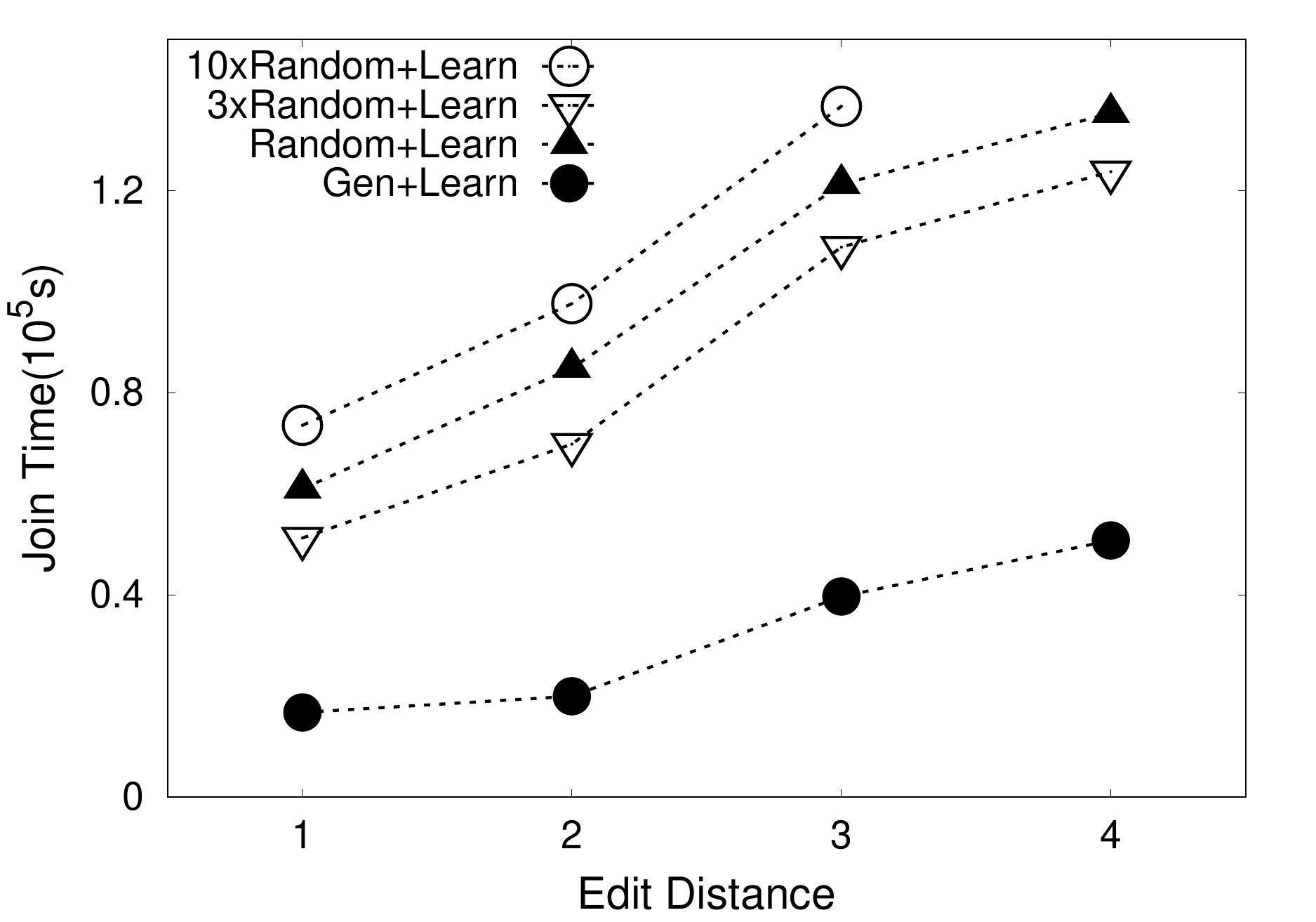,width=0.25\textwidth}}
		\subfigure[\small{\sdatat}]{
			\label{subfig:random4}
			\hspace*{-1.0em}\epsfig{figure=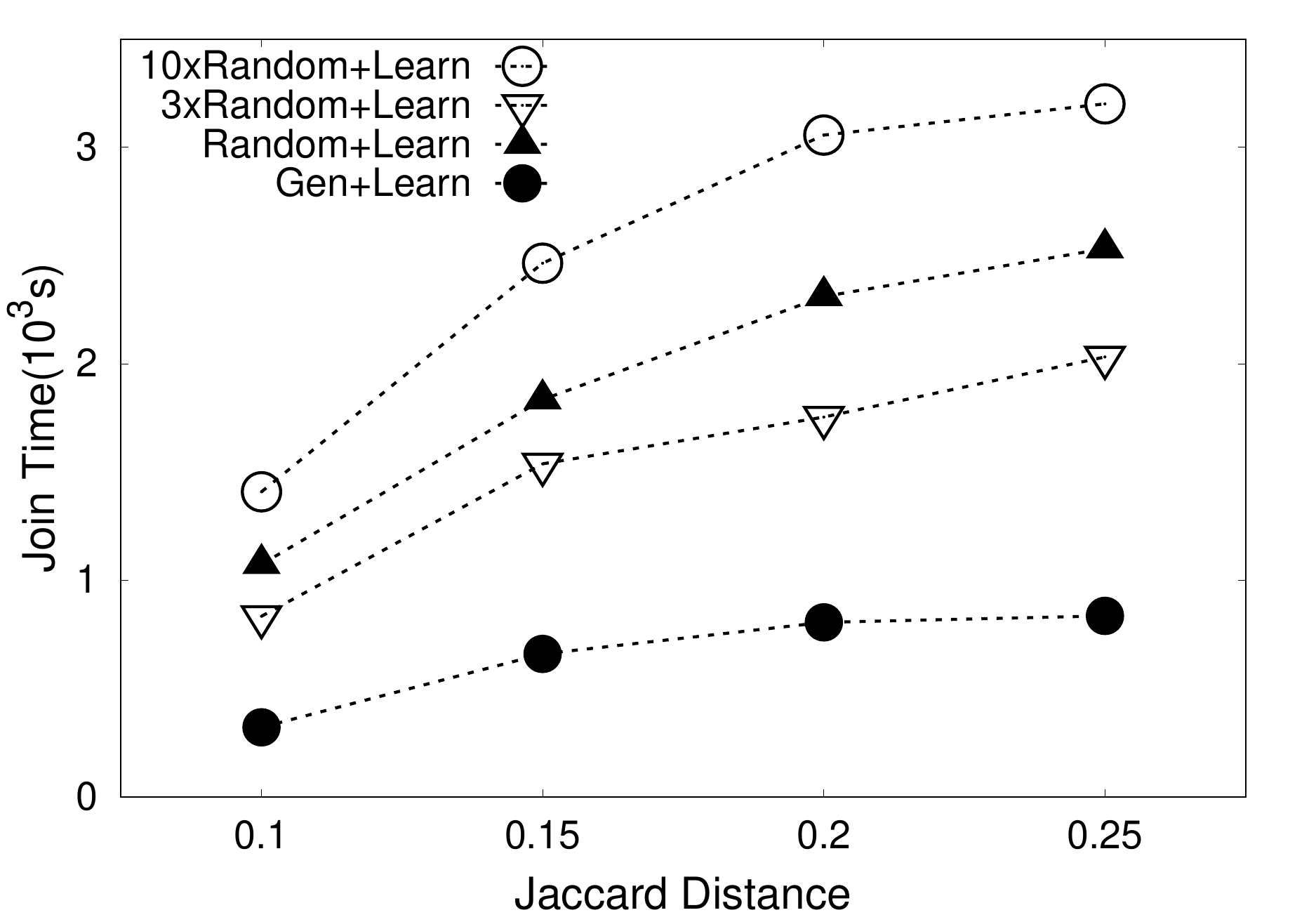,width=0.25\textwidth}}
	\end{center}\vspace{-1em}\vspace{-.5em}
	\caption{Effect of Varying Sample Size}\label{fig:random}\vspace{-1em}
\end{figure*}
\subsection{Experiment Setup}\label{subsec-setup}

The evaluation is conducted based on the methodology of a recent experimental survey~\cite{DBLP:journals/pvldb/FierABLF18}. 
We evaluate our proposed techniques on four real world datasets. 
The first two datasets are used to evaluate distance functions for vector data: \vdatao is a dataset of movie rating~\footnote{www.netflix.com}. 
Each movie is rated with an integer from 1 to 5 by users. 
We select the top 20 movies with the most number of ratings from 421,144 users as previous work did~\cite{DBLP:conf/kdd/WangMP13}. 
\vdatat is a widely-used dataset in the field of image processing~\footnote{http://corpus-texmex.irisa.fr/}. 
The next two datasets are used to evaluate distance functions for string data: \sdatao~\footnote{http://www.gregsadetsky.com/aol-data/} is a collection of query log from a search engine. 
\sdatat is a dataset of abstracts of biomedical literature from PubMed~\footnote{https://www.ncbi.nlm.nih.gov/pubmed/}. 
The details of datasets are shown in Table~\ref{tbl:datastat}. 

We compare \name with state-of-the-art methods: 
For vector data, we compare our work with four methods: \mrsim~\cite{DBLP:conf/sigmod/SilvaR12}, \mpass~\cite{DBLP:conf/kdd/WangMP13}, \clusterjoin~\cite{DBLP:journals/pvldb/SarmaHC14} and \kdtree~\cite{DBLP:journals/tkde/ChenYCGZC17} on \lonenormdist and \eudist, respectively. 
For string data, we compare our work with \mrsim~\cite{DBLP:conf/sigmod/SilvaR12}, \massjoin~\cite{DBLP:conf/icde/DengLHWF14} and \kdtree~\cite{DBLP:journals/tkde/ChenYCGZC17} for \jacdist and \eddist. 
For \jacdist, we also compare with \fsjoin~\cite{DBLP:conf/icde/RongLSWLD17}. 
As there are no public available implementations of above methods on Spark, we implement all of them by ourselves. Although there are also some other related work~\cite{DBLP:conf/sigmod/OkcanR11,DBLP:conf/icde/FriesBSS14,DBLP:conf/sigmod/VernicaCL10,DBLP:conf/icde/AfratiSMPU12,DBLP:journals/pvldb/MetwallyF12}, previous studies have showed that they cannot outperform above selected algorithms, so we do not compare with them due to space limitation.

All the algorithms, including both our proposed methods and state-of-the-art ones, are implemented with Scala on the platform of Apache Spark 2.1.0. 
We use the default settings of HDFS. 
The experiments are run on a 16-node cluster (one serves as a master node and others serve as slave nodes). 
Each node has four 2.40GHz Intel Xeon E312xx CPU, 16GB main memory and 1TB disk running 64-bit Ubuntu Server 16.04. 
We run all the algorithms 10 times and reported the average results.
Some methods cannot finish within $150,000$ seconds.
In this case, we regard it as a timeout and do not report their results in the figure. 

Among all experiments, we use the default parameters as: the sample size $k = 3200$, the number of partitions $p = 60$, the number of dimensions of mapping space $n = 10$ and the number of computing nodes (slave nodes) as $15$.
The choice of hyper-parameters will be discussed later in Section~\ref{subsec-evals}.

\subsection{Effect of Proposed Techniques}\label{subsec-evals}
\begin{figure}[h]
	\begin{center}
		\subfigure[\small{\vdatat}]{
			\label{subfig:samplesize-data3}
			\hspace*{-1.0em}\epsfig{figure=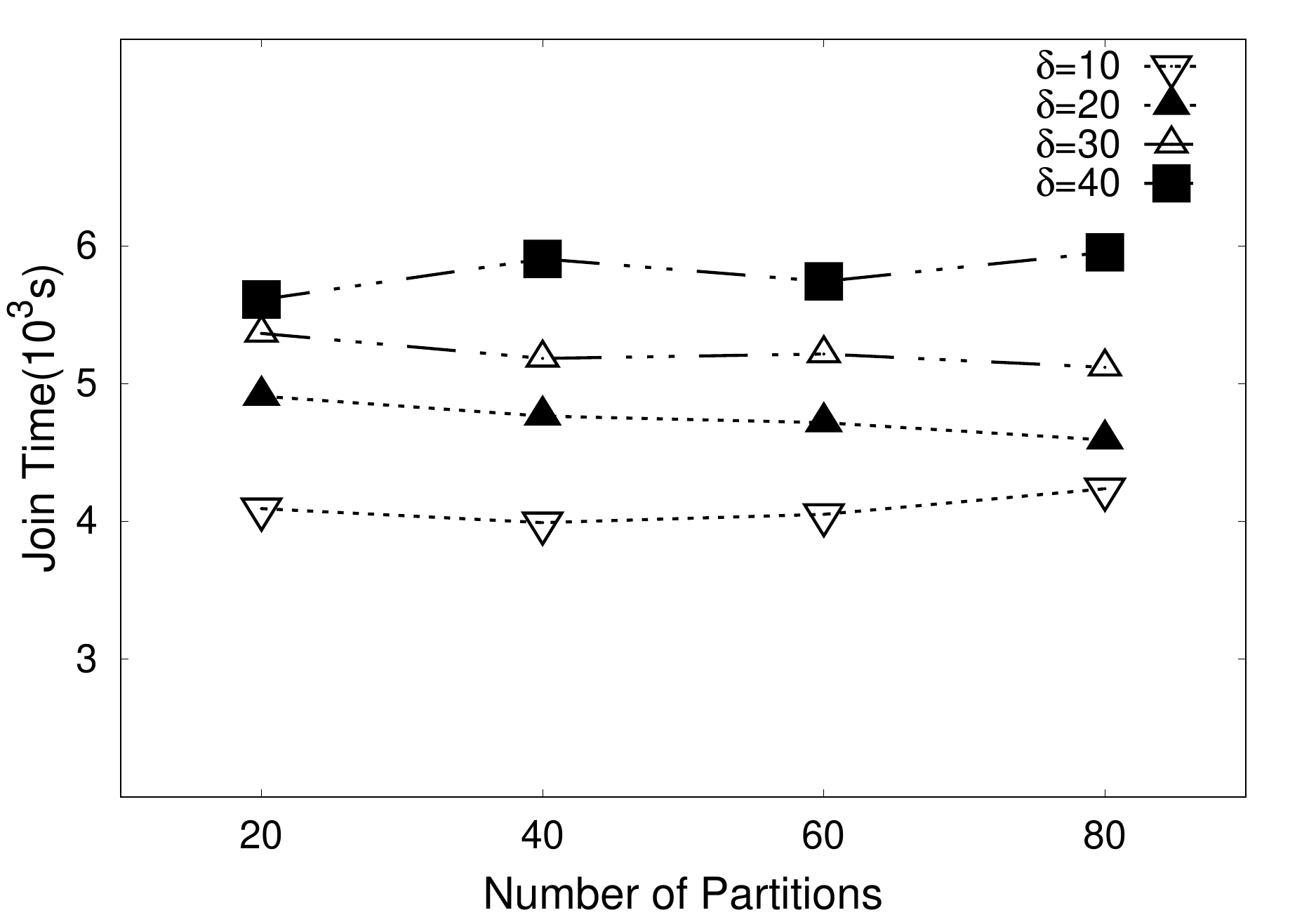,width=0.5\linewidth}}
		\subfigure[\small{\sdatao}]{
			\label{subfig:samplesize-data4}
			\hspace*{-1.0em}\epsfig{figure=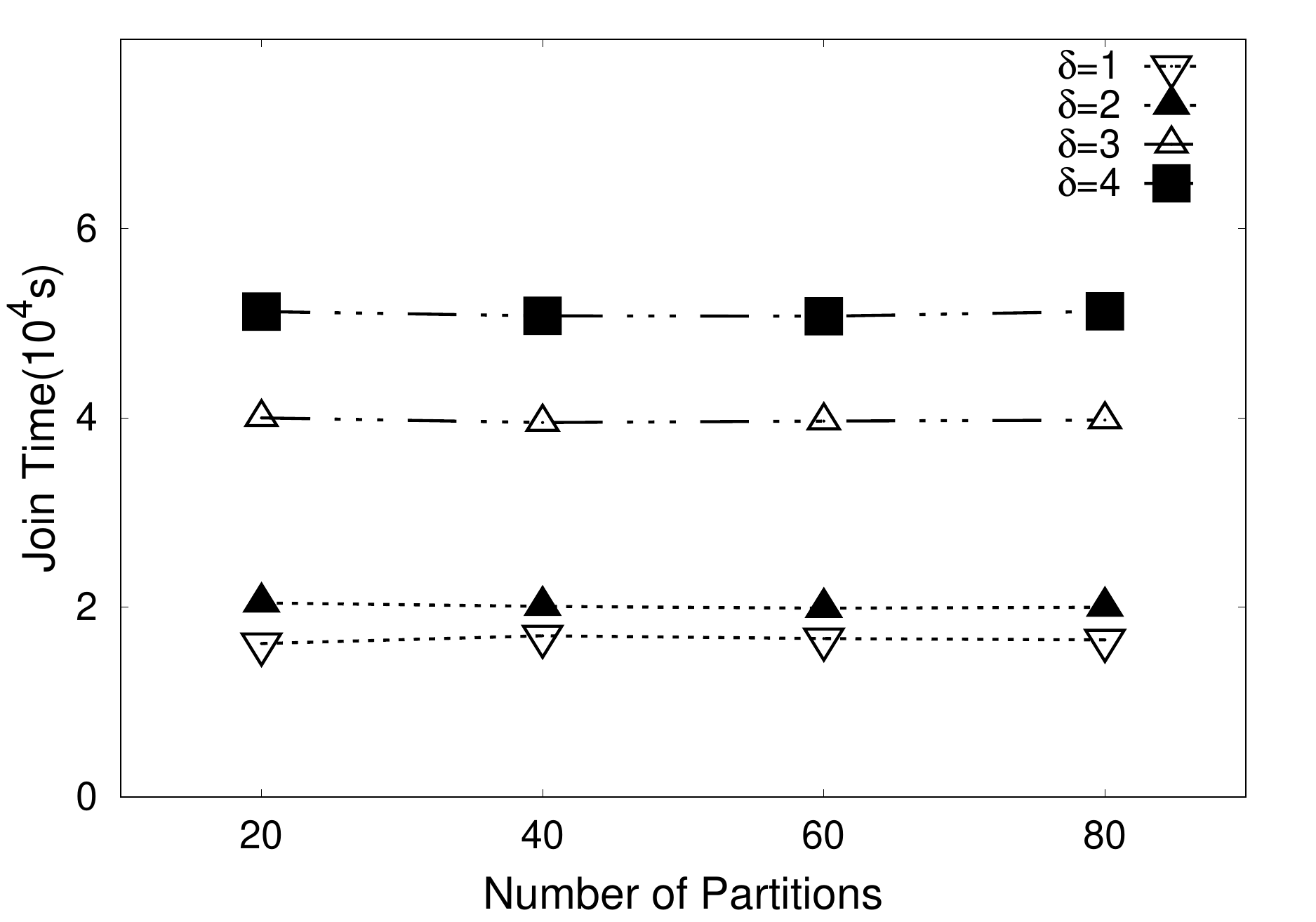,width=0.5\linewidth}}
	\end{center}\vspace{-1em}\vspace{-.5em}
	\caption{Effect of Partition Number}\label{fig:para:partition_samplesize}\vspace{-1em}
\end{figure}
\begin{figure*}[t]
	\begin{center}
		\subfigure[\small{\vdatao}]{
			\label{subfig:baseline1}
			\hspace*{-1.0em}\epsfig{figure=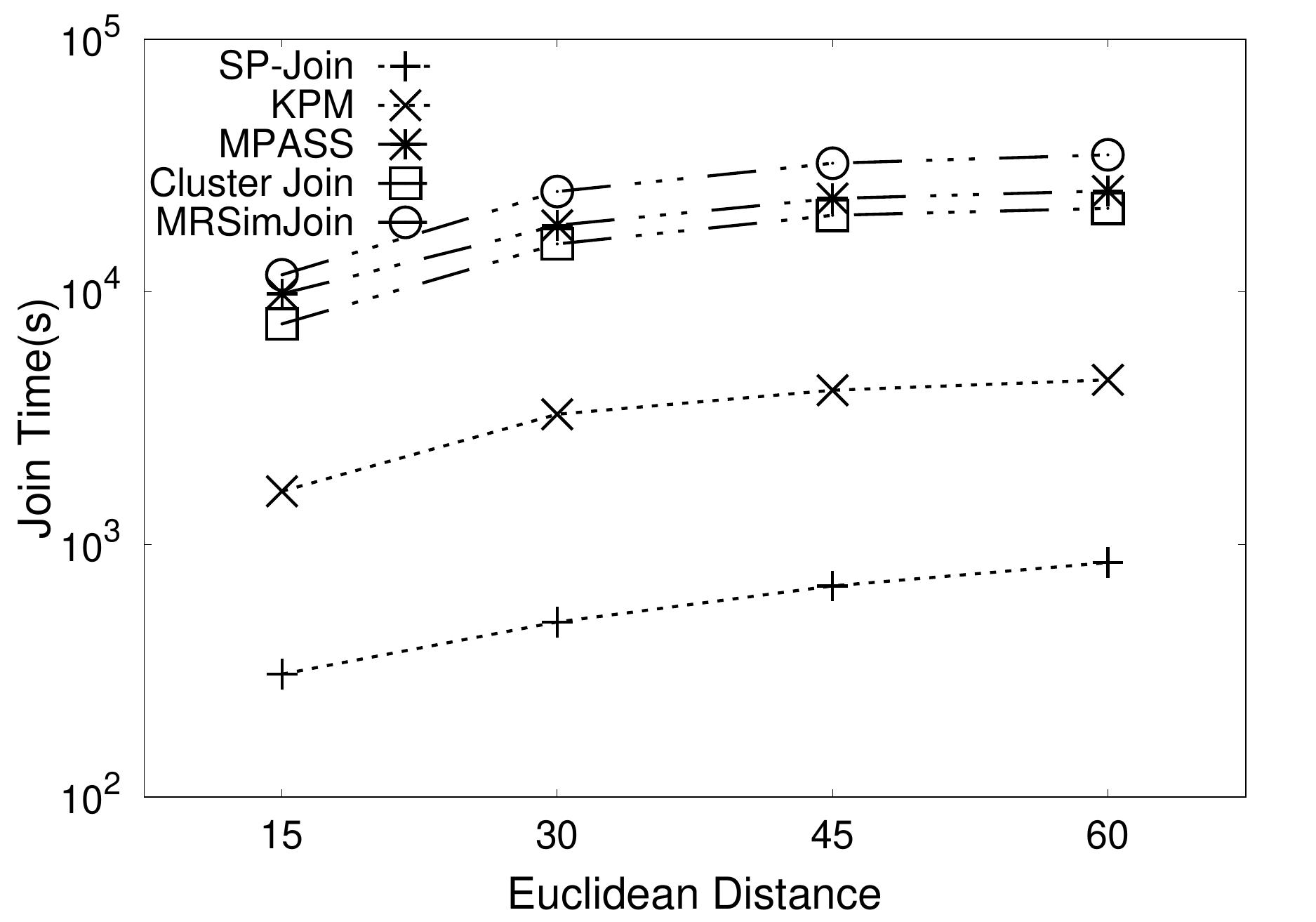,width=0.25\textwidth}}
		\subfigure[\small{\vdatat}]{
			\label{subfig:baseline2}
			\hspace*{-1.0em}\epsfig{figure=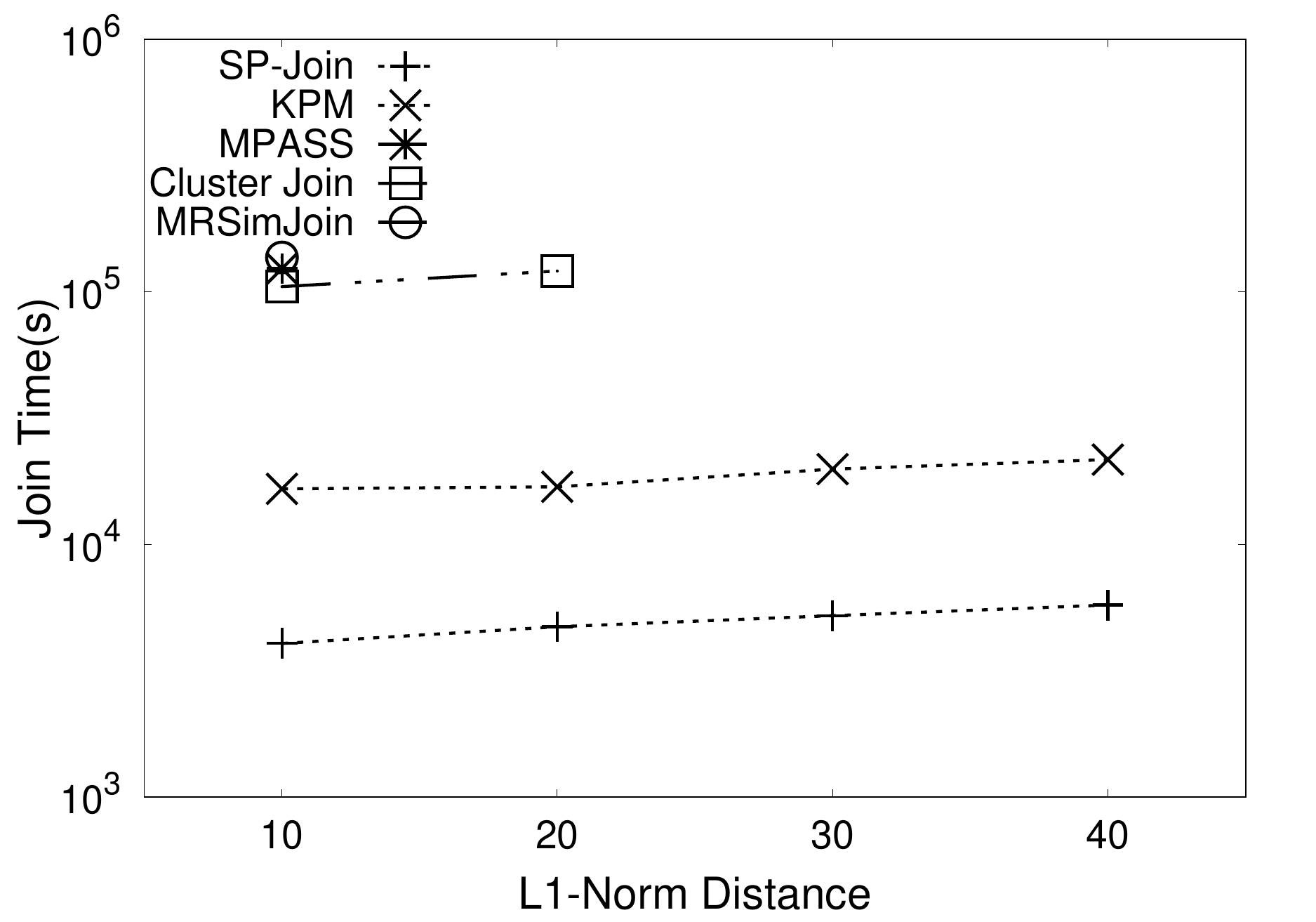,width=0.25\textwidth}}
		\subfigure[\small{\sdatao}]{
			\label{subfig:baseline3}
			\hspace*{-1.0em}\epsfig{figure=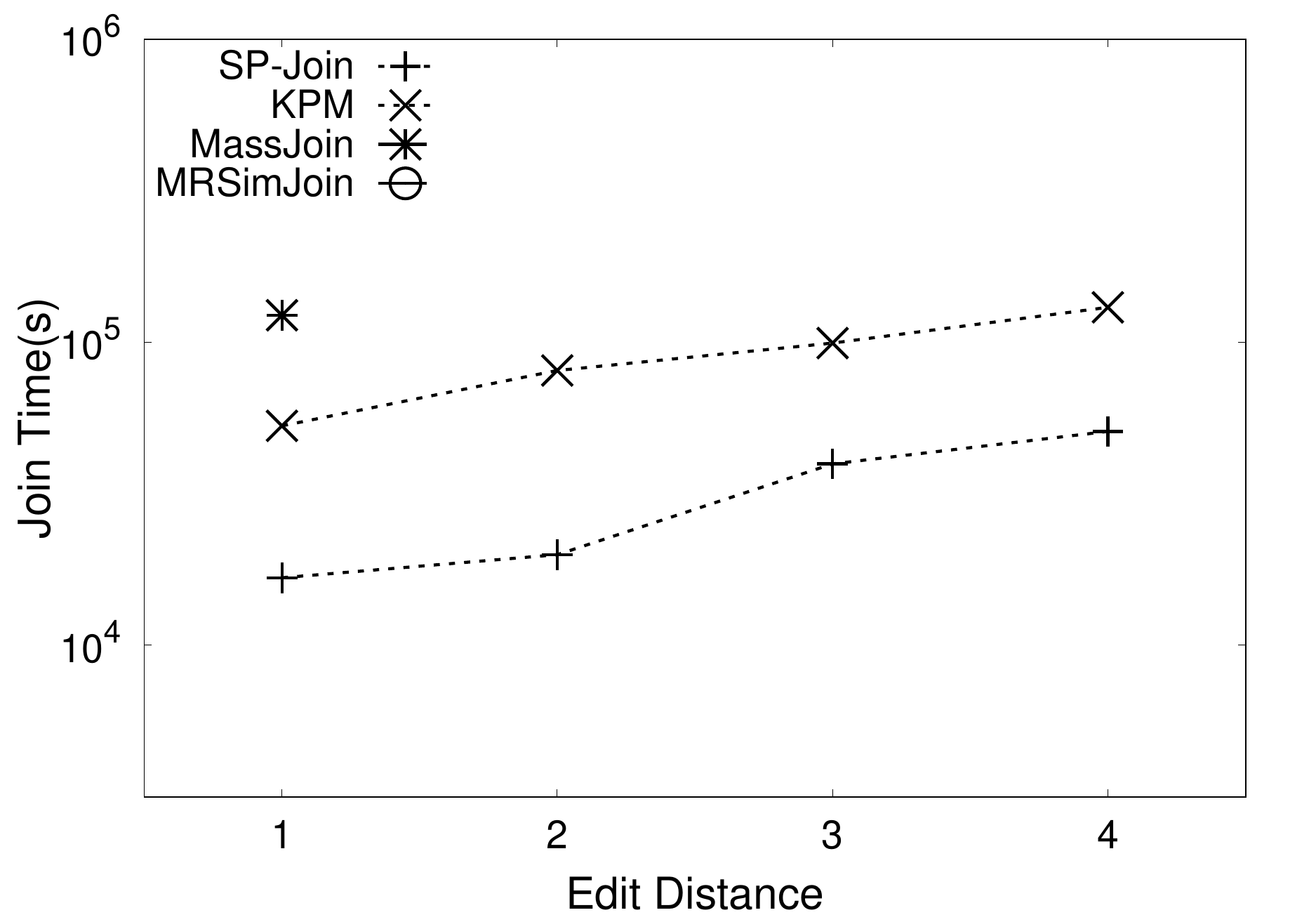,width=0.25\textwidth}}
		\subfigure[\small{\sdatat}]{
			\label{subfig:baseline4}
			\hspace*{-1.0em}\epsfig{figure=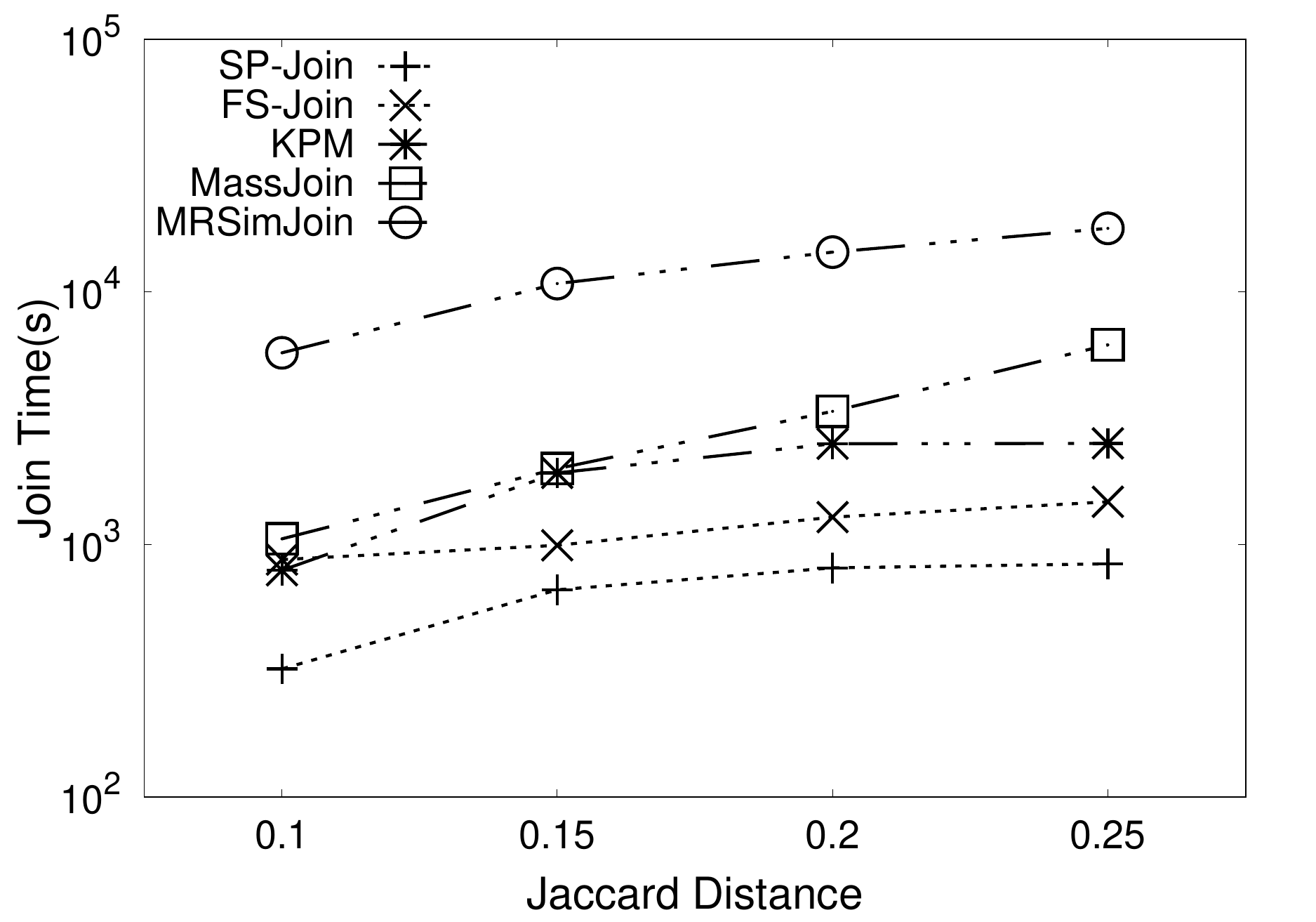,width=0.25\textwidth}}
	\end{center}\vspace{-1em}\vspace{-.5em}
	\caption{Compare with State-of-the-art Methods}\label{fig:baseline}\vspace{-1em}
\end{figure*}

First, we evaluate the effect of proposed techniques. 
We implement four methods combining different techniques proposed in sampling and map phases.
For the sampling phase, we implement 3 methods: \randm is the simple random sampling approach.
\dstaw is the distribution-aware sampling approach proposed in Section~\ref{subsec-distsample} and \gener is the generative sampling method proposed in Section~\ref{subsec-gensample}. 
For map phase, \basictr is the method that splits the space iteratively by randomly selecting a dimension (Section~\ref{subsec-basicmap}).
\entropy is the learning-based method using cost variation to select the dimension in each iteration (Section~\ref{subsec-advmap}). 
The results are shown in Figure~\ref{fig:samplemr}. 
Regarding the join time breakdown, the random sampling consists of 2\% of overall runtime on average while the generative sampling costs about 6\%.

We can see that the \gener and \dstaw outperform \randm under all the settings. 
The reason is that we utilize distribution information from the dataset to obtain samples instead of just performing random sampling. 
For example, when the threshold is 30 on \vdatat, the overall join times of \gener\unskip+\basictr and \dstaw\unskip+\basictr are 7872 and 15835 seconds respectively, while that of \randm\unskip+\basictr is 26885 seconds.  
Moreover, we observe that \gener outperforms \dstaw in most cases. 
The saving of cost mainly comes from the communication cost: \gener only needs to transmit the parameters of distribution and is independent of the sample size. 
While the sampling quality is comparable, \gener can obviously reduce the overhead of sampling. 

To further demonstrate the advantage of our proposed sampling technique w.r.t random sampling, we conduct another experiment by significantly increasing the sample size of the random method to 3 times and 10 times of that of generative method.
The results are shown in Figure~\ref{fig:random}.
We can see that when the sample size increases slightly (3X), there is a marginal improvement on the performance of random sampling based framework.
However, when the sample size increases significantly (10X), the performance degrades seriously and becomes even worse than the original ones. 
The reason is that the overhead of partition strategies in the map phase increases along with larger sampling size.
Thus the benefits of larger sample size could be counteracted by such overhead.
This further shows that the random sampling cannot make good use of larger sample size and it is necessary to devise more effective sampling techniques.

Finally, we discuss the effect of hyper-parameter.
For sample size $k$, we decide it according to the theoretical analysis.
With the help of learning-based partition strategy, our approach is not sensitive to the dimension of target space $n$ in the range of [5,30] and we set it empirically.
We then show the empirical results of varying the number of partitions $p$. 
Due to the space limitation, here we only show the results on two datasets: \vdatat and \sdatao in Figure~\ref{fig:para:partition_samplesize}. 
On the other two datasets, it also shows a similar trend. 
We can see that our method is not sensitive to $p$: the difference between the best and worst performance by varying $p$ is just about 10\%. 
This demonstrates the robustness of our methods.

\subsection{Compare with State-of-the-art Methods}\label{subsec-baseline}

Next we compare \name with state-of-the-art methods. 
For all baseline methods, we try our best to tune their parameters according to the descriptions in previous studies. 
The results are shown in Figure~\ref{fig:baseline}.
We can see that \name achieves the best result on all datasets. 
For example, on \vdatao($\delta$ is 45), the join time of \name is only 684 seconds, while \mrsim, \mpass, \clusterjoin and \kdtree use 32266, 23382, 20064 and 4074 seconds respectively. 
On \sdatat($\delta$ is 0.1), the join time of \name is 450 seconds, while \mrsim, \massjoin, \fsjoin and \kdtree use 5726, 1470, 1218 and 1104 seconds respectively.
\begin{figure*}[!t]
	\begin{center}
		\subfigure[\small{\vdatao}]{
			\label{subfig:scal-node1}
			\hspace*{-1.0em}\epsfig{figure=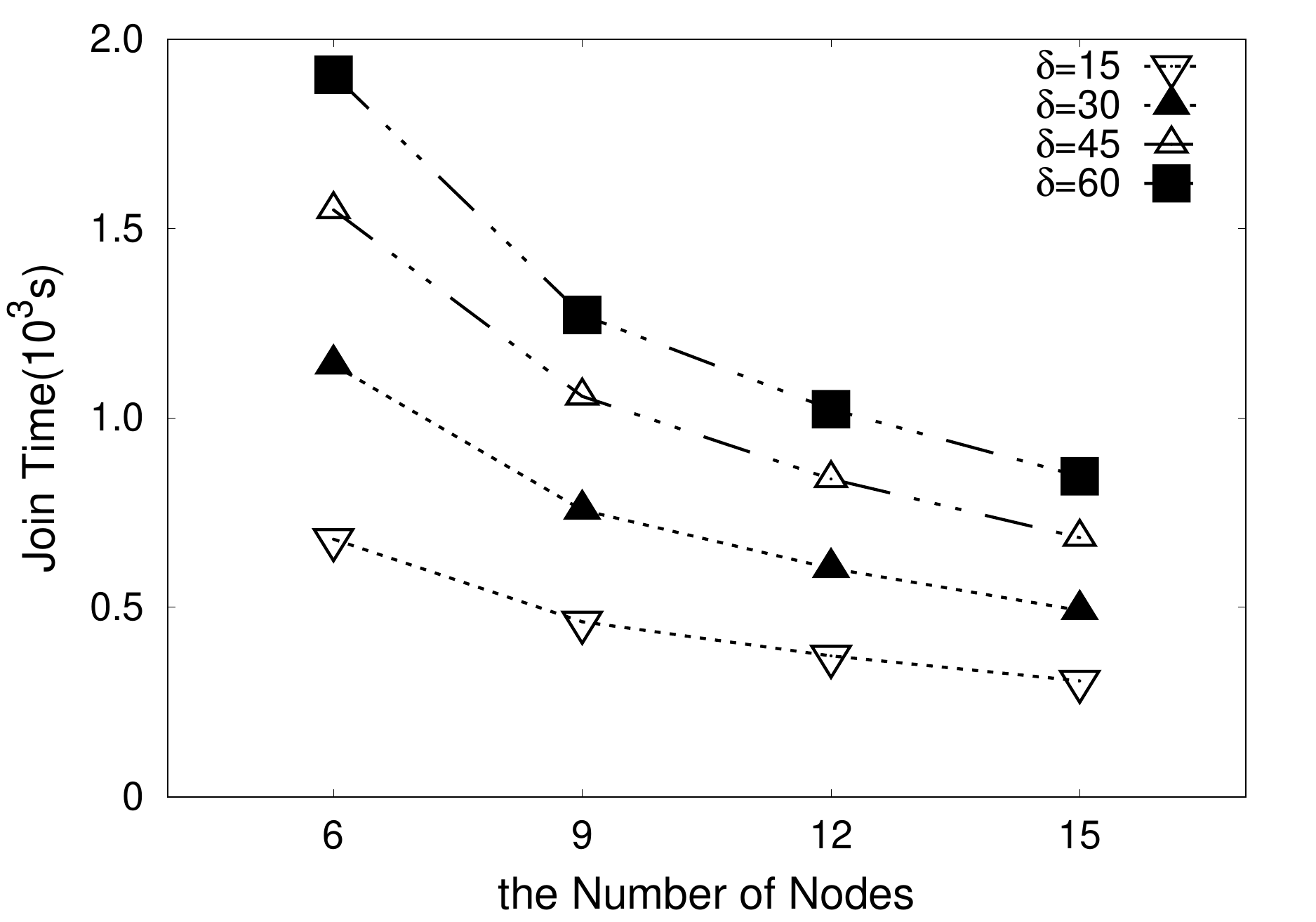,width=0.25\textwidth}}
		\subfigure[\small{\vdatat}]{
			\label{subfig:scal-node2}
			\hspace*{-1.0em}\epsfig{figure=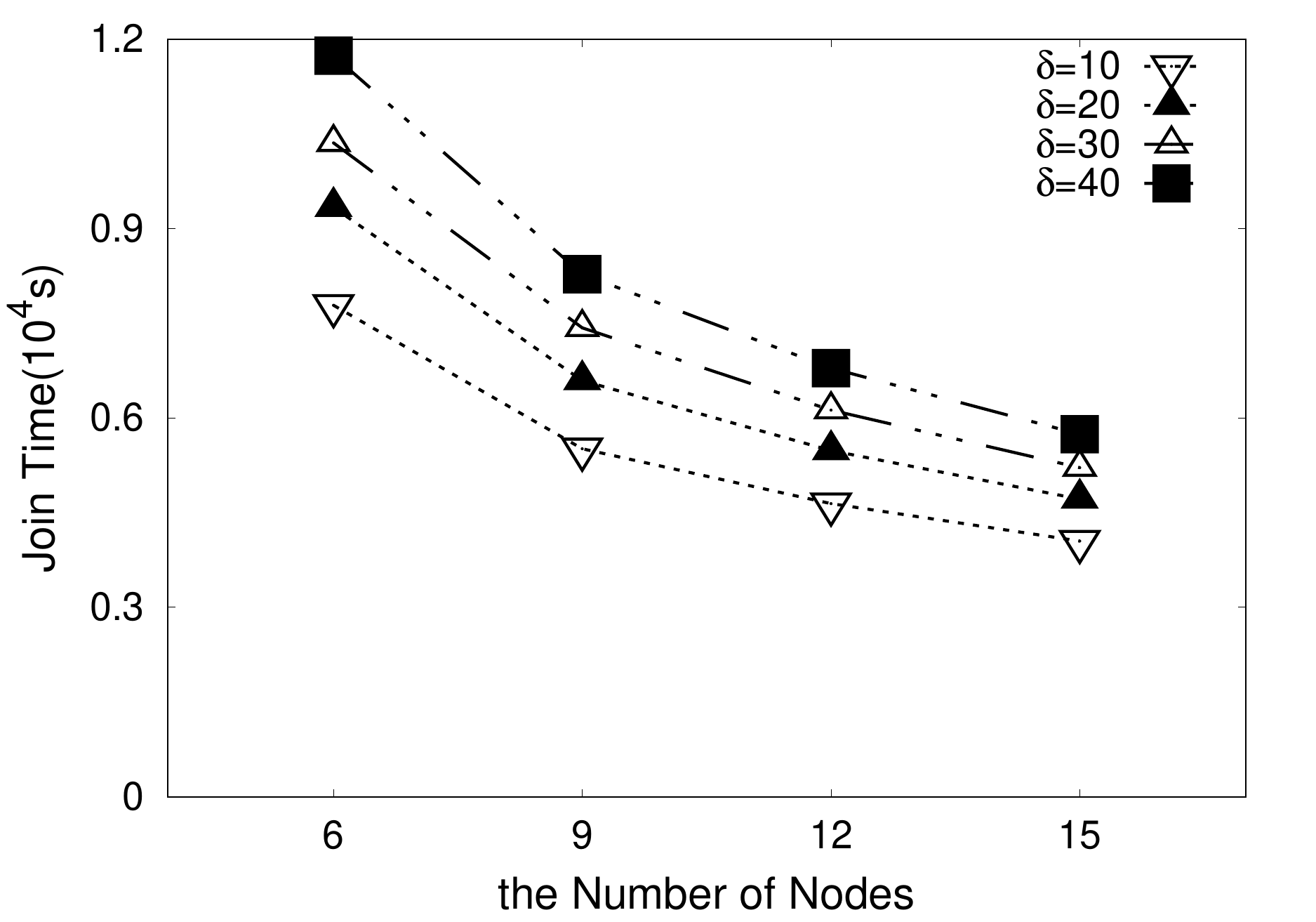,width=0.25\textwidth}}
		\subfigure[\small{\sdatao}]{
			\label{subfig:scal-node3}
			\hspace*{-1.0em}\epsfig{figure=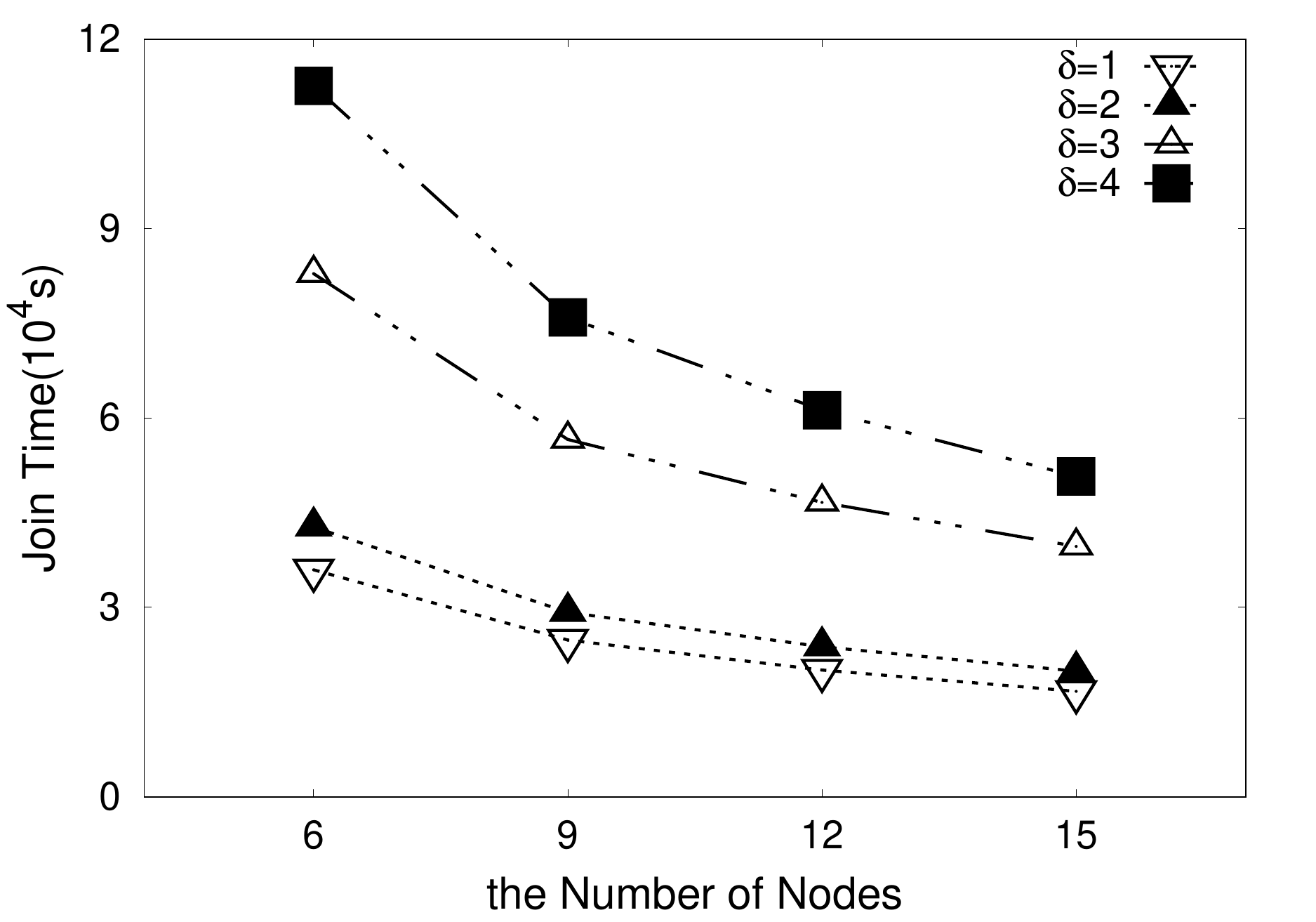,width=0.25\textwidth}}
		\subfigure[\small{\sdatat}]{
			\label{subfig:scal-node4}
			\hspace*{-1.0em}\epsfig{figure=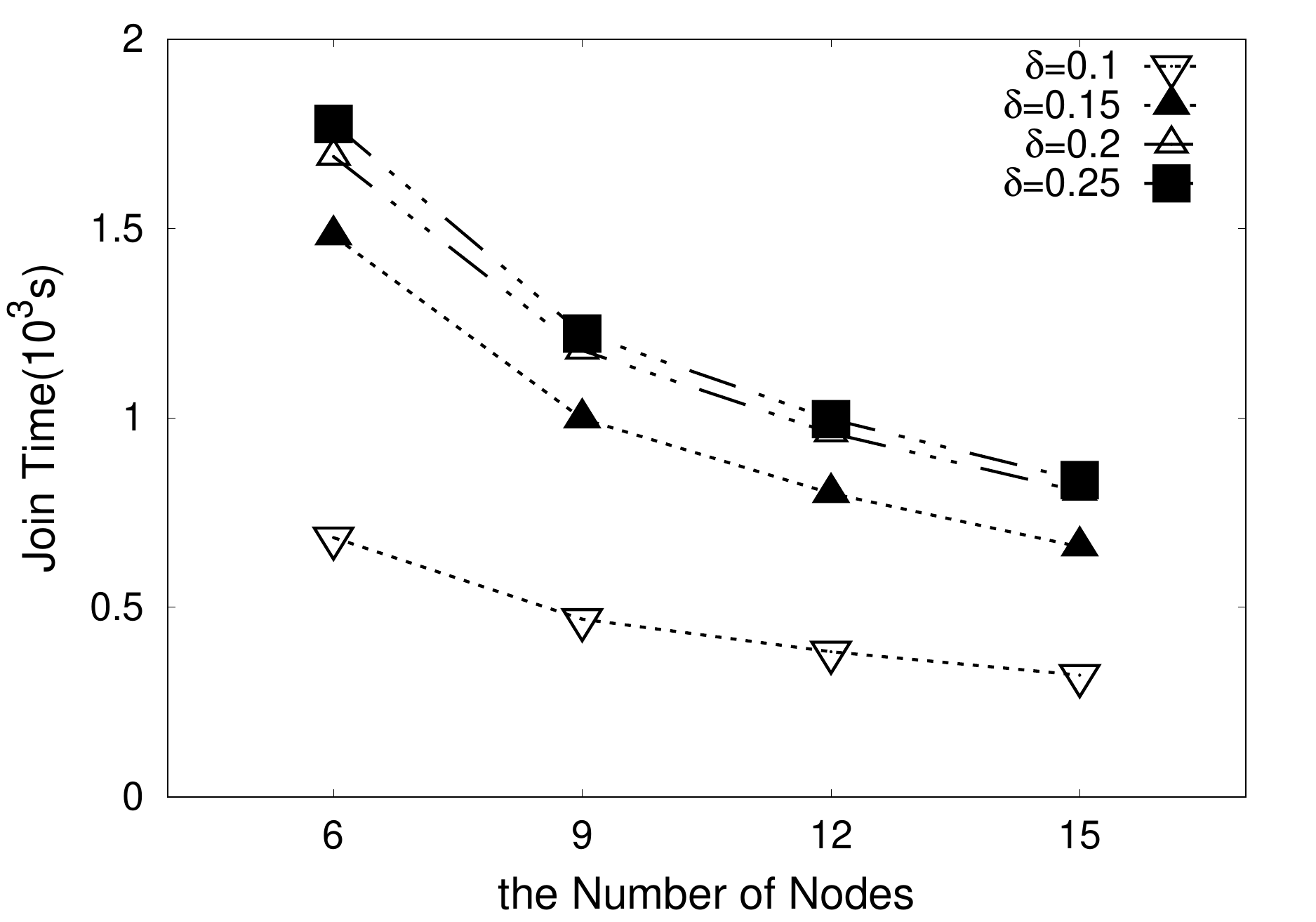,width=0.25\textwidth}}
	\end{center}\vspace{-1em}\vspace{-.5em}
	\caption{Scalability: Effect of Number of Nodes}\label{fig:scal:node}\vspace{-1em}
	\begin{center}
		\subfigure[\small{\vdatao}]{
			\label{subfig:scal-data1}
			\hspace*{-1.0em}\epsfig{figure=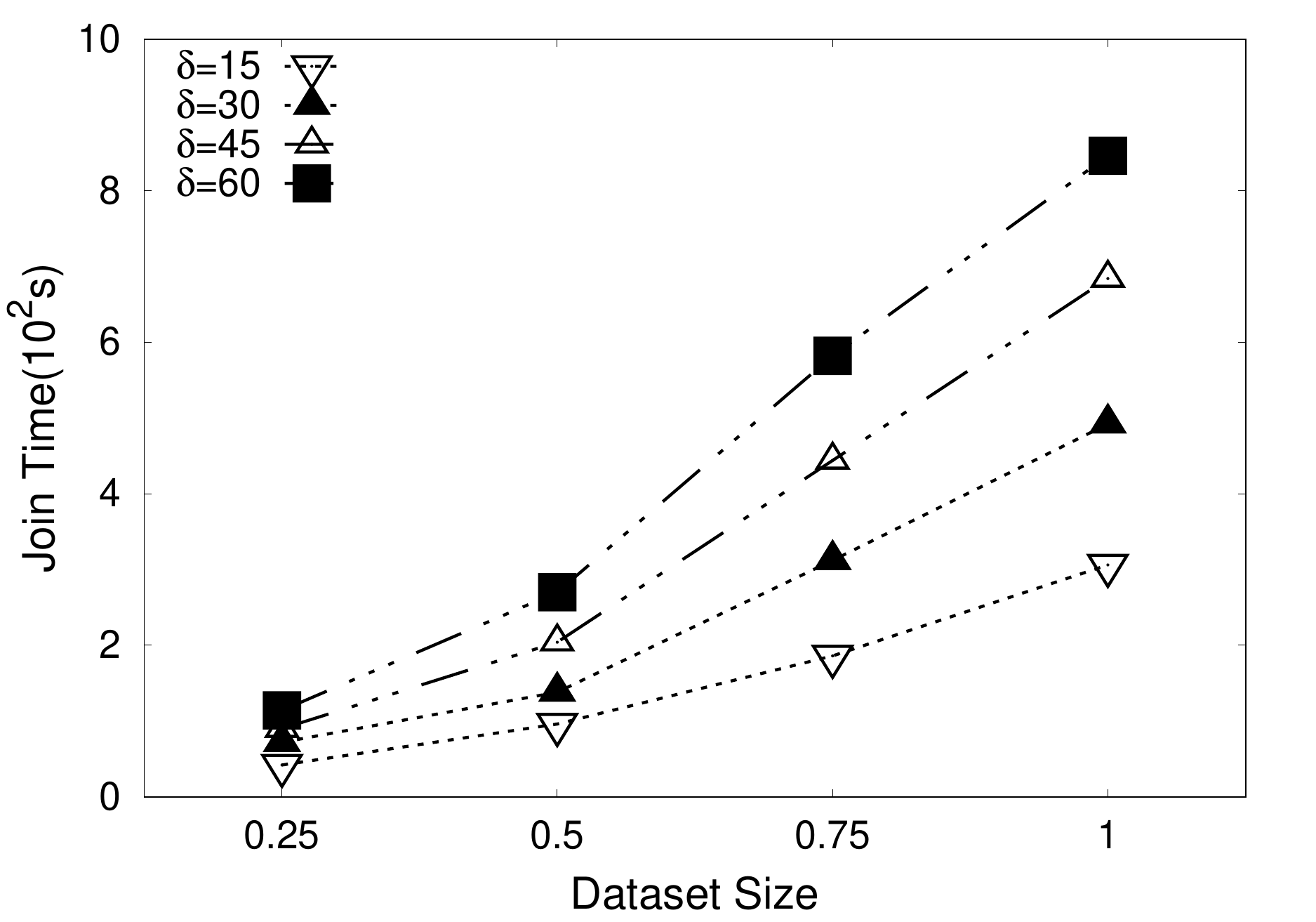,width=0.25\textwidth}}
		\subfigure[\small{\vdatat}]{
			\label{subfig:scal-data2}
			\hspace*{-1.0em}\epsfig{figure=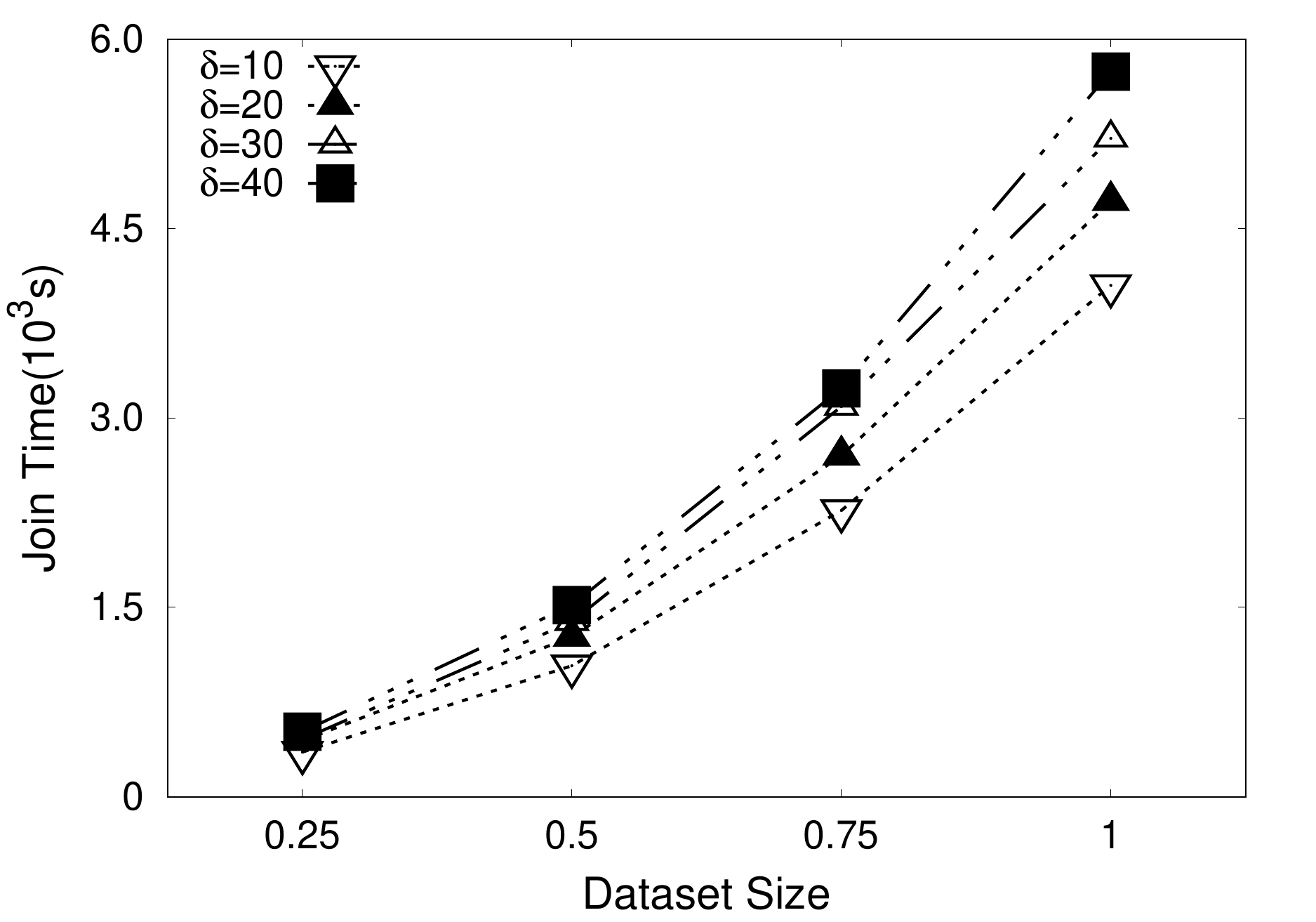,width=0.25\textwidth}}
		\subfigure[\small{\sdatao}]{
			\label{subfig:scal-data3}
			\hspace*{-1.0em}\epsfig{figure=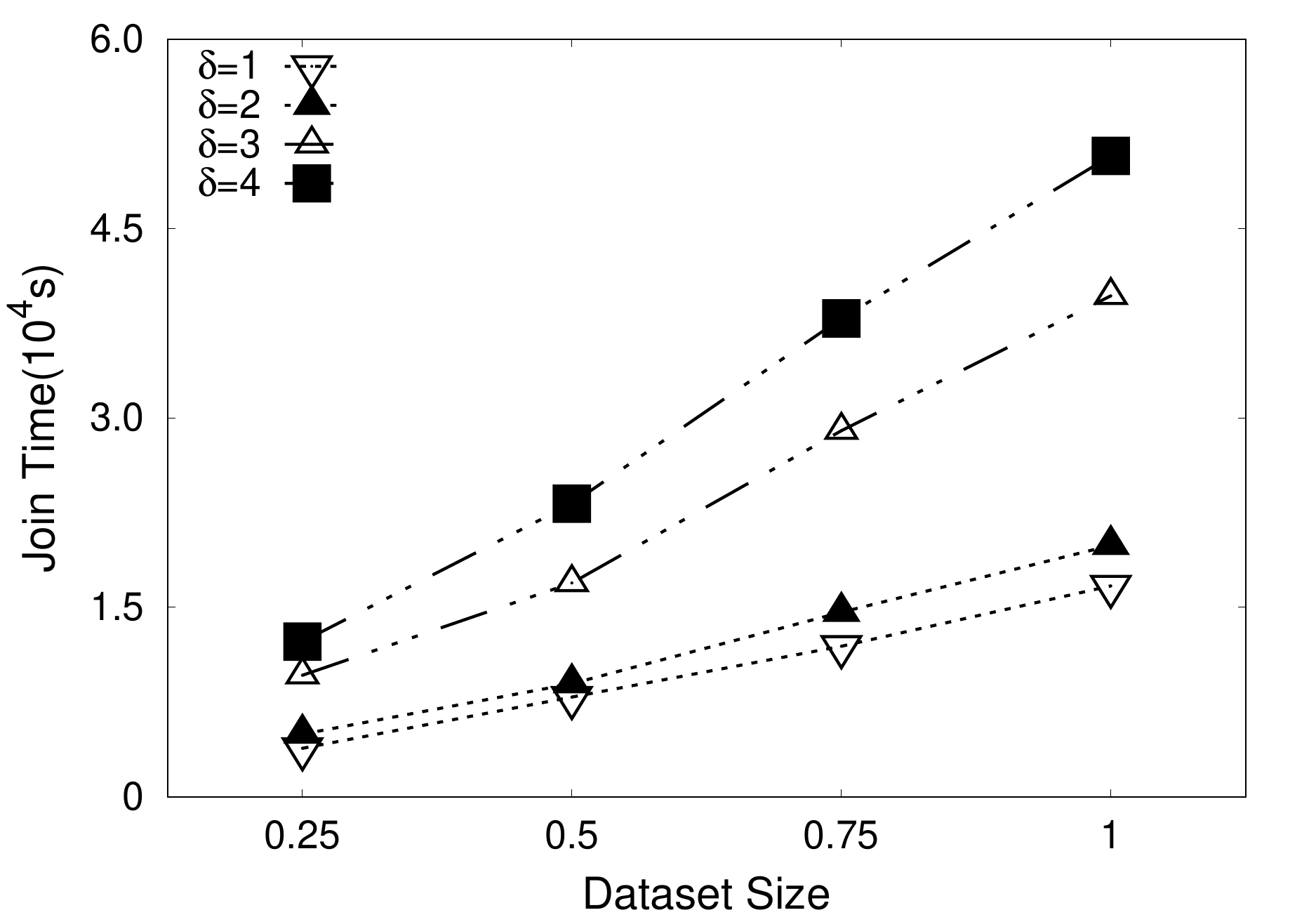,width=0.25\textwidth}}
		\subfigure[\small{\sdatat}]{
			\label{subfig:scal-data4}
			\hspace*{-1.0em}\epsfig{figure=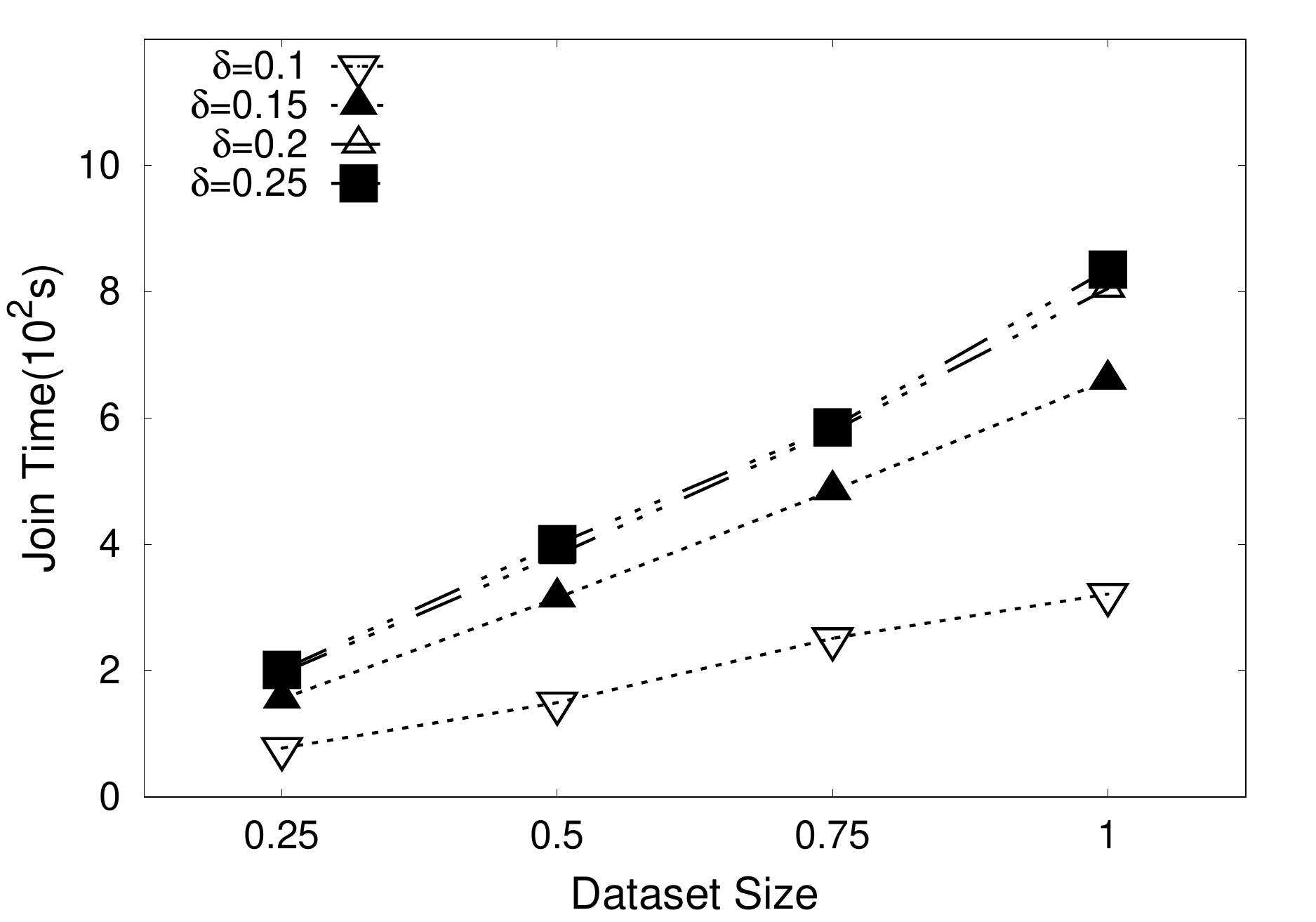,width=0.25\textwidth}}
	\end{center}\vspace{-1em}\vspace{-.5em}
	\caption{Scalability: Effect of Data Size}\label{fig:scal:data}\vspace{-1em}
\end{figure*}

For vector data, the performance of \kdtree ranks second because it can avoid repartition by trying to divide the space evenly, while repartition is required for \mpass and \clusterjoin. 
\name is better than \kdtree mainly because it employs effective sampling techniques to obtain the pivots of partition, while \kdtree just uses random sampling.

For string data, \name performs best because \fsjoin and \mpass need to generate a large amount of signatures for filtering, which leads to heavy overhead. 
Similar phenomenon has also been found in the experiments of~\cite{DBLP:conf/icde/RongLSWLD17}. 
Note that \mrsim run out of time even when $\delta=1$ on \sdatao, therefore no results are shown in Figure~\ref{subfig:baseline3}.
As their methods need multiple stages of MapReduce jobs, there would be a very heavy overhead for filtering and preprocessing, such as the expensive signature generation process. 
However, compared with them our sampling and partition techniques are rather light-weighted. 
Therefore, we also achieve good performance for string data. 

Besides, we also adopt the notion of gap factor employed in~\cite{DBLP:journals/pvldb/MannAB16} to show the robustness of \name. 
More specifically, we measure the average, median, and maximum deviation of the runtime of \name from the best reported one, which are 1.23, 1.04 and 3.08 over all datasets and thresholds. 

\subsection{Scalability}\label{subsec-scale}

Finally, we evaluate the scalability of our method. 
We conduct experiments to test effects of both scaling out (the number of nodes) and scaling up (the data size). 
The results of varying the number of computing nodes from 6 to 15 are shown in Figure~\ref{fig:scal:node}. 
We can see that the performance of our algorithm improves obviously with the increasing number of nodes in a cluster. 
For example, on the \vdatat dataset for \lonenormdist($\delta$ is 10), the join time with $6,9,12,15$ nodes are 7785, 5513, 4644 and 4050 seconds respectively. 
It demonstrates that our method is able to make good use of computing resources.
The results of varying data size are shown in Figure~\ref{fig:scal:data}. 
We can see that as the size of a dataset increases, our method also scales very well for different thresholds. 
For example, on the \sdatat dataset for \jacdist ($\delta$ is 0.25), the time costs of join with 25\%, 50\%, 75\%, 100\% dataset are 201, 400, 565, 835 seconds respectively.

\subsection{More Experimental Results}\label{subsec-moreexp}
\begin{figure*}
	\begin{center}
		\subfigure[\small{\vdatao}]{
			\label{subfig:cand-eculid}
			\hspace*{-1.0em}\epsfig{figure=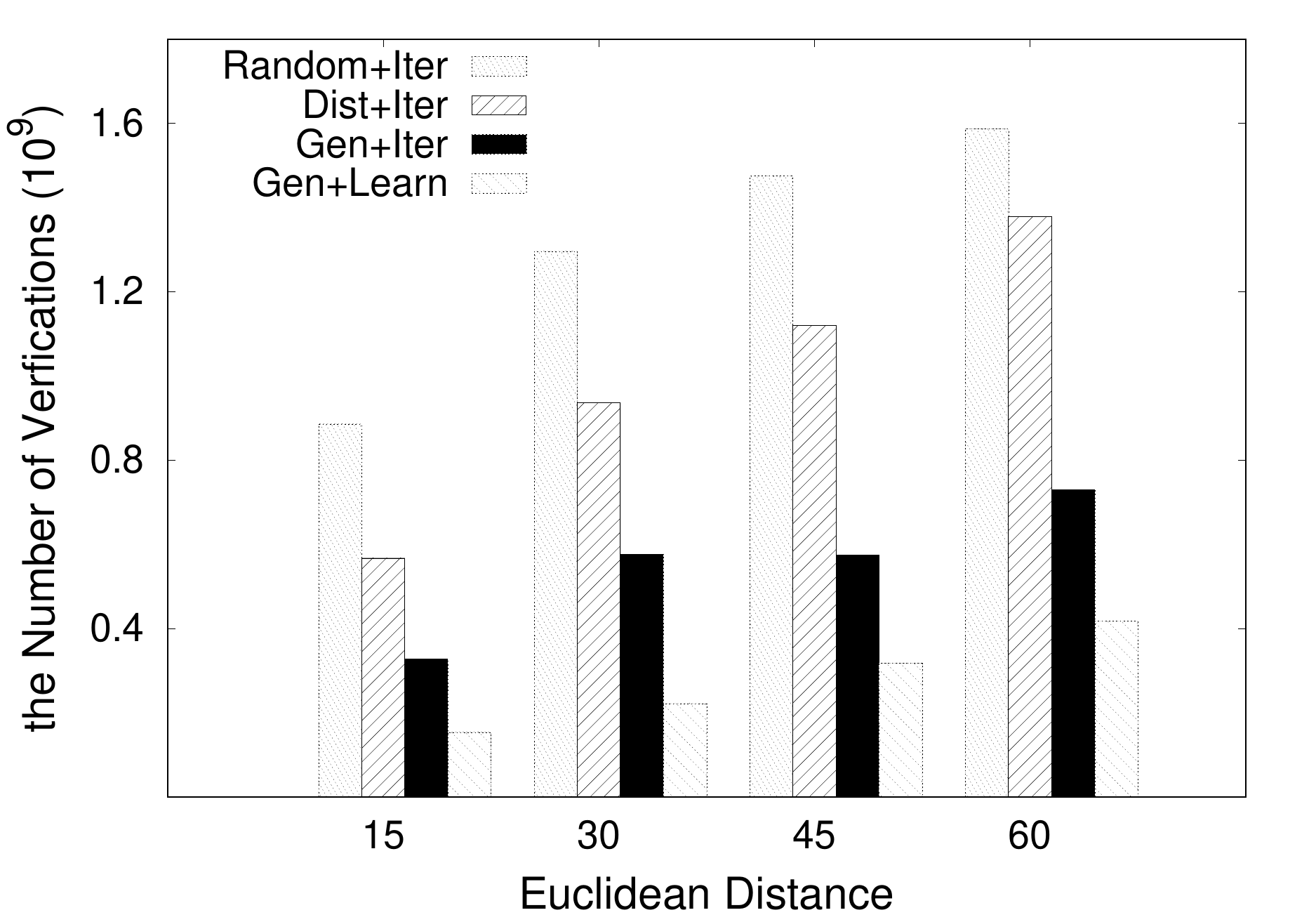,width=0.25\textwidth}}
		\subfigure[\small{\vdatat}]{
			\label{subfig:cand-onenorm}
			\hspace*{-1.0em}\epsfig{figure=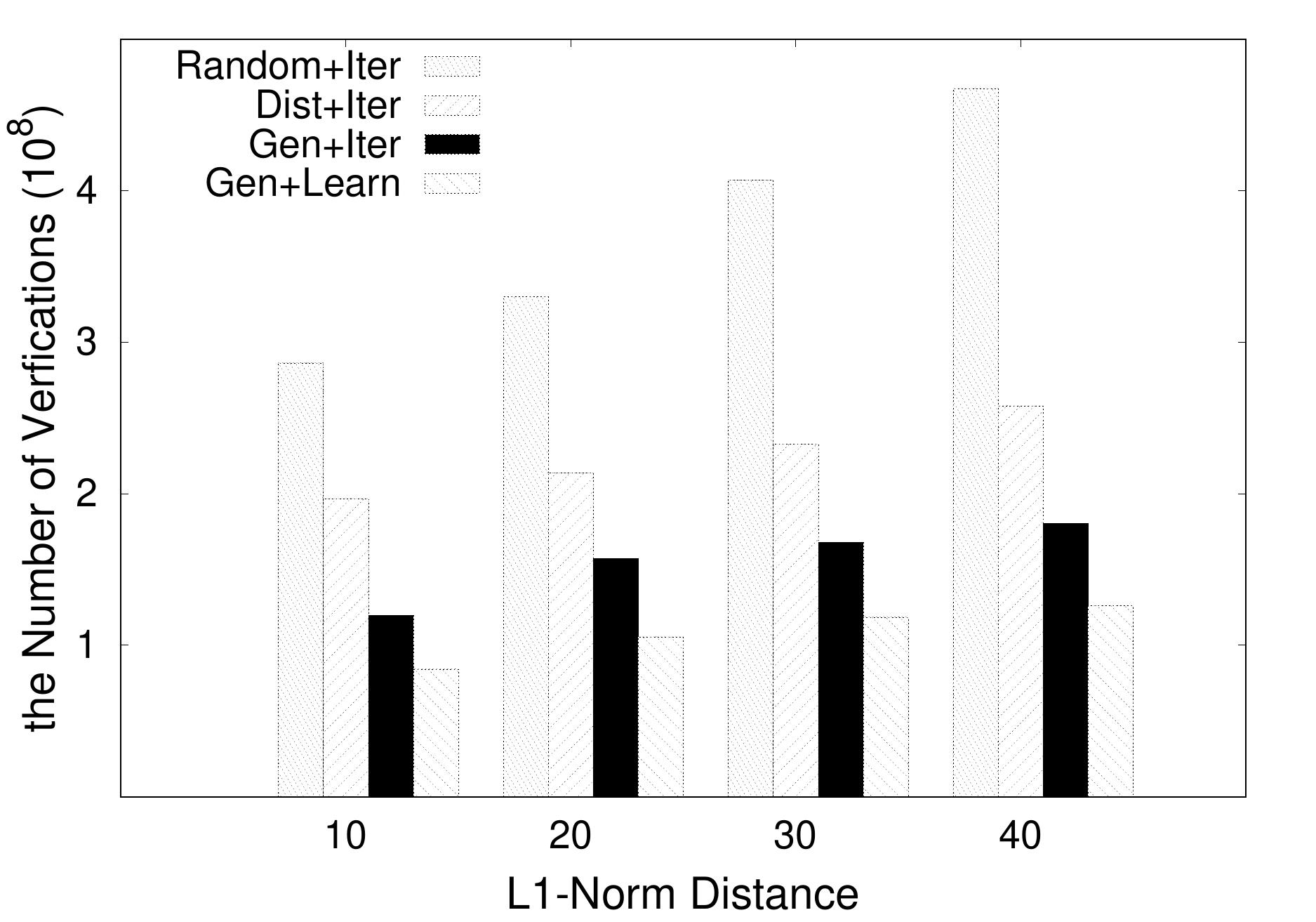,width=0.25\textwidth}}
		\subfigure[\small{\sdatao}]{
			\label{subfig:cand-edit}
			\hspace*{-1.0em}\epsfig{figure=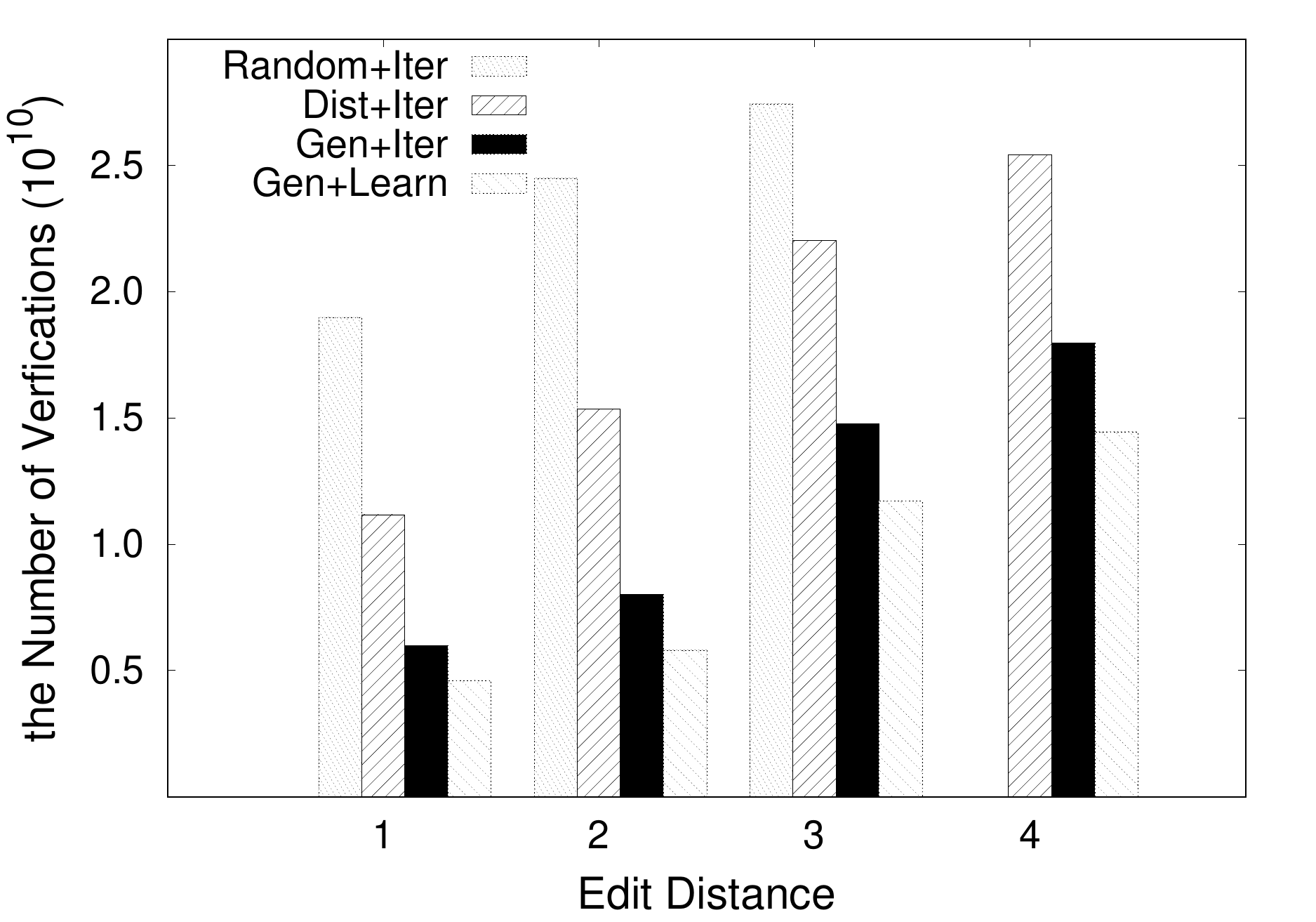,width=0.25\textwidth}}
		\subfigure[\small{\sdatat}]{
			\label{subfig:cand-jaccard}
			\hspace*{-1.0em}\epsfig{figure=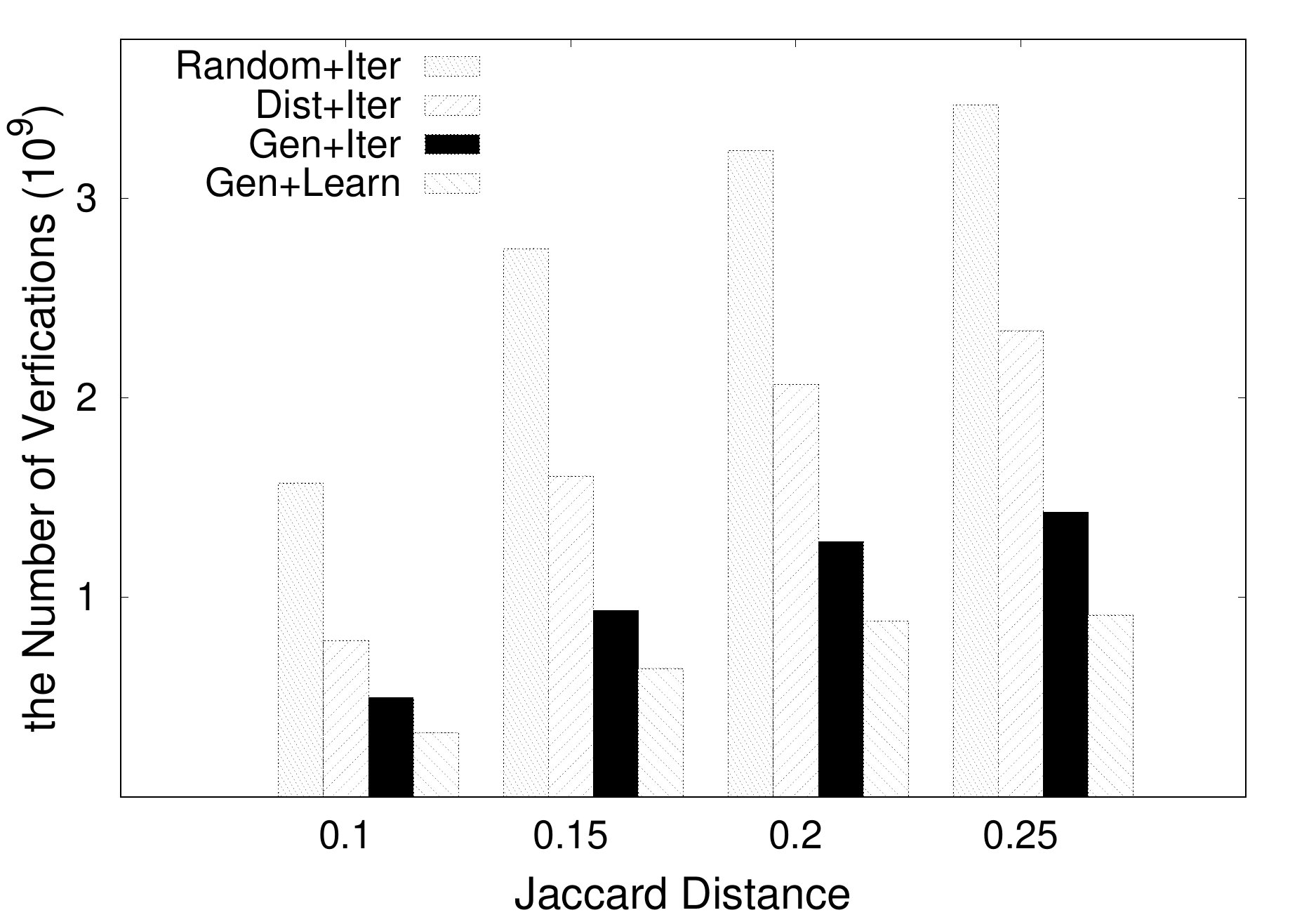,width=0.25\textwidth}}
	\end{center}\vspace{-1em}\vspace{-.5em}
	\caption{Effect of Proposed Techniques:  Numbers of Verification}\label{fig:scal:cand}
\end{figure*}	
\begin{table*}[h!t]
	\centering
	\caption{Statistics of Partitions}
	\label{tbl:partitionstat}
	\begin{tabular}{|c|c|c|c|c|c|c|c|c|}
		\hline
		&\multicolumn{2}{c|}{\lonenorm (SIFT)}&\multicolumn{2}{c|}{\eu (Netflix)}&\multicolumn{2}{c|}{\ed (AOL)}&\multicolumn{2}{c|}{\jac (PubMed)}\\
		\hline
		&AVER&STDEV&AVER&STDEV&AVER&STDEV&AVER&STDEV\\
		\hline
		\kdtree      
		&2.34E+08&4.02E+07&7.44E+08&2.20E+08&1.43E+10&8.90E+09&1.16E+09&4.24E+08\\
		\hline
		\randm\unskip+\basictr   &2.86E+08&4.51E+07&8.85E+08&2.18E+08&1.89E+10&9.17E+09&1.57E+09&4.78E+08\\
		\hline
		\dstaw\unskip+\basictr
		&1.97E+08&3.87E+07&5.67E+08&1.13E+08&1.12E+10&5.79E+09&7.83E+08&2.14E+08\\
		\hline
		\gener\unskip+\basictr 
		&1.20E+08&3.68E+07&3.28E+08&6.69E+07&5.99E+09&4.13E+09&4.98E+08&1.72E+08\\
		\hline
		\gener\unskip+\entropy     
		&8.42E+07&2.57E+07&1.53E+08&6.16E+07&4.60E+09&2.78E+09&3.22E+08&1.08E+08\\
		\hline
	\end{tabular}
\end{table*}
We also evaluate the number of verifications to further demonstrate the effect of proposed techniques.
The results are shown in Figure~\ref{fig:scal:cand}. 
We have the following observations: Firstly, \randm\unskip+\basictr involves the most number of verifications. 
This is because for random sampling, there tends to be skewness in the sampled pivots, resulting in the uneven division for \innerpart. 
Also, the size of \outerpart in the map phase will be much larger due to the shortcomings of \basictr mentioned in Section~\ref{subsec-advmap}.  
Secondly, with the same sampling technique \gener, \entropy involves fewer verifications than \basictr.
Thirdly, \gener\unskip+\entropy has the least number of verification due to the effectiveness of both techniques for sampling and map phase. 
Finally, we can see that the results of number of verifications in Figure~\ref{fig:scal:cand} is consistent with that of the join time in Figure~\ref{fig:samplemr}. 
It shows that our proposed techniques can significantly reduce the computational overhead and make a proper trade-off between the filter cost and filter power.

We then collect the information of average number and the standard deviation about the number of verifications in all partitions after the map phase to show the effect of load balancing. 
The results are shown in Table~\ref{tbl:partitionstat}. 
Due to the limitation of space, we only report the results on one threshold for each dataset(15, 10, 1, 0.1 respectively). 
We can see that \name has fewer average number and standard deviation. 
It demonstrates that our proposed techniques can make the partitions more evenly so that improve the overall performance.

%% file: src/sec8-related.tex
\vspace{-1em}
\section{Related Work}\label{sec-related}

 \noindent \textbf{MapReduce}\hspace{.5em}  MapReduce~\cite{DBLP:journals/cacm/DeanG08} is a popular distributed computing framework for Big Data processing in large clusters due to its scalability and built-in support for fault tolerance. 
 It has been a widely applied technique for elastic data processing in the cloud. 
 A good survey on the applications of MapReduce in dealing with large scale of analytical queries is provided in~\cite{DBLP:journals/vldb/DoulkeridisN14}. 
 There are many distributed computing platforms supporting MapReduce.
 Jiang et al.~\cite{DBLP:journals/pvldb/JiangOSW10} performed a performance study of MapReduce on the Hadoop platform. 
 Shi et al.~\cite{DBLP:journals/pvldb/ShiQMJWRO15} further made a comprehensive performance comparison between different platforms. 
 The MapReduce programming framework has also been adopted to other applications, such as skyline query~\cite{DBLP:journals/pvldb/ParkMS15} and diversity maximization~\cite{DBLP:journals/pvldb/CeccarelloPPU17}.\\

\noindent \textbf{Similarity Join}\hspace{.5em}  Similarity join is an essential operator in many real world applications. 
Jacox et al.~\cite{DBLP:journals/tods/JacoxS08} provided an overview of Similarity Join in metric space.  
A recent experimental survey is made in~\cite{DBLP:journals/pvldb/ChenGZJYY17} to cover the issues about indexing techniques in metric space, which plays a significant role in accelerating metric query processing. 
Many previous studies focused on designing efficient filter techniques for string similarity join: count filter~\cite{DBLP:conf/vldb/GravanoIJKMS01} and  prefix filter~\cite{DBLP:conf/icde/ChaudhuriGK06}are two state-of-the-art approaches. 
There are also specific filtering techniques proposed for token-based similarity metrics~\cite{DBLP:conf/www/XiaoWLY08}  and edit distance~\cite{DBLP:journals/vldb/YuWLZDF17}. 
And such techniques have been widely adopted by other problems related to string similarity queries~\cite{DBLP:conf/icde/WangLDZF15,DBLP:conf/edbt/WangLLZ19}.\\

\noindent \textbf{Similarity Join using MapReduce framework}\hspace{.5em} Recently, similarity join using MapReduce has attracted significant attraction. 
There are many MapReduce based studies for similarity join in metric space: Okcan et al.~\cite{DBLP:conf/sigmod/OkcanR11} proposed \thetajoin algorithm to handle similarity join with various types of predicates. 
Lu et al.~\cite{DBLP:journals/pvldb/LuSCO12} and Kim et al.~\cite{DBLP:conf/icde/KimS12} focused on the kNN Join problem. 
Wang et al.~\cite{DBLP:conf/kdd/WangMP13} developed the \mpass framework by leveraging the distance filtering and data compression techniques to reduce network communication. 
Fries et al.~\cite{DBLP:conf/icde/FriesBSS14} focused on similarity join of high dimensional data. 
Sarma et al.~\cite{DBLP:journals/pvldb/SarmaHC14} proposed \clusterjoin framework by applying the double pivots filter in the map phase and the 2D hash technique in the preprocessing phase. 
Chen et al.~\cite{DBLP:journals/tkde/ChenYCGZC17} adopted techniques in the field of spatial database, such as Space Filling Curve and KD-Tree indexing to achieve the goal of load balance.

There have also been many MapReduce based frameworks for similarity join on a specific data type, such as string and set.
They cannot be extended to the general distance functions in metric space as above ones.
A comprehensive experimental study is made in~\cite{DBLP:journals/pvldb/FierABLF18}. 
Vernica et al.~\cite{DBLP:conf/sigmod/VernicaCL10} adopted the prefix filter in the MapReduce framework to enhance the filter power of parallel similarity join. 
Metwally et al.~\cite{DBLP:journals/pvldb/MetwallyF12} proposed \vsmjoin, a method that ``smartly'' computes similarity scores at a token level. 
Afrati et al.~\cite{DBLP:conf/icde/AfratiSMPU12} made a detailed cost analysis on each stage of MapReduce similarity join on string data. 
Deng et al.~\cite{DBLP:conf/icde/DengLHWF14}\ and Rong et al.~\cite{DBLP:conf/icde/RongLSWLD17} further integrated more filtering techniques of string similarity query to improve the performance.

%% file: src/sec9-conclusion.tex
\vspace{-0.7em}
\section{Conclusion}\label{sec-con}

In this paper, we propose \name, a general MapReduce-based framework to support similarity join in metric space.
We design novel sampling techniques with theoretical guarantee to select high quality pivots for partitioning the dataset. 
We devise an iterative partition method along with learning techniques to ensure load balancing and improve the effectiveness in map and reduce phase to prune dissimilar pairs. 
Experimental results on four real world datasets show that our method significantly outperforms state-of-the-art methods with a variety of distance functions on different types of data.